\documentclass[english]{article}
\usepackage[T1]{fontenc}
\usepackage[latin9]{inputenc}
\usepackage{geometry}
\usepackage[active]{srcltx}
\usepackage{amsthm}
\usepackage{amsmath}
\usepackage{amssymb}
\usepackage{esint}

\makeatletter
\theoremstyle{plain}
\newtheorem{thm}{\protect\theoremname}[section]
  \theoremstyle{definition}
  \newtheorem{defn}[thm]{\protect\definitionname}
  \theoremstyle{plain}
  \newtheorem{lem}[thm]{\protect\lemmaname}
  \theoremstyle{plain}
  \newtheorem{prop}[thm]{\protect\propositionname}

\usepackage{feynmf}
\usepackage{subfig}

\makeatother

\usepackage{babel}
  \providecommand{\definitionname}{Definition}
  \providecommand{\lemmaname}{Lemma}
  \providecommand{\propositionname}{Proposition}
\providecommand{\theoremname}{Theorem}

\begin{document}

\title{\date{26.05.2014}Review of quantum gravity}

\author{Autor: Benjamin Schulz%
\thanks{Colmdorfstr. 32, 81249 München, e-mail: Benjamin.Schulz@physik.uni-muenchen.de%
} }

\maketitle
This document is merely a review article on quantum gravity. It is
organized as follows: in the section 1, it is argued why one should
construct a quantum theory of gravity. The importance of the singularity
theorems of general relativity is discussed, and the famous Penrose
singularity theorem is proven. In section 2.1, the covariant quantization
approach of gravity is reviewed and the Feynman rules of quantum gravity
are derived. The problem of divergences that occur at two loop order
is mentioned in section 2.2. In section 2.3 some comments are made
on the non perturbative evaluation of the quantum gravitational path
integral in the framework of Euclidean quantum gravity. In section
3, an article by Hawking is reviewed that shows the gravitational
path integral at one loop level to be dominated by contributions from
some kind of virtual gravitational instantons. In section 4, the canonical,
non-perturbative quantization approach that is based on the Wheeler
deWitt equation is reviewed. After deriving the Wheeler deWitt equation
in section 4.1, arguments from deWitt are described in section 4.2
which show occurrence of infinities at small distances within the
framework of canonical quantum gravity. In section 4.3, the loop quantum
gravity approach is reviewed shortly and its problems are mentioned.
In section 5, arguments from Hawking are reviewed which show the gravitational
path integral to be an approximate solution of the Wheeler deWitt
equation. In section 6, the black hole entropy is derived in various
ways. Section 6.1 uses the gravitational path integral for this calculation
and section 6.2 reviews how the black hole entropy can be derived
from canonical quantum gravity. In section 7.1, arguments from Dvali
and Gomez who claim that gravity can be quantized in a way which would
be in some sense ``self-complete'' are critically assessed. In section
7.2 a model from Dvali and Gomez for the description of quantum mechanical
black holes is critically assessed and compared with the standard
quantization methods of gravity.\tableofcontents{}

\section{Why one should construct a theory of quantum gravity?}

A classical particle with rest mass m becomes a black hole if its
entire mass is confined within its Schwarzschild radius $r{}_{s}=\frac{2Gm}{c^{2}}$
with G as the gravitational constant, and c as the speed of light
in vacuum. Quantum effects usually begin at the Compton wavelength
$\lambda=\frac{h}{mc}$, where h is Planck's constant and m is the
particle's rest mass. So, for the surrounding of a black hole, quantum
effects would become important if the black hole has a mass of $m=\sqrt{\frac{hc}{2G}}$.
This is an energy range of around $4.31\cdot10^{18}GeV/c^{2}$, which
is far beyond the energy of current particle accelerators. Therefore,
one might question whether it is reasonable, to do research on quantum
gravity. In the following, we will set $h=c=G=1$ if not explicitly
stated otherwise.

Well known classical solutions of general relativity, are e.g the
Schwarzschild solution, which describes static spherical black holes
of mass $M$ 
\begin{equation}
ds^{2}=-\left(1-\frac{2M}{r}\right)dt+\left(1-\frac{2M}{r}\right)^{-1}dr^{2}+r^{2}d\Omega,\label{eq:schwarzschild}
\end{equation}
 ($d\Omega^{2}=d\theta^{2}+sin^{2}\theta d\phi^{2}$ is the metric
of a unit two sphere in spherical coordinates). Physically more realistic
than the Schwarzschild solution is the Kerr solution which describes
rotating black holes 
\begin{eqnarray*}
ds^{2} & = & -\left(1-\frac{2Mr}{\rho}\right)dt^{2}-\frac{2Mar\sin^{2}\theta}{\rho^{2}}(dtd\phi+d\phi dt)\\
 &  & +\frac{\rho^{2}}{\Delta}dr^{2}+\rho^{2}d\theta^{2}+\frac{\sin^{2}\theta}{\rho^{2}}(r^{2}+a^{2})^{2}-a^{2}\Delta\sin^{2}\theta d\phi,
\end{eqnarray*}
 ($\Delta(r)=r^{2}-2Mr+a^{2}$ and $\rho(r,\theta)=r^{2}+acos^{2}\theta$,
with $a$ as some constant). And there is the important Friedmann
Robertson Walker metric hat describes the evolution of a spatially
homogeneous and isotropic universe 
\[
ds^{2}=-dt^{2}+R(t)d\sigma,
\]
(where t is the timelike coordinate and $d\sigma^{2}=\gamma_{ij}(u)du^{i}du^{j}$
is the line element of a maximally symmetric three manifold $\Sigma$
with $u^{1},u^{2},u^{3}$ as coordinates and $\gamma_{ij}$ as a symmetric
three dimensional metric). 

Unfortunately, in all these solutions, singularities are present where
the curvature becomes infinite. More precisely, with the Riemannian
tensor $R_{\mu\nu\beta}^{\lambda}$ defined by 
\begin{equation}
R_{\;\sigma\mu\nu}^{\rho}=\partial_{\mu}\Gamma_{\nu\sigma}^{\rho}-\partial_{\nu}\Gamma_{\mu\sigma}^{\rho}+\Gamma_{\mu\lambda}^{\rho}\Gamma_{\nu\sigma}^{\lambda}-\Gamma_{\nu\lambda}^{\rho}\Gamma_{\mu\sigma}^{\lambda},\label{eq:riemanntensor}
\end{equation}
where 
\begin{equation}
\Gamma_{\mu\nu}^{\sigma}=\frac{1}{2}g^{\sigma\rho}(\partial_{\mu}g_{\nu\rho}+\partial_{\nu}g_{\rho\mu}-\partial_{\rho}g_{\mu\nu})\label{eq:christoffel}
\end{equation}
is the Christoffel connection on the spacetime with metric tensor
$g_{\mu\nu}$, the solutions above have singularities in the sense
that at some point scalars like $R=g^{\mu\nu}R_{\mu\nu}$, $R^{\mu\nu}R_{\mu\nu}$,
$R^{\mu\nu\rho\sigma}R_{\mu\nu\rho\sigma}$,  or $R_{\mu\nu\rho\sigma}R^{\rho\sigma\lambda\tau}R_{\lambda\tau}^{\;\;\mu\nu}$
become infinite, where $R_{\mu\nu}=g^{\alpha\beta}R_{\alpha\mu\nu\beta}$
with $R_{\alpha\mu\nu\beta}=g_{\alpha\lambda}R_{\mu\nu\beta}^{\lambda}$ 

In the early years physicists were skeptic that these solutions given
above should have physical significance, due to the singularities
they contain. However, beginning in 1965, Penrose and Hawking have
shown a variety of singularity theorems. These theorems imply that
under reasonable conditions, the manifolds described by classical
general relativity must contain singularities. In the following section,
we review the derivation of the Penrose singularity theorem. Then,
we argue that the singularity theorems make the development of a theory
of quantum gravity necessary, since the singularities inside black
holes would render the physics of infalling matter inconsistent.

\subsection{The occurrence of singularities in gravity}

In this section, we give a short review of the first of the singularity
theorem that were published by Penrose \cite{penrose1}. The information
found in this section is merely a short summary of the excellent descriptions
in \cite{penrose,Carroll,Wald,Hawking}, with most of the proofs changed
only slightly. The focus of the text in this section is on topology,
which has seemingly beautiful applications on relativity. In order
to keep this section short, some important proofs from differential
geometry have been omitted. When necessary, the reader is pointed
out to appropriate references for these geometric results. 

We begin by noting some basic definitions and theorems from point
set topology

\subsubsection{Basic definitions and theorems of point set topology}
\begin{defn}
Basic definitions of point set topology 
\begin{itemize}
\item A topological space $(X,\mathcal{T})$ is a set $X\neq\emptyset$
with a collection of subsets $\mathcal{T}\subset\mathcal{P}(X):=\left\{ Y\subset X\right\} ,$
where
\begin{equation}
\begin{cases}
X,\emptyset\in\mathcal{T},\\
\forall U_{\lambda}\in\mathcal{T}\Rightarrow\cup_{\lambda\in\Lambda}U_{\lambda}\in\mathcal{T}\text{, with }\lambda\in\Lambda\text{ and }\Lambda\text{ as arbitrary index set,}\\
\forall U_{1},\ldots,U_{n}\in\mathcal{T}\Rightarrow\cap_{i=1}^{n}U_{i}\in\mathcal{T}\text{ with }n\in\mathbb{N}.
\end{cases}
\end{equation}
 
\item The sets $V\in\mathcal{T}$ are called open sets. 
\item Let $x\in X$ and $A\subset X$. Then $A$ is called a neighborhood
of $x$ if $\exists V\in\mathcal{T}:x\in V\subset A$ and $x$ is
then called an inner point of the neighborhood $A$. 
\item A subset $C\subset X$ is called closed, if the complement $C^{c}\equiv X\backslash C\equiv\left\{ x\in X:x\notin C\right\} $
is open, i.e. $C^{c}\in\mathcal{T}$.
\item A collection $\{U_{\lambda}\}$ of sets $U_{\lambda}\in\mathcal{T}$,
where$\lambda\in\Lambda$, and $\Lambda$ is an arbitrary index set,
is called an open cover of a set $A\subset X$, if 
\begin{equation}
A\subset\left(\cup_{\lambda\in\Lambda}U_{\lambda}\right)
\end{equation}

\item A subcover $\{U_{\alpha}\}$ of $A$ is a subset $\{U_{\alpha}\}\subset\{U_{\lambda}\}$,
where $\alpha\in\Theta,\Theta\subset\Lambda$, which fulfills 
\begin{equation}
A\subset\left(\cup_{\alpha\in\Theta}U_{\alpha}\right)
\end{equation}

\item The set $A$ is called compact if every open sub cover of $A$ has
a finite subcover. 
\item Let $\{U_{\lambda}\}$ be an open cover of $X$. An open cover $\left\{ V_{\beta}\right\} $
with $\beta\in\Psi$, where $\Psi$is an arbitrary index set, is called
a refinement of $\{U_{\lambda}\}$ if 
\begin{equation}
\forall\beta\exists\lambda:V_{\beta}\subset U_{\lambda}.
\end{equation}

\item an open cover$\left\{ U_{\beta}\right\} $ is locally finite if $\forall x\in X$
there exists an open neighborhood $U(x)$ such that the set 
\begin{equation}
\left\{ \beta\in\Psi:U_{\beta}\cap U(x)\neq\emptyset\right\} 
\end{equation}
is finite. 
\item A topological space $(X,\mathcal{T})$ is called paracompact if every
open cover $\{O_{\lambda}\}$ of $X$ has a locally finite refinement
$\left\{ V_{\beta}\right\} $. 
\item $(X,\mathcal{T})$ is Hausdorff if for all $x,y\in X,x\neq y$ there
exist neighborhoods $U$ of $x$ and $V$ of $y$ such that $U\cap V=\emptyset$.
In the following, we will consider only Hausdorff spaces.
\item $(X,\mathcal{T})$ is regular, if for each pair of a closed set $A$
and a point $p\notin A$, there exists neighborhoods $U,V$, where
$A\subset U$ and $p\in V$ and $U\cap V=\emptyset$.
\item Let $(X,\mathcal{T})$ be a topological space which is Hausdorff.
A sequence of points $\left\{ x_{n}\right\} \in X$ is said to converge
to a point $x$ if given any open neighborhood $O$ of $x,$
\begin{equation}
\exists n:x_{n}\in O\forall n\in\mathbb{N}.
\end{equation}

\item The point $x$ is then the limit point of $\left\{ x_{n}\right\} $.
\item A subset $Y\subset X$ is called dense in $X$ if for every point
$x\in X,$ $x$ is either in $Y$ or is a limit point of $Y$. 
\item $(X,\mathcal{T})$ is separable if $X$ contains a countable dense
subset.
\item Let $Y\subset X$ and $Y'$ be the set of all limit points of $Y$.
Then, the closure of $Y$ is $\overline{Y}=Y\cup Y'$.
\item A point $y\in X$ is an accumulation point of $\left\{ x_{n}\right\} $if
every open neighborhood of $y$ contains infinitely many points of
$\left\{ x_{n}\right\} $.
\item A topological space $(X,\mathcal{T})$ is first countable if for each
point $x\in X$
\begin{equation}
\exists U_{1},\ldots,U_{n}\in\mathcal{T},n\in\mathbb{N}
\end{equation}
 such that for any open neighborhood $V$ of $x$, there $\exists i:U_{i}\subset V$.
\item $(X,\mathcal{T})$ is second countable if there exists a countable
collection $\left\{ U_{i}\right\} _{i=1}^{\infty}$of open subsets
of $\mathcal{T}$ such that every open set $A\in\mathcal{T}$ can
be expressed as 
\begin{equation}
A=\cup_{i=1}^{\infty}U_{i}
\end{equation}

\item Let $(M,g)$ be a metric space with metric $g$. A metric topology
is the collection of sets that can be realized as union of open balls
\begin{equation}
B(x_{0},\epsilon)\equiv\left\{ x\in X|g(x_{0},x)<\epsilon\right\} \label{eq:ball}
\end{equation}
 where $x_{0}\in X$ and $\epsilon>0$,  
\item A topological space $(X,\mathcal{T})$ is called metrizable, if there
is a metric $g:(X,X)\rightarrow[0,\infty)$ on $X$ such that the
metric topology of $(X,d)$ equals $\mathcal{T}$.
\item A topological space $(X,\mathcal{T})$ is called connected, if the
only subsets that are both open and closed are the sets $X$ and $\emptyset$.
\item Let $(X,\mathcal{T})$ be a topological space. A function $h:V\subset X\rightarrow U\subset\mathbb{R}^{4}$
is called $4$ dimensional chart if $h$ is a homeomorphism of the
open set$V\subset X$ to an open set $U\subset R^{4}$. 
\item A family of cards $A=\left\{ h_{\alpha}:X_{\alpha}\rightarrow U_{\alpha}\right\} _{\alpha\in\Lambda}$,
 where $\Lambda$ is an arbitrary index set and $\cup_{\alpha}X_{\alpha}=X,$
is called an atlas of $X$. 
\item An atlas $A$ is called differentiable if $\forall\alpha,\beta\in\Lambda\times\Lambda:h_{\beta}\circ(h_{\alpha}|X_{\alpha}\cap X_{\beta})^{-1}$
is a diffeomorphism. 
\item Two differentiable atlases $A$ and $B$ are equivalent, if $A\cup B$
is a differentiable atlas. 
\item A 4-dimensional differentiable manifold is a set $X$ of a topological
space which is connected, Hausdorff and second countable, together
with an equivalence class of differentiable atlases that are homeomorphisms
to $\mathbb{R}^{4}$. 
\end{itemize}
\end{defn}
One has the following famous theorems in point set topology, which
we will use often:
\begin{thm}
\label{thm:metrizable} Let $(X,\mathcal{T})$ be any topological
space which is metrizable. Then, if $(X,\mathcal{T})$ is separable,
$(X,\mathcal{T})$ is second countable.\end{thm}
\begin{proof}
See \cite{Uhrysohn} for a proof of this rather weak statement. This
follows as part of Urysohn's lemma but actually, the Urysohn lemma
is much more powerful. It shows that a topological space is separable
and metrizable if and only if it is regular, Hausdorff and second-countable.
We will use theorem when we show that $C(p,q)$ is compact. We need
the result only in the weak form above.\end{proof}
\begin{thm}
\label{thm:(Bolzano-Weierstrass)}(Bolzano-Weierstrass) Let $(X,\mathcal{T})$
be a topological space and $A\subset X$. If A is compact, then every
sequence of points $\left\{ x_{n}\right\} \in A$ has an accumulation
point in $A$. Conversely, if $(X,\mathcal{T})$ is second countable,
and every sequence has an accumulation point in $A$ then $A$ is
compact.\end{thm}
\begin{proof}
Stated in \cite{Wald} on p. 426, proven in almost every analysis
script one can get.\end{proof}
\begin{thm}
(Heine-Borel) A subset of $\mathbb{R}^{n},n\in\mathbb{N}$ is compact
if and only if it is bounded and closed.\end{thm}
\begin{proof}
Stated in \cite{Wald} on p. 425, proven in almost every analysis
script one can get. 
\end{proof}
One can use the Heine-Borel theorem for proving:
\begin{thm}
\label{thm:attainsmaximum} (Extreme value theorem) A continuous function
$f:X\rightarrow\mathbb{R},$ from a compact set of a topological space
is bounded and attains its maximum and minimum values.\end{thm}
\begin{proof}
Stated in \cite{Wald} on p. 425. Proven in almost every analysis
script one can get.\end{proof}
\begin{thm}
\label{thm:compact}Let $(X,\mathcal{T})$be a topological space and
$A_{i}\subset X,i\in\Lambda,$with $\Lambda$ as an arbitrary index
set. If $A_{i}$ are compact, then $\cup A_{i}$ is is compact.\end{thm}
\begin{proof}
Let $U$ be an open cover of $\cup A_{i}$. Since $U$ is an open
cover of each $A_{i}$,  there exists a finite subcover $U_{i}\subseteq U$
for each $A_{i}$. Then $\cup U_{i}\subset U$ is a finite subcover
of $\cup A_{i}$. \end{proof}
\begin{thm}
\label{riemannian}Let $M$ be a paracompact differentiable manifold.
Then, $M$ is metrizable with a positive definite Riemannian metric.\end{thm}
\begin{proof}
This is shown in most texts on Riemannian geometry. For example in
\cite{JuergenJost} on p. 25 as corollary 1.4.4. Note that the author
of \cite{JuergenJost} already defines the term manifolds such that
they must be paracompact.
\end{proof}

\subsubsection{The Raychaudhuri's equation and conjugate points on geodesics}

Most of the topology definitions and lemmas will be used later. For
now, we need some definitions for the spacetime manifold, and the
curves in that manifold:
\begin{defn}
Basic definitions from Lorentzian geometry:
\begin{itemize}
\item A pair $(M,g)$ is a Lorentzian manifold if $M$ is a paracompact,
$4$ dimensional differentiable manifold and $g$ is a symmetric non-degenerate
$2$-tensor field on $M$, called metric, which has the signature
$(-,+,+,+)$ .
\item Let $(M,g)$ be a Lorentzian manifold and $p\in M$ . A vector $v\in T_{p}M$
is said to be timelike, if $g_{\mu\nu}v^{a}v^{b}<0$,  spacelike,
if $g_{\mu\nu}v^{\mu}v^{\nu}>0$ and null if $g_{\mu\nu}v^{\mu}v^{\nu}=0$.
The length of v is 
\begin{equation}
|v|=\sqrt{|g_{\mu\nu}v^{\mu}v^{\nu}|},
\end{equation}

\item Two timelike tangent vectors $x,y\in T_{p}M$ define the same time
orientation, if 
\begin{equation}
g_{\mu\nu}x^{\mu}y^{\nu}>0.
\end{equation}

\item A Lorentzian manifold is stably causal if and only if there exists
a global time function $t:M\rightarrow\mathbb{R},$ where $\nabla^{a}t(p)$
is a timelike vectorfield $\forall p\in M$.
\item A timelike vector $v\in T_{p}M$ is future pointing if and only if
$g_{\mu\nu}\nabla^{\mu}t(p)v^{\nu}>0$ and past pointing if $g_{\mu\nu}\nabla^{\mu}t(p)v^{\nu}<0$
. 
\item A Lorentzian manifold is time ordered if a continuous designation
of future pointing and past pointing for timelike vectors can be made
over the entire manifold.
\item A differentiable curve $\gamma:I\subset R\mapsto M$ is said to be
timelike, spacelike or null, if its tangent vector $\dot{\gamma}=\frac{d\gamma^{\mu}(\zeta)}{d\zeta}$
is timelike, spacelike, or null for all $\lambda\in I$ . For a timelike
curve, the tangent vector can not be null. A differential curve is
a causal curve if $\dot{\gamma}$is either timelike or null.
\item If a curve $\gamma$ is a differentiable spacelike curve connecting
the points $p,q\in M$, it has a the length function 
\begin{equation}
l(\gamma)=\int_{a}^{b}(g_{\mu\nu}\dot{\gamma}^{\mu}\dot{\gamma}^{\nu})^{1/2}dt
\end{equation}
 where $\gamma(a)=p$,  $\gamma(b)=q$. If instead $\gamma$ is a
differentiable timelike curve, the length is defined by the so called
proper time 
\begin{equation}
\tau(\gamma)=\int_{a}^{b}(-g_{\mu\nu}\dot{\gamma}^{\mu}\dot{\gamma}^{\nu})^{1/2}dt
\end{equation}

\item A curve $\gamma$ is called a geodesic if and only if it fulfills
the geodesic equation 
\begin{equation}
\frac{d^{2}\gamma^{\sigma}}{d\lambda^{2}}+\Gamma_{\mu\nu}^{^{\sigma}}\frac{d\gamma^{\mu}d\gamma^{\nu}}{d\lambda d\lambda}=0
\end{equation}

\end{itemize}
\end{defn}
There is an important theorem on the uniqueness of geodesics around
small neighborhoods that we will use often:
\begin{thm}
Let $(M,g)$ be a time ordered Lorentzian manifold. Then, each point
$p_{0}\in M$ has an open neighborhood $V\subset M$ called geodesically
convex, such that the spacetime $(V,g)$ obtained by restricting $g$
to V satisfies: if $p,q\in V$ then there exists a unique geodesic
joining $p$ to $q$, with the geodesic confined entirely to $V$.\end{thm}
\begin{proof}
Shown in\cite{Oneil} as proposition 31 on p. 72. 
\end{proof}
We now have to make some definitions that link the curves in a spacetime
manifold to the causality of events that happen in that spacetime
\begin{defn}
Causality
\begin{itemize}
\item A differentiable curve $\gamma$ is a future (past) directed causal
curve, if $\dot{\gamma}$ is non-spacelike and future (past) directed.
\item A continuous curve $\lambda$ is future directed, if for each $p\in\lambda$
there exists a geodesically convex neighborhood $U$ such that if
$\lambda(\zeta_{1}),\lambda(\zeta_{2})\in U$ with $\zeta_{1}<\zeta_{2}$,
then there exists a future directed differentiable curve connecting
$\lambda(\zeta_{1})$ to $\lambda(\zeta_{2})$. 
\item The chronological future of a set $p\subset M$ is the set $I^{+}(p)$
of all points in $M$ that can be connected from $p$ by a future
directed timelike curve.
\item The causal future of a set $p\subset M$ is the set $J^{+}(p)$ of
all points in $M$ that can be connected from $p$ by a future directed
causal curve.
\item The chronological future of a set $p\subset M$ relative to $U\subset M$
is the set $I^{+}(p,U)\equiv U\cap I^{+}(p).$
\item The causal future of a set $p\subset M$ relative to $U\subset M$
is the union $J^{+}(p,U)\equiv\left(p\cap U\right)\cup(U\cap J^{+}(p))$
\end{itemize}
\end{defn}
In the following, we often use an important theorem on the lengths
of timelike geodesics in geodesically convex neighborhoods:
\begin{thm}
\label{twin} Let $(M,g)$ be a time oriented spacetime and $p_{0}\in M$.
Then there exists a geodesically convex open neighborhood $V\subset M$
of $p_{0}$ such that the spacetime $(V,g)$ obtained by restricting
$g$ to $V$ satisfies the following property: If$q\in I^{+}(p),$
$c$ is the unique timelike geodesic connecting $p$ to $q$ and $\gamma$
is any piecewise smooth timelike curve connecting $p$ to $q$ then
$\tau(\gamma)\leq\tau(c)$, with equality if and only if $\gamma$
is the unique geodesic.\end{thm}
\begin{proof}
Shown in \cite{Oneil} on p. 133-134 as Lemma 14, note also the refinement,
on p. 134, which shows that Lemma 14 holds also for piecewise smooth
curves.
\end{proof}
On a manifold, one can have entire families of geodesics. The following
provides some useful definitions for them.
\begin{defn}
Definitions for geodesic congruences

\label{Definitions2}
\begin{itemize}
\item Let $\gamma_{s}(t)$, $t,s\in\mathbb{R}$ denote a smooth one parameter
family of geodesics on a manifold $(g,M)$. The collection of these
geodesics defines a smooth two-dimensional surface. The coordinates
of this surface may be chosen to be s and t, provided that the geodesics
do not cross. Such a family of geodesics is called geodesic congruence. 
\item Denote the entire surface generated by the geodesic congruence as
the set of coordinates $x^{\mu}(t,s)\in M$. Then, the tangent vectors
to the geodesic congruence are defined by $T^{\mu}=\frac{\partial x^{\mu}}{\partial\tau}$,
and the geodesic deviation vectors are $S^{\mu}=\frac{\partial x^{\mu}}{\partial s}$. 
\item The vectorfield $v^{\mu}=T^{\nu}\nabla_{\nu}S^{\mu}$defines the relative
velocity of the geodesic congruence and 
\begin{equation}
a^{\mu}=T^{\lambda}\nabla_{\lambda}v^{\mu}=T^{\lambda}\nabla_{\lambda}(T^{\nu}\nabla_{\nu}S^{\mu})
\end{equation}
 is the relative acceleration 
\item For a tensor $T_{\mu_{1},\ldots\mu_{l}}$ we define the symmetric
part $T_{(\mu_{1},\ldots\mu_{l})}=\frac{1}{l!}\sum_{\sigma}T_{\mu_{\sigma(1),\ldots\sigma(l)}}$
where the sum goes over all permutations $\sigma$ of $1,\ldots l$
. The antisymmetric part $T_{[\mu_{1},\ldots\mu_{l}]}=\frac{1}{l!}\sum_{\sigma}sign(\sigma)T_{\mu_{\sigma(1),\ldots\sigma(l)}}.$where
$sign(\sigma)$ is 1 for even and $-1$ for odd permutations.
\item Let $T^{\mu}$ be the tangent vector field to a timelike geodesic
congruence parametrized by proper time $\tau$. We define the projection
tensor 
\begin{equation}
P_{\mu\nu}=g_{\mu\nu}+T_{\mu}T_{\nu}.
\end{equation}

\item Any vector in the tangent space of a point $p\in M$ can be projected
by $P_{\mu\nu}$ into the subspace of vectors normal to $T^{\mu}$.
We now can define a tensorfield 
\begin{equation}
B_{\mu\nu}=\nabla_{\nu}T_{\mu}
\end{equation}
where we call $\theta=B^{\mu\nu}P_{\mu\nu}$ the expansion
\item The shear is defined by the symmetric tensorfield 
\begin{equation}
\sigma_{\mu\nu}=B_{(\mu\nu)}-\frac{1}{3}\theta P_{\mu\nu}
\end{equation}

\item The twist is an antisymmetric tensorfield defined by $\omega_{\mu\nu}=B_{[\mu\nu]}$. 
\item Let $(M,g)$ be a manifold on which a Christoffel connection is defined
and let $\gamma$ be a geodesic with tangent $T$. A solution $S$
of the geodesic deviation equation 
\begin{equation}
T^{\lambda}\nabla_{\lambda}(T^{\nu}\nabla_{\nu}S^{\mu})=R_{\;\nu\rho\sigma}^{\mu}T^{\nu}T^{\rho}S^{\sigma},
\end{equation}
is called Jacobi field. 
\end{itemize}
\end{defn}
We can use these definitions to derive equations that describe whether
the geodesics in a given congruence are converging to each other.
The following equations are described in most relativity books, e.g.
\cite{Carroll,Wald} . They will be used during the proofs of the
singularity theorem.
\begin{lem}
Let the tangent vector of a geodesic congruence be parametrized by
proper time $t=\tau$. Then, the tangent field fulfills $T_{\mu}T^{\mu}=-1$
and $T^{\lambda}\nabla_{\lambda}T^{\mu}=0$. The geodesic acceleration
vector fulfills the geodesic deviation equation: 
\begin{equation}
a^{\mu}=R_{\;\nu\rho\sigma}^{\mu}S^{\nu}T^{\rho}T^{\sigma}
\end{equation}
\end{lem}
\begin{proof}
Because of 
\begin{equation}
d\tau^{2}=-g_{\mu\nu}dx^{\mu}dx^{\nu}
\end{equation}
and $T^{\mu}=\frac{\partial x^{\mu}}{\partial\tau}$, we have 
\begin{equation}
T_{\mu}T^{\mu}=g_{\mu\nu}T^{\mu}T^{\nu}=-1.
\end{equation}
The expression $T^{\lambda}\nabla_{\lambda}T^{\mu}=0$ follows from
the geodesic equation for $x^{\mu}$. For the relative acceleration
vector, we have: 
\begin{eqnarray}
a^{\mu} & = & T^{\rho}\nabla_{\rho}(T^{\sigma}\nabla_{\sigma}S^{\mu})\nonumber \\
 & = & T^{\rho}\nabla_{\rho}(S^{\sigma}\nabla_{\sigma}T^{\mu})\nonumber \\
 & = & (T^{\rho}\nabla_{\rho}S^{\sigma})(\nabla_{\sigma}T^{\mu})+T^{\rho}S^{\sigma}\nabla_{\rho}\nabla_{\sigma}T^{\mu}\nonumber \\
 & = & (S^{\rho}\nabla_{\rho}T^{\sigma})(\nabla_{\sigma}T^{\mu})+T^{\rho}S^{\sigma}\left(\nabla_{\sigma}\nabla_{\rho}T^{\mu}+R_{\;\nu\rho\sigma}^{\mu}T^{\nu}\right)\nonumber \\
 & = & (S^{\rho}\nabla_{\rho}T^{\sigma})(\nabla_{\sigma}T^{\mu})+T^{\sigma}\nabla_{\sigma}(T^{\rho}\nabla_{\rho}T^{\mu})-(S^{\sigma}\nabla_{\sigma}T^{\rho})\nabla_{\rho}T^{\mu}+R_{\;\nu\rho\sigma}^{\mu}T^{\nu}T^{\rho}S^{\sigma}\nonumber \\
 & = & R_{\;\nu\rho\sigma}^{\mu}T^{\nu}T^{\rho}S^{\sigma}
\end{eqnarray}
where in the second line it was used that 
\begin{equation}
T^{a}\nabla_{a}X^{b}=X^{a}\nabla_{a}T^{b},
\end{equation}
the third line uses Leibniz's rule, the fourth line replaces a double
covariant derivative by the derivative in the opposite order plus
the Riemann tensor. Finally, in the fifth line, $T^{\rho}\nabla_{\rho}T^{\mu}=0$
is used together with the Leibniz's rule for canceling two identical
terms.
\end{proof}
Furthermore, one can establish the following theorem: 
\begin{thm}
Given the definitions \ref{Definitions2} then 
\begin{equation}
B_{\mu\nu}T^{\mu}=B_{\mu\nu}T^{\nu}=\sigma_{\mu\nu}T^{\mu}=\sigma_{\mu\nu}T^{\nu}=\omega_{\mu\nu}T^{\mu}=\omega_{\mu\nu}T^{\nu}=0
\end{equation}
 and Raychaudhuri's equation 
\begin{equation}
\frac{d\theta}{d\tau}=-\frac{1}{3}\theta^{2}-\sigma_{\mu\nu}\sigma^{\mu\nu}+\omega_{\mu\nu}\omega^{\mu\nu}-R_{\mu\nu}T^{\mu}T^{\nu}\label{eq:raychaudhuri}
\end{equation}
 holds.\end{thm}
\begin{proof}
Because of 
\begin{equation}
\nabla_{\nu}(T^{\mu}T_{\mu})=\nabla_{\nu}(-1)=0,
\end{equation}
we get, using 
\begin{equation}
T^{\mu}B_{\mu\nu}=T^{\mu}\nabla_{\nu}T_{\mu}
\end{equation}
 the result $T^{\mu}B_{\mu\nu}=0$. With 
\begin{equation}
T^{\nu}B_{\mu\nu}=T^{\nu}\nabla_{\nu}T_{\mu}
\end{equation}
 and the geodesic equation $T^{\nu}\nabla_{\nu}T_{\mu}=0$, we arrive
at $T^{\nu}B_{\mu\nu}=0$. Since we have 
\begin{equation}
\sigma_{\mu\nu}=B_{(\mu\nu)}-\frac{1}{3}\theta P_{\mu\nu}
\end{equation}
 and 
\begin{equation}
\theta=B^{\mu\nu}P_{\mu\nu},
\end{equation}
it follows from 
\begin{equation}
B_{\mu\nu}T^{\mu}=B_{\mu\nu}T^{\nu}=0
\end{equation}
that 
\begin{equation}
\sigma_{\mu\nu}T^{\mu}=\sigma_{\mu\nu}T^{\nu}=0
\end{equation}
 and similarly $\omega_{\mu\nu}T^{\nu}=\omega_{\mu\nu}T^{\mu}=0$.
Furthermore 
\begin{eqnarray}
T^{\rho}\nabla_{\rho}B_{\mu\nu} & = & T^{\rho}\nabla_{\rho}\nabla_{\nu}T_{\mu}\nonumber \\
 & = & T^{\rho}\nabla_{\nu}\nabla_{\rho}T_{\mu}+T^{\rho}R_{\;\mu\nu\rho}^{\lambda}T_{\lambda}\label{eq:rhayhau}\\
 & = & \nabla_{\nu}(T^{\rho}\nabla_{\rho}T_{\mu})-(\nabla_{\nu}T^{\rho})(\nabla_{\rho}T_{\mu})-R_{\lambda\mu\nu\sigma}T^{\sigma}T^{\lambda}\nonumber \\
 & = & -B_{\;\nu}^{\sigma}B_{\mu\sigma}-R_{\lambda\mu\nu\sigma}T^{\sigma}T^{\lambda}
\end{eqnarray}
taking the trace of this equation yields: 
\begin{equation}
\frac{d\theta}{d\tau}=-\frac{1}{3}\theta^{2}-\sigma_{ab}\sigma^{ab}+\omega_{ab}\omega^{ab}-R_{cd}T^{c}T^{d}
\end{equation}
\end{proof}
\begin{thm}
Given the definitions \ref{Definitions2}, then $\sigma_{\mu\nu}\sigma^{\mu\nu}\geq0$,
$\omega_{\mu\nu}\omega^{\mu\nu}\geq0$ . Furthermore $\omega_{\mu\nu}=0$
if and only if $T^{\mu}$ is orthogonal to a family of hypersurfaces\end{thm}
\begin{proof}
$\sigma_{\mu\nu}$ is a spatial tensorfield that is $\sigma_{\mu\nu}T^{\nu}=0$,
where $T^{\nu}$ is timelike. Therefore, we have $\sigma_{\mu\nu}\sigma^{\mu\nu}\geq0$
and similarly $\omega_{\mu\nu}\omega^{\mu\nu}\geq0$ . By Frobenius'
theorem, see \cite{Wald} p. 434-436, a vector field $T^{\mu}$ is
hypersurface orthogonal, if and only if $T_{[\mu}\nabla_{\mu}T_{\lambda]}=0$.
With $B_{\mu\nu}=\nabla_{\nu}T_{\mu}$,  $\omega_{\mu\nu}=B_{[\mu\nu]}$
and $T^{\mu}\omega_{\mu\nu}=T^{\nu}\omega_{\mu\nu}=0$, it follows
that $\omega_{\mu\nu}=0$ if and only if $T^{\nu}$ is orthogonal
to a family of hypersurfaces.
\end{proof}
Having defined the equations for families of geodesics, we now give
some further definitions related to the causality structure of the
Lorentzian manifold under study.
\begin{defn}
Basic definitions for curvature and causality structures
\begin{itemize}
\item Let $(M,g)$ be a stably causal manifold and $\lambda$ a differentiable
future directed causal curve. We say $p$ is a future endpoint of
$\lambda$ if for every neighborhood $U$ of $p$ there exists a $t_{0}$
such that $\lambda(t)\in U\forall t>t_{0}$. A past endpoint is defined
similarly, with $\lambda(t)\in U\forall t<t_{0}$
\item The curve $\lambda$ is said to be future (past) inextendible if it
has no future (past) endpoint. 
\item A subset $\Sigma\subset M$ is called achronal, if there do not exist
$p,q\in\Sigma$ such that $q\in I^{+}(p)$ or: 
\begin{equation}
I^{+}(\Sigma)\cap\Sigma=\emptyset.
\end{equation}

\item The past (future) domain of dependence of the achronal set$\Sigma\subset M$
is the set $D^{-(+)}(\Sigma)$ of all points $p\in M$ such that any
future (past) inextendible curve starting at $p$ intersects $\Sigma$.
\item The domain of dependence is the set 
\begin{equation}
D(\Sigma)=D^{+}(\Sigma)\cup D^{-}(\Sigma)
\end{equation}

\item A spacetime $(M,g)$ is strongly causal if $\forall p\in M$ and every
open neighborhood $O$ of $p$ there is a neighborhood $V\subset O$
of p such that no causal curve intersects $V$ more than once . 
\item An achronal subset $\Sigma\subset M$ is called a Cauchy surface,
if $D(\Sigma)=M$.
\item A Lorentzian manifold is globally hyperbolic, if it is stably causal
and possesses a Cauchy surface $\Sigma$ 
\item Let $(M,g)$ be a globally hyperbolic manifold with a Cauchy surface
$\Sigma$, A point $p$ on a geodesic $\gamma$ of the geodesic congruence
orthogonal to $\Sigma$ is called a conjugate point to $\Sigma$ along
$\gamma$ if there exists a Jacobi-field $S^{\mu}$ which is non-vanishing
except on $p$. 
\item We define the extrinsic curvature $K_{ab}$of $\Sigma$ as 
\begin{equation}
K_{ab}=B_{ba}
\end{equation}
 and its trace as 
\begin{equation}
K=P^{ab}K_{ab}.
\end{equation}

\end{itemize}
\end{defn}
A standard proof from differential geometry lectures shows that a
timelike geodesic of maximal length that starts on a Cauchy surface
must be orthogonal to it:
\begin{thm}
\label{orthogonal}Let $(M,g)$ be a globally hyperbolic manifold
with a Cauchy surface $\Sigma$ and let $\gamma$ be a timelike geodesic
starting at $p\in\Sigma$ to a point $q\in I^{+}(\Sigma$). Then,
$\gamma$ only maximizes length if and only if it is orthogonal to
$\Sigma$ \end{thm}
\begin{proof}
Proof is given in \cite{Oneil} on p. 280.
\end{proof}
The following proof is taken from Wald's book \cite{Wald}. It relates
the conditions for stable and strong causality:
\begin{thm}
\label{stablystronglycausal} Let $(M,g)$ be a spacetime which is
stably causal. Then, it is also strongly causal.\end{thm}
\begin{proof}
Let $t$ be the global time function on $M$. Since $g_{\mu\nu}\dot{\gamma}^{\nu}\nabla^{\mu}t(p)>0$
for all $p\in\gamma$ where $\gamma$ is a future directed timelike
curve, $t$ strictly increases along every future directed timelike
curve. Given any $p\in M$ and any open neighborhood $O$ of $p$,
we can choose an open neighborhood $V\subset O$ shaped such that
the limiting value of $t$ along every future directed causal curve
leaving $V$ is greater than the limiting value of $t$ on every future
directed causal curve entering $V$. Since $t$ increases along every
future directed causal curve, no causal curve can enter $V$ twice.\end{proof}
\begin{thm}
\label{conjugate}Let $(M,g)$ be a globally hyperbolic spacetime
satisfying $R_{\mu\nu}\zeta^{\mu}\zeta^{\nu}\geq0$ for all timelike
$\zeta$, Let $\Sigma$ be a Cauchy surface with $K=\theta<0$. Then
within proper time $\tau\leq3/|K|$ there exists a conjugate point
$q$ along the timelike geodesic $\gamma$ orthogonal to $\Sigma$,assuming
that $\gamma$ extends that far. 
\end{thm}
The proof is adapted from Wald's book\cite{Wald}.where one can find
a proof for the situation where a point is conjugate to another point
on a geodesic. Only minor changes were necessary to adapt this proof
for the situation of a geodesic orthogonal to a hypersurface.
\begin{proof}
Let $\gamma$ be a timelike geodesic orthogonal to $\Sigma$ with
tangent $T^{\mu}=\frac{\partial\gamma^{\mu}}{\partial\tau}$. We consider
the congruence of all timelike geodesics passing through $\Sigma$
that are orthogonal to $\Sigma$. Since the deviation vector $S^{\mu}=\frac{\partial x^{\mu}}{\partial s}$
of this geodesic congruence is orthogonal to $\gamma$, we have $S^{\mu}T_{\mu}=0.$
Therefore, we can introduce an orthonormal basis $e_{a}^{\mu}(\tau)$,
$a=1,\ldots3$ which is orthogonal to $\gamma,$ i.e. where $e_{a}^{\mu}e_{b\mu}=\delta_{ab}$,
and $e_{a}^{\mu}T_{\mu}=0$ and is parallel transported along $\gamma,$
that is 
\begin{equation}
\frac{d\gamma^{\nu}}{d\tau}\nabla_{\nu}e_{a}^{\mu}:=\frac{D}{d\tau}e_{a}^{\mu}=0.
\end{equation}
holds. Inserting this with $S^{\mu}=(e_{\mu})_{a}S^{a}$ in the geodesic
deviation equation 
\begin{equation}
T^{\lambda}\nabla_{\lambda}(T^{\nu}\nabla_{\nu}S^{\mu})=R_{\;\nu\rho\sigma}^{\mu}T^{\nu}T^{\rho}S^{\sigma},
\end{equation}
the term on the left hand side simplifies to 
\begin{equation}
\frac{dS^{\mu}}{d\tau^{2}}=R_{\;\nu\rho\sigma}^{\mu}T^{\nu}T^{\rho}S^{\sigma}.
\end{equation}
Since this is an ordinary linear differential equation, $S^{\mu}$
must depend linearly on the initial $S^{\mu}(0)$ and $\frac{dS^{\mu}(0)}{d\tau}$
. Since we are searching for conjugate points, the congruence $S^{\mu}(0)$
should not be vanishing on $\Sigma$ but all geodesics of the congruence
should be hypersurface-orthogonal to $\Sigma$. Therefore, $\frac{d}{d\tau}S^{\mu}(0)=0$,
and noting that $S^{\mu}=(e_{\mu})_{a}S^{a}$with $a=1,\ldots3$,
we can make an ansatz, 
\begin{equation}
S^{\sigma}(\tau)=A_{\; a}^{\sigma}S^{a}(0).
\end{equation}
Plugging this in the coordinate form of the geodesic deviation equation
gives 
\begin{equation}
\frac{dA_{\; a}^{\mu}}{d\tau^{2}}=R_{\;\nu\rho\sigma}^{\mu}T^{\nu}T^{\rho}A_{\; a}^{\sigma}\label{eq:diffa}
\end{equation}
A point p that is conjugate to $\Sigma$ develops by definition if
and only if $S^{\mu}(\tau)=0$. With our ansatz, this condition becomes
$det(A_{\; a}^{\mu})=0$. Between $\Sigma$ and $p$, the Jacobi-field
should be non vanishing. Therefore $det(A_{\; a}^{\mu})\neq0$ for
this region and an inverse of $A_{\; a}^{\mu}$ exists. We have
\begin{eqnarray}
\frac{dS^{\nu}}{d\tau} & = & T^{\mu}\nabla_{\mu}S^{\nu}\nonumber \\
 & = & T^{\mu}\nabla_{\mu}((e_{\nu})_{a}S^{a})\nonumber \\
 & = & (e_{\nu})_{a}T^{\mu}\nabla_{\mu}S^{a}\nonumber \\
 & = & (e_{\nu})_{a}B_{\;\mu}^{a}S^{\mu}\nonumber \\
 & = & B_{\; a}^{\nu}S^{a}\nonumber \\
 & = & B_{\; a}^{\nu}A_{\; b}^{\sigma}S^{b}(0).
\end{eqnarray}
with $b=1,\ldots3$. From our ansatz, we get
\begin{eqnarray}
\frac{dS^{\nu}}{d\tau} & = & \frac{dA_{\; a}^{\nu}}{d\tau}S^{a}(0)\nonumber \\
 & = & B_{\; a}^{\nu}A_{\; b}^{\sigma}S^{b}(0)
\end{eqnarray}
 or in matrix notation $\frac{dA}{d\tau}=BA$ that is 
\begin{equation}
B=\frac{dA}{d\tau}A^{-1}
\end{equation}
Since $\theta=tr(B)$ and for any non-singular matrix $A$: 
\begin{equation}
tr\left(\frac{dA}{d\tau}A^{-1}\right)=\frac{1}{det(A)}\frac{d}{d\tau}(det(A))
\end{equation}
 it follows that 
\begin{equation}
\theta=\frac{1}{det(A)}\frac{d}{d\tau}(det(A)).
\end{equation}
Since $A$ satisfies the ordinary linear differential equation (\ref{eq:diffa}),
$\frac{d}{d\tau}(det(A)$ can not become infinite. Hence if and only
if $\theta=-\infty$ at a point $q$ on $\gamma$, then $A\rightarrow0$
and $det(A)\rightarrow0$ which implies that $q$ will be a conjugate
point. 

Our congruence is hypersurface orthogonal on $\Sigma$ by construction.
Therefore, $\omega_{bc}=0$ on $\Sigma$. The antisymmetric part of
eq. (\ref{eq:rhayhau}) is
\begin{equation}
T^{c}\nabla_{c}\omega_{ab}=-\frac{2}{3}\theta\omega_{ab}-2\sigma_{[b}^{c}\omega_{a]c}
\end{equation}
see \cite{Wald} p. 218, or \cite{Carroll} p. 461. Therefore, if
$\omega$ is orthogonal to a family of hypersurfaces at at one time,
it remains so, along the geodesic $\gamma$ which implies $\omega_{ab}=0$
for all the time. Since $\sigma_{ab}\sigma^{ab}\geq0$ and by proposition
\begin{equation}
R_{ab}T^{a}T^{b}\geq0,
\end{equation}
we have from Raychaudhuri's equation (\ref{eq:rhayhau}): 
\begin{equation}
\frac{d\theta}{d\tau}+\frac{1}{3}\theta^{2}\leq0
\end{equation}
and hence 
\begin{equation}
\theta^{-1}(\tau)\geq\theta_{0}^{-1}+\frac{1}{3}\tau,
\end{equation}
where $\theta_{0}$is the initial value of $\theta$ on $\Sigma$.
As a result, if $\theta_{0}\leq0$ on $\Sigma$, the expansion will
converge to $\theta\rightarrow-\infty$ at $\tau\leq3/|\theta_{0}|$,
and there will be a conjugate point $p$ on $\gamma$. 

The intrinsic curvature $K_{ab}$ is $K_{ab}=B_{ba}=B_{ab}$. Since
the antisymmetric part $B_{[ab]}=\omega_{ab}=0$ only the symmetric
part
\begin{equation}
\sigma_{ab}+\frac{1}{3}\theta P_{ab}=B_{(ab)}
\end{equation}
 is non-vanishing. With 
\begin{equation}
\theta=P^{ab}B_{ab},
\end{equation}
we get 
\begin{equation}
K=P^{ab}K_{ab}=\theta.
\end{equation}
Therefore, if $K<0$ on $\Sigma$ and 
\begin{equation}
R_{ab}T^{a}T^{b}\geq0,
\end{equation}
there will be a conjugate point to $\Sigma$ on $\gamma$ within proper
time $\tau\leq3/|\theta_{0}|$\end{proof}
\begin{thm}
\label{maximizelength}Let $(M,g)$ be a globally hyperbolic Lorentzian
manifold with a Cauchy surface $\Sigma$, and $\gamma$ a timelike
geodesic orthogonal to $\Sigma$, and $p$ a point on $\gamma$ .
If there exists a conjugate point between $\Sigma$ and $p$ then
$\gamma$ does not maximize length among the timelike curves connecting
to $\Sigma$ to $p$.\end{thm}
\begin{proof}
A heuristic argument would be the following: Let $q$ be a conjugate
point along $\gamma$ between $\Sigma$ and $p$. Then, there exists
another geodesic $\overline{\gamma}$ orthogonal to $S$ with the
same approximate length. This geodesic must intersect $\gamma$ at
$q$. Let $V$ be a geodesically convex neighborhood of $q$ and $r\in V$
a point along $\overline{\gamma}$ between $\Sigma$ and $q$ and
let $s\in V$ a point along $\gamma$ between $q$ and $p$. Then,
we get a piecewise smooth timelike curve $\tilde{\gamma}$ obtained
by following:1) $\overline{\gamma}$ between $\Sigma$ and $r$ then
2) the unique geodesic between $r$ and $s$ and, finally, 3) the
geodesic $\gamma$ between $s$ and $p$. This curve $\tilde{\gamma}$
connects $\Sigma$ to $p$ and has strictly bigger length than $\gamma$
by the twin paradox \ref{twin}. Moreover, $\tilde{\gamma}$ can be
smoothed while retaining bigger length. A rigorous version of this
statement is provided by Penrose in \cite{penrose} on p. 65 as theorem
7.27. Penrose uses in his didactically excellent introduction a so
called synchronous coordinate system that simplifies his calculation. 
\end{proof}

\subsubsection{The topological space $C(p,q)$ of continuous curves from $p$ to
$q$ in $M$}

Now we define the topological space of the set of continuous curves
$C(p,q)$ from points $p$ to $q\in M$. We will use this in the proof
of the singularity theorem. 
\begin{defn}
The topological space $C(p,q)$ of continuous curves from $p$ to
$q$ in $M$
\begin{itemize}
\item Let $(M,g)$ be a strongly causal spacetime. The set $C(p,q)$ denotes
the set of continuous future directed causal curves from $p$ to $q\in I^{+}(p)$. 
\item Similarly, The set $\tilde{C}(p,q)$ denotes the set of smooth future
directed causal curves from $p$ to $q$. 
\item A trip from $p$ to $q$ where $q\in I^{+}(p)$ is a piecewise future
directed timelike geodesic with past endpoint $p$ and future endpoint
$q$. The set of trips from $p$ to $q$ is denoted by $\kappa(p,q)$.
It is obviously a subset of $C(p,q)$
\item Following Geroch's article \cite{Geroch}, we will define a topology
$\mathcal{T}$ on $C(p,q)$. First we define sets 
\begin{equation}
O(U)\equiv\left\{ \lambda\in C(p,q)|\lambda\subset U\right\} 
\end{equation}
and a set $O$ is called open if it can be expressed as $O=\cup O(U)\in\mathcal{T}$,
 where $U\subset M$ is an open set.
\item Let $\left\{ \lambda_{n}\right\} \in C(p,q)$ be a sequence of curves
with $n\in\Lambda\subseteq\mathbb{N}$ and $\Lambda$as some index
set. A point $x\in M$ is a convergence point of $\left\{ \lambda_{n}\right\} $if,
given any open neighborhood $U$ of $x,$ there exist an 
\begin{equation}
N\in\mathbb{N}:\lambda_{n}\cap U\neq\emptyset,\forall n>N
\end{equation}

\item A curve $\lambda\in C(p,q)$ is a convergence curve of $\left\{ \lambda_{n}\right\} $
if each point $x\in\lambda$ is a convergence point of $\left\{ \lambda_{n}\right\} $. 
\item A point $x\in M$ is a limit point of $\left\{ \lambda_{n}\right\} $
if for every $U$ which is a neighborhood of $x$ the set 
\begin{equation}
\left\{ n\in\mathbb{N}:\lambda_{n}\cap U\neq\emptyset\right\} 
\end{equation}
 is infinite. 
\item A curve $\lambda\in C(\Sigma,q)$ is a limit curve of $\left\{ \lambda_{n}\right\} $
if there exists a sub sequence $\left\{ \lambda_{n'}\right\} \subset\left\{ \lambda_{n}\right\} _{n=1}^{\infty}$with
$n'\in\Lambda'\subset\mathbb{N}$ and $\Lambda'$ as some index set,
for which $\lambda$ is a convergence curve.
\item Let $(M,g)$ be a globally hyperbolic spacetime with a Cauchy surface
$\Sigma$ and $q\in I^{+}(\Sigma)$. Then, $C(\Sigma,q)$ is the set
\begin{equation}
C(\Sigma,q)=\cup_{p\in\Sigma}C(p,q)
\end{equation}

\item Similarly, the set $\tilde{C}(p,q)$ is the set 
\begin{equation}
\tilde{C}(\Sigma,q)=\cup_{p\in\Sigma}\tilde{C}(p,q)
\end{equation}

\end{itemize}
\end{defn}
\begin{thm}
\label{thm:(Geroch)3}Let$(M,g)$ be strongly causal spacetime. Then,$(C(p,q),\mathcal{T})$
defines a topological space.\end{thm}
\begin{proof}
For all $O(U_{1}),O(U_{2})\in\mathcal{T},$ we have 
\begin{equation}
O(U_{1})\cap O(U_{2})=O(U_{1}\cap U_{2})\in\mathcal{T}
\end{equation}
and $\forall O(U_{i})\in\mathcal{T}$ we have 
\begin{equation}
\left(\cup_{i}O(U_{i})\right)\in\mathcal{T}
\end{equation}
 Therefore, $(C(p,q),\mathcal{T})$ describes a topological space. \end{proof}
\begin{thm}
Let$(M,g)$ be strongly causal spacetime. Then,$(C(p,q),\mathcal{T})$
is Hausdorff.\end{thm}
\begin{proof}
$(M,g)$ is Hausdorff. This means that for each of the two curves
$\lambda$ and $\lambda'$, with $\lambda\neq\lambda'$, we can find
two open neighborhoods $X,Y$ where $\lambda\subset X$ and $\lambda'\subset Y$
with $X\cap Y=\emptyset$. Since $(M,g)$ is strongly causal, we can
find two neighborhoods $U_{\lambda},U_{\lambda'}$ where $\lambda\subset U_{\lambda}\subset X$
and $\lambda'\subset U_{\lambda'}\subset Y$ where $\lambda$ intersects
$U_{\lambda}$ only twice, and similarly $\lambda'$ intersects $U_{\lambda'}$
only twice. Since $U_{\lambda}\cap U_{\lambda'}=\emptyset$, it follows
that, $U_{\lambda}\cap\lambda'=\emptyset$ and $U_{2}\cap\lambda=\emptyset$
and we have 
\begin{equation}
O(U_{\lambda})=:\left\{ \lambda\in C(p,q)|\lambda\subset U_{\lambda}\right\} 
\end{equation}
and 
\begin{equation}
O(U_{\lambda'})=:\left\{ \lambda'\in C(p,q)|\lambda'\subset U_{\lambda'}\right\} 
\end{equation}
such that 
\begin{equation}
O(U_{\lambda})\cap O(U_{\lambda'})=\emptyset
\end{equation}
which implies that $(C(p,q),\mathcal{T})$ is Hausdorff.\end{proof}
\begin{thm}
\label{thm:separable}Let$(M,g)$ be a strongly causal spacetime,
Then $C(p,q)$ is separable
\end{thm}
In his book \cite{Wald}, Wald writes on p. 206: ``See Geroch \cite{Geroch}
for a sketch of the proof of this result''. Unfortunately, Geroch
just writes the following on this in a footnote at p. 445: 
\begin{quotation}
To construct a countable dense set in $C(p,q)$, choose a countable
dense set $p_{i}$ including the points $p$ and $q$, in $M$, {[}That
one can find such a countable dense set was shown by Geroch in \cite{Geroch2}{]}.
Consider curves $\gamma\in C(p,q)$ which consist of geodesic segments,
each of which lies in a normal neighborhood and joins two $p_{i}$.
\end{quotation}
By our basic definitions, let $X,\mathcal{T}$ be a topological space.
A subset $Y\subset X$ is called dense in $X$ if for every point
$x\in X,$ $x$ is either in $Y$ or is a limit point of $Y$. The
``proof'' of Geroch above merely states that one can construct a
dense set of points in $M$, But by simply connecting these points
to construct a set of curves, it does not become clear that $C(p,q)$
contains a subset $\kappa(p,q)\subset C(p,q)$ of curves such an arbitrary
curve $\gamma\in C(p,q)$ is either in $\kappa(p,q)$ or a limit curve
of $\kappa(p,q)$. 

Fortunately, in his article \cite{penrose}, Penrose provides on p.
50 a picture, which seems to indicate a proof idea. The proof below
is an own attempt to realize the proof idea that is indicated by Penrose
in that picture.
\begin{proof}
Let the curve $\lambda\in C(p,q)$ and $U_{\lambda1}\subset M$ be
an open set where $\lambda\subset U_{\lambda1}$. Since $M$ is a
paracompact, by theorem \ref{riemannian}, one can always define assign
a positive definite Riemannian metric to $M$ and define an open ball
$B(p,a)$ as in eq.\ref{eq:ball} around each point in $p\in M$ with
radius $a$ with respect to the Riemannian metric. 

Without loss of generality, we can define $O_{1,0}=B(p_{0,0},r_{1,0})\subset U_{\lambda1}$
to be an open neighborhood of $p_{0,0}=p$ with radius $r_{1,0}$
that is geodesically convex. Since $M,g$ is strongly causal, one
can then define a closed neighborhood$\overline{X}_{1,0}\subset O_{1,0}$
with such that $\overline{X}_{1,0}$ is intersected by any curve only
twice. 

We now denote the intersection point of $\overline{X}_{1,0}$ and
$\lambda$ as $p_{1,0}\in I^{+}(p_{0,0})$. For $p_{1,0}$ we can
then find a geodesically convex open neighborhood $O_{20}=B(p_{1,0},r_{2,0})$
with radius $r_{20}$ . We can find a closed neighborhood $\overline{X}_{2,0}\subset O_{2,0}$
that is intersected by any curve only twice. We then go to the intersection
point $p_{2,0}\in I^{+}(p_{1,0})$ of $\overline{X}_{2,0}$ and $\lambda$.
We can continue inductively, constructing $n\in\mathbb{N}$ balls
with radii $r_{i,0}$ and $n$ intersection points $p_{i,0}$ until
we reach a neighborhood $X_{n,0}$ that contains $q=p_{n+1,0}$. So
we get a countable series of points in geodesically convex neighborhoods.
We connect the pairs of points $p_{i,0},p_{i+1,0}$ with the unique
timelike geodesic given by theorem \ref{twin} for $i=0,\ldots,n$
and get a trip $\lambda_{1}$ from $p$ to $q$ which is in an open
set $O_{\lambda_{1}}=\cup_{i}O_{i,0}\subset U_{\lambda_{1}}$. 

Now we begin the sequence again, but this time, we let the open balls
$O_{i,1}$have radii of $r_{i,1}=r_{i,0}/2$. By connecting the points
$p_{i,1},p_{i+1,1}$ for all $i=0,\ldots n$, we get another trip
$\lambda_{2}$ from $p$ to $q$ which is in a neighborhood $O_{\lambda_{2}}=\cup_{i}O_{i,2}\subset O_{\lambda_{1}}$.
We continue inductively, with $O_{i,j}=B(p_{i,j},r_{i,0}/2^{j})$
and we get a countable series of trips $\lambda_{j}$ in open neighborhoods
\begin{equation}
O_{\lambda_{j+1}}=\cup_{i}O_{i,j+1}\subset O_{\lambda_{j}}=\cup_{i}O_{i,j}
\end{equation}
where $O_{\lambda_{1}}\subset U_{\lambda_{1}}$. Since the radii of
the balls $O_{i,j}$ around the curve converge to zero with increasing
$j$, there exists for any open set $U_{\lambda}\in M$ where $\lambda\subset U_{\lambda}$
an $N\in\mathbb{N}$ such that the trip $\lambda_{j}\subset U$ for
all $j>N$. Therefore, if $\lambda\in C(p,q)$ is not itself a trip,
it will be the limit curve of a countable sequence of trips. Thereby,
$C(p,q)$ contains a countable dense subset, which implies that $C(p,q)$
is separable.\end{proof}
\begin{thm}
\label{thm:(Geroch)} Let$(M,g)$ be a strongly causal spacetime.
Then, $\left(C(p,q),\mathcal{T}\right)$ is second countable\end{thm}
\begin{proof}
The proof idea is taken from Geroch's article. Since the set of trips
$\kappa(p,q)$ is dense in $C(p,q)$, $C(p,q)$ must be separable.
A metric on $C(p,q)$ is obtained using the maximal distance between
curves with respect to a positive definite metric that we can always
assign to $M$ by theorem \ref{riemannian}. As $C(p,q)$is a separable
metrizable space, by theorem \ref{thm:metrizable}, it is then second
countable.
\end{proof}
The proof is taken from Geroch's article\cite{Geroch}. I have added
some comments and changes to avoid some potential problems. In the
proof of Geroch, $t$ lies in the interval $[0,1]$. However, it is
not guaranteed that the curve will have a parametrization which, if
$t$ lies in this interval, leads $\lambda$ to move out of the geodesically
convex neighborhood. Then, the existence of a curve $\gamma(t)$ in
this interval could not be guaranteed. 
\begin{thm}
\label{thm:(Geroch)2}(Geroch) Let $(M,g)$ be stably causal and let
$\lambda$ be a past inextendible curve starting at point $p\in M$.
Then, there exists a past inextendible timelike curve $\gamma$ such
that $\gamma(t)\in I^{+}(\lambda(t))$ .\end{thm}
\begin{proof}
Since $M$ is paracompact, one can, by theorem \ref{riemannian},
assign a positive definite Riemannian metric to $M$. Without loss
of generality, we assume the curve parameter to be in $[0,\infty)$
and $\lambda(0)=p$. Let $d(p,q)$ denote the infimum of the length
of all curves connecting $p$ with $q$ measured with respect the
positive definite Riemannian metric that we can always assign to $M$.
We may choose a point $q\in I^{+}(p)$, where $p,q\in U_{1}$ with
$U_{1}$ as a geodesically convex neighborhood of $p$ and 
\begin{equation}
d(p,q)<C,
\end{equation}
 where $C$ is some positive constant. 

Starting at $q$ we can, using theorem \ref{twin}, construct a timelike
curve $\gamma(t)$ consisting of segments that connect two points
in geodesically convex neighborhoods. Then, one can find an $\epsilon_{1}>0$
such that this curve can be defined for all $t\in[0,\epsilon_{1}]$
where $\epsilon_{1}>0$ and 
\begin{equation}
\gamma(t)\in I^{+}(\lambda(t))\text{ where }\gamma(t)\in U_{1}\text{ and }d(\gamma(t),\lambda(t))<\frac{C}{1+t}\label{eq:cond}
\end{equation}
Now we can find an $\epsilon_{2}>\epsilon_{1}$ such that 
\begin{equation}
\gamma(\epsilon_{1})\in I^{+}(\lambda(\epsilon_{2})),\gamma(t)\in U_{2}
\end{equation}
 and $U_{2}$ is a convex neighborhood with $\gamma(\epsilon_{1})\in U_{2}$
and$\lambda(\epsilon_{2})\in U_{2}$. With theorem \ref{twin}, we
may extend $\gamma$ in $U_{2}$ such that it is timelike and $\forall t\in[\epsilon_{1},\epsilon_{2}]$
the curve $\gamma$ is subject to eq. (\ref{eq:cond}) with $U_{1}$
replaced by $U_{2}$ By induction, we continue to extend $\gamma$
to a curve, defined for $t\in[0,\infty)$. Since 
\begin{equation}
d(\gamma(\infty),\lambda(\infty))=0,
\end{equation}
 The resulting curve will be past inextendible because any endpoint
of $\gamma$ would also be an endpoint of $\lambda$.\end{proof}
\begin{thm}
(Hawking-Ellis) \label{thm:(Hawking-Ellis)} Let $(M,g)$ be a time
ordered manifold and let $\left\{ \lambda_{n}\right\} $be a sequence
of past inextendible causal curves which have a limit point $p$ Then
there exists a past inextendible causal curve $\lambda$ passing through
$p$ which is a limit curve of $\left\{ \lambda_{n}\right\} $. 
\end{thm}
The proof is taken from the book of Hawking and Ellis\cite{HawkingEllis}.
I have added some minor adaptions that are necessary for our situation.
In his book Hawking considers a sequence of future inextendible curves.
\begin{proof}
Since $M$ is a paracompact, by theorem \ref{riemannian}, one can
always assign a positive definite Riemannian metric to $M$ and define
an open ball $B(p,a)$ as in eq. \ref{eq:ball} around each point
in $p\in M$ with radius $a$ with respect to the Riemannian metric.
Let $U_{1}$ be a geodesically convex neighborhood as in theorem \ref{twin}
around $p$. Let $b>0$ be such that $B(p,b)$ is defined on $M$
and let $\left\{ \lambda_{1,0,n}\right\} $ be a subsequence of $\lambda_{n}\cap U_{1}$
which converges to $p$ Since $\overline{B}(p,b)$ is compact, by
theorem \ref{thm:(Bolzano-Weierstrass)} it will contain limit points
of $\left\{ \lambda_{1,0,n}\right\} $. Any such limit point $y$
must either lie in $J^{-}(p,U_{1})$ or $J^{+}(p,U_{1})$ since otherwise,
there would be a neighborhood $V_{1}$of $y$ and $V_{2}$ of $p$
between which there would be no spacelike curve in $U_{1}$. Let 
\begin{equation}
x_{11}\in J^{-}(p,U_{1})\cap\overline{B}(p,b)
\end{equation}
be one of these limit points, and let $\left\{ \lambda_{1,1,n}\right\} $
be a subsequence of $\left\{ \lambda_{1,0,n}\right\} $ which converges
to $x_{11}$. Then $x_{11}$ will be a point of our limit curve $\lambda$. 

Define 
\begin{equation}
x_{ij}\in J^{-}(p,U_{1})\cap\overline{B}(p,i^{-1}jb)
\end{equation}
as a limit point of the subsequence $\left\{ \lambda_{i-1,i-1,n}\right\} $for
$j=0$ and of $\left\{ \lambda_{i,j-1,n}\right\} $ for $i\geq j\geq1$
and define $\left\{ \lambda_{i,j,0}\right\} $as subsequence of the
above subsequence which converges to $x_{ij}$. Any two of the $x_{ij}$
will have non spacelike separation. Therefore, the closure of the
union of all $x_{ij}$for $j\geq i$ will give a non spacelike curve
$\lambda$ from $p=x_{i0}$ to $x_{11}=x_{ii}$ . 

The subsequence 
\begin{equation}
\left\{ \lambda_{n}'\right\} =\left\{ \lambda_{m,m,n}\right\} 
\end{equation}
of $\left\{ \lambda_{n}\right\} $ intersects each of the balls 
\begin{equation}
B(x_{mj},m^{-1}b)
\end{equation}
 for $0\leq j\leq m$ and therefore $\lambda$ will be a limit curve
of $\left\{ \lambda_{n}\right\} $ from $p$ to $x_{11}$. Now we
let $U_{2}$ be a convex neighborhood about $x_{11}$ and repeat this
construction, where we begin this time with the sequence $\left\{ \lambda'_{n}\right\} $. \end{proof}
\begin{thm}
\label{thm:intersect}Let $(M,g)$ be a globally hyperbolic Lorentzian
manifold with a Cauchy surface $\Sigma$ and $\lambda$ be a past
inextendible curve. Then $\lambda$ intersects $I^{-}(\Sigma)$.\end{thm}
\begin{proof}
Suppose $\lambda$ did not intersect $I^{-}(\Sigma)$. By theorem
\ref{stablystronglycausal}, $(M,g)$ is stably causal and then by
theorem \ref{thm:(Geroch)3}, we could find a past inextendible curve
$\gamma$, such that 
\begin{equation}
\forall t:\gamma(t)\in I^{+}(\lambda(t)),
\end{equation}
where 
\begin{equation}
I^{+}(\lambda)\subset I^{+}(\Sigma\cup I^{+}(\Sigma))=I^{+}(\Sigma)
\end{equation}
that is $\gamma(t)\subset I^{+}(\Sigma)$. Even if we extend $\gamma$
indefinitely into the future, it can not intersect $I^{-}(\Sigma)$
because otherwise the achronality of $\Sigma$ would be violated.
Since by definition all past inextendible curves must intersect $\Sigma$,
no such $\gamma$ can exist and $\lambda$ must enter $I^{-}(\Sigma)$.\end{proof}
\begin{thm}
\label{thm:compact-1}Let $(M,g)$ be a globally hyperbolic spacetime
$\Sigma$ a Cauchy surface and $q\in I^{+}(\Sigma)$, then $C(\Sigma,q),$
is compact.
\end{thm}
This proof is adapted from Wald \cite{Wald} on p. 206, who took it
from Hawking and Ellis\cite{HawkingEllis}. I have merely added some
additional explanations.
\begin{proof}
By theorem \ref{stablystronglycausal}, $(M,g)$ is strongly causal.
By theorem\ref{thm:(Geroch)}, $C(p,q)$ is then second countable
for all $p,q\in M$. We begin by showing that $C(p,q)$ with $p\in\Sigma,q\in D^{+}(p)$
is compact. By theorem \ref{thm:(Bolzano-Weierstrass)}, one only
has to show that every sequence of curves $\left\{ \lambda_{n}\right\} _{n=1}^{\infty}\in C(p,q)$
has a limit curve $\lambda\in C(p,q)$. Then $C(p,q)$ is compact.
Let $\left\{ \lambda_{n}\right\} _{n=1}^{\infty}$ be a sequence of
future directed causal curves from $p\in\Sigma$ to $q\in I^{+}(p)$.
Hence, $\left\{ \lambda_{n}\right\} $ is a sequence of causal curves,
all passing through $q$ by construction. Therefore, $q$ is a limit
point of this sequence. 

Now, we temporarily remove $p$. The spacetime $(M,g)$ is Hausdorff.
Hence, we can choose an open neighborhood $V_{i}$ of a point $x_{i}\neq p$
with $p\notin V_{i}$. We can do this for all points $x\in M\backslash p$.
The union $\cup V_{i}$ is then open and $\cup V_{i}=M/p$. With $M\backslash p$
open and $g$, $(M\backslash p,g)$ is still a Lorentzian manifold.
Furthermore, $(M\backslash p,g)$ is still strongly causal since we
can define a global time function $\tilde{t}=M\backslash p\rightarrow\mathbb{R}$
with $\tilde{t}=t\forall x\in M\backslash p$ where $t$ is the time
function of $(M,g)$. Strong causality of $(M\backslash p,g)$ implies
stable causality by theorem \ref{stablystronglycausal}. Hence, the
removal of $p$ does not change the definitions of past an future
directions of the curves. However, the resulting curves $\left\{ \lambda_{n}\right\} \in M\backslash p$
do not have a past endpoint $p\in M$ anymore. Since all curves $\left\{ \lambda_{n}\right\} $
now end before $\Sigma,$ the removal of $p$ from $M$ implies that
$\Sigma$ is not a Cauchy surface anymore. 

Since $\left\{ \lambda_{n}\right\} $ are past inextendible in $M\backslash p$
and $(M\backslash p,g)$ is stably causal, we can apply the proof
in theorem \ref{thm:(Hawking-Ellis)} with some modification and get
a limit curve $\lambda\in M\backslash p$ of$\left\{ \lambda_{n}\right\} $.
One has, however, to be cautious with the compact balls $\overline{B}(q,b_{i})$
that are used with theorem \ref{thm:(Hawking-Ellis)}. The radiuses
$b_{i}$ of these balls have to be such that the removed point $p$
is never contained in the interior of any one these balls or intersected
by their closure. That is, the balls must lie entirely within $M\backslash p$.
This can be easily achieved by choosing the $\overline{B}(q,b)$ sufficiently
small. The limit curve $\lambda$ obtained from theorem \ref{thm:(Hawking-Ellis)}
is passing through $q$ and comes infinitely close to $p$. 

By theorem \ref{thm:intersect}, there can not be a past inextendible
curve in a globally hyperbolic spacetime that does not intersect $I^{-}(\Sigma)$
. Since none of the curves $\lambda_{n}$ enter $I^{-}(\Sigma)$,
the spacetime $(M\backslash p,g)$ can not a globally hyperbolic spacetime.
If we add $p$ again to $M$, $(M,g)$ becomes globally hyperbolic
be definition. All$\left\{ \lambda_{n}\right\} _{n=1}^{\infty}$ are
future directed causal curves beginning in $\Sigma$. Hence, they
do not enter $I^{-}(\Sigma).$ With $p\in\Sigma$ restored, either
$\lambda$ will remain past inextendible, or it will end in $p\in\Sigma$.
Suppose, $\lambda$ will remain past-inextendible. Then, it must enter
$\Sigma$ since in a globally hyperbolic spacetime, all past inextendible
curves enter $\Sigma$. Suppose $\lambda$ also enters $I^{-}(\Sigma)$.
$(M,g)$ is Hausdorff, therefore, we can find for any point $u\in\lambda\cap I^{-}(\Sigma)$
a neighborhood that does not lie in $p\in\Sigma$. Hence, if $\lambda$were
past inextendible,$\lambda$ could not be a limit curve of $\left\{ \lambda_{n}\right\} _{n=1}^{\infty}\in C(p,q)$,
where $p\in\Sigma$. Therefore, $\lambda$ must end in $p.$ As a
result, $\lambda\in C(p,q),$ where $p\in\Sigma,q\in I^{+}(p)$. So,
every sequence $\left\{ \lambda_{n}\right\} _{n=1}^{\infty}\in C(p,q)$
has an accumulation curve $\lambda\in C(p,q)$ and by theorem \ref{thm:(Bolzano-Weierstrass)}
$C(p,q)$ then must be compact. By definition, $C(\Sigma,q)$ is the
union of compact sets. Therefore, by theorem \ref{thm:compact-1},
$C(\Sigma,q)$ it is also compact.\end{proof}
\begin{defn}
Finally, we need some definitions for continuous, non differentiable
curves
\begin{itemize}
\item We call the length functional $\tau$ of a curve in $\tilde{C}(p,q)$
upper semicontinuous in $\tilde{C}(p,q)$ if for each $\lambda\in\tilde{C}(p,q)$
given $\epsilon>0$ there exists an open neighborhood $O\subset\tilde{C}(p,q)$
of $\lambda$ such that for all $\lambda'\in O$, we have 
\begin{equation}
\tau(\lambda')\geq\tau(\lambda)+\epsilon
\end{equation}
 where $\tau$ is the length functional. 
\item For a function $\gamma\in C(p,q)$ and an open neighborhood $O\subset C(p,q)$
of $\gamma$ we define 
\begin{equation}
T(O)=sup\left\{ \tau(\lambda)|\lambda\in O,\lambda\in\tilde{C}(p,q)\right\} 
\end{equation}
and 
\begin{equation}
\tau(\gamma)=inf\left\{ T(O)|O\text{ is an open neighborhood of }\mu\right\} 
\end{equation}

\end{itemize}
\end{defn}
\begin{thm}
(Hawking-Ellis) Let $(M,g)$ be a strongly causal spacetime and $p,q\in M,q\in I^{+}(p)$
then, $\tau(\lambda)$ with $\lambda\in\tilde{C}(p,q)$ is upper semicontinuous.\end{thm}
\begin{proof}
The proof is taken from Wald\cite{Wald}, who took it from the book
of Hawking and Ellis\cite{HawkingEllis}. I only have added some additional
explanations. Let $\lambda\in\tilde{C}(p,q)$ be parametrized by proper
time 
\begin{equation}
t=\int_{t=0}^{t(p)}\sqrt{-T^{\mu}T_{\mu}}dt\label{eq:time}
\end{equation}
 with a tangent $T$. Within a geodesically convex neighborhood of
each point $r\in\lambda$, the spacelike geodesics orthogonal to $T$
form a three dimensional spacelike hypersurface.Within a small open
neighborhood $U$ of $\lambda$, a unique hypersurface will pass through
each point of $U$. On $U$ we can define a function $F(p)=t$ which
corresponds to the proper time value of$\lambda$ at the hypersurface
in which $p$ lies. From eq. (\ref{eq:time}), we get$-\nabla^{\mu}F=T^{\mu}$
on $\lambda$. Since $T^{\mu}$ is timelike, so is $\nabla^{\mu}F$
on $\lambda$. Furthermore, the tangent vector is normalized, and
so 
\begin{equation}
T^{\mu}T_{\mu}=\nabla^{\mu}F\nabla_{\mu}F=-1
\end{equation}
on $\lambda$.

Let $\gamma\in\tilde{C}(p,q)$ with $\gamma\in U$. We parametrize
$\gamma$ by $F$ and let it have a tangent $l^{\mu}$. Then $l^{\mu}\nabla_{\mu}F=1$.
We can compose $l^{\mu}$ in a spacelike and timelike part 
\begin{equation}
l^{\mu}=\alpha\nabla^{\mu}F+m^{\mu}
\end{equation}
where $m^{\mu}\nabla_{\mu}F=0$ in order to fulfill $l^{\mu}\nabla_{\mu}F=1$,
and thus $m^{\mu}$ is spacelike. Then 
\begin{equation}
l^{\mu}\nabla_{\mu}F=1=\alpha\nabla^{\mu}F\nabla_{\mu}F
\end{equation}
and therefore 
\begin{equation}
l^{\mu}=\frac{\nabla^{\mu}F}{\nabla^{\mu}F\nabla_{\mu}F}+m^{\mu}
\end{equation}
and 
\begin{equation}
l_{\mu}l^{\mu}=\frac{1}{\nabla^{\mu}F\nabla_{\mu}F}+m_{\mu}m^{\mu}
\end{equation}
since $m_{\mu}m^{\mu}$is spacelike it is $\geq0$ therefore 
\begin{eqnarray}
\sqrt{-l_{\mu}l^{\mu}} & \leq & \sqrt{-\frac{1}{\nabla^{\mu}F\nabla_{\mu}F}}
\end{eqnarray}

The function $\nabla^{\mu}F$ is continuous with 
\begin{equation}
\nabla^{\mu}F\nabla_{\mu}F=-1
\end{equation}
on $\lambda$ . Therefore, given an $\epsilon>0$, we can choose a
neighborhood $U'\subset U$ of $\lambda$ on which 
\begin{equation}
\sqrt{-\frac{1}{\nabla^{\mu}F\nabla_{\mu}F}}\leq1+\frac{\epsilon}{\tau(\lambda)}
\end{equation}
Hence, for any $\gamma\in\tilde{C}(p,q)$ contained in $U',$ after
we switch the parametrization to the arc length where $\frac{d\tau}{dt}=1$:
\begin{eqnarray}
\tau(\gamma) & = & \int_{0}^{t(p)}\sqrt{-l^{a}l_{a}}dF\nonumber \\
 & = & \int_{0}^{t(p)}\sqrt{-l^{a}l_{a}}\frac{d\tau}{dt}dt=\int_{0}^{\tau(\lambda)}\sqrt{-l^{a}l_{a}}d\tau\nonumber \\
 & \leq & \left(1+\frac{\epsilon}{\tau(\lambda)}\right)\tau(\lambda)=\tau(\lambda)+\epsilon
\end{eqnarray}
As a result $\tau(\gamma)\leq\tau(\lambda)+\epsilon$ for all $\lambda\in\tilde{C}(p,q)$
and $\gamma\in\tilde{C}(p,q)$ where $\gamma\in U$,  with $U$ as
a neighborhood of $\lambda$. This argument holds for all curves in
$\tilde{C}(p,q)$. Therefore, it also holds for curves in $\tilde{C}(\Sigma,q)$. \end{proof}
\begin{thm}
\label{thm:maximal }Let $(M,g)$ be a global hyperbolic spacetime
with Cauchy surface $\Sigma$ and a point $p\in D^{+}(\Sigma)$. Then
among the timelike curves connecting $p$ to $\Sigma$ there exists
a cube $\gamma$ with maximal length. This is a timelike geodesic
orthogonal to $\Sigma$. 
\end{thm}
The first paragraph of this proof is from Wald\cite{Wald}, the rest
is an own attempt. 
\begin{proof}
The length functional $\tau(\gamma)$ of a curve $\gamma\in C(\Sigma,q)$
is a continuous function. Because of theorem \ref{thm:compact-1}
$C(\Sigma,q)$ is compact and due to theorem \ref{thm:attainsmaximum}
the length functional then becomes maximal for a curve in $C(\Sigma,q)$.
We have to show that this maximum is attained for a timelike geodesic.
We begin by proving that the timelike geodesic maximizes length compared
to any other continuous curve in a geodesically convex neighborhood.
Given a geodesically convex neighborhood $O,$ and two points $p,q$
in $O$, by theorem \ref{twin} the unique geodesic $\gamma$ connecting
$p,q$ has greater than or equal length than any piecewise smooth
curve $\mu$ connecting $p,q$. Therefore, by upper semicontinuity,
we must have 
\begin{equation}
\tau(\mu)\leq\tau(\gamma)
\end{equation}
for any continuous curve $\mu$. If equality held with $\mu\neq\gamma$
define a point $u\in\mu,u\notin\gamma$. Let $\gamma_{1}$ be the
geodesic segment connecting $p,u$ and $\gamma_{2}$ be the geodesic
segment connecting $u,q$. Since by theorem \ref{twin}, $\gamma_{1}$
maximizes length from $p$ to $u$ and $\gamma_{2}$maximizes length
from $u$ to $q$, we have 
\begin{equation}
\tau(\gamma_{1})+\tau(\gamma_{2})\geq\tau(\mu)=\tau(\gamma)
\end{equation}
which contradicts theorem \ref{twin} that implies $\gamma$ having
greater length than any other piecewise smooth curve connecting $p,q$.
Therefore we have $\tau(\mu)<\tau(\gamma)$ on $O$.

Let $\gamma$ be an arbitrary curve in $C(p,q)$. With the construction
in theorem \ref{thm:separable}, we can cover $\gamma$ with a countable
set of points $\{p_{ij}\}$ in geodesically convex open neighborhoods
$O_{i,j}$, where $\{p_{i,j}\}\subset\{p_{i,j+1}\}$. The pairs of
points $p_{i,j},p_{i+1,j}$ can be piecewise connected by unique timelike
geodesics $\lambda_{i,j}$ of theorem \ref{twin}. The piecewise curve
$\lambda_{j}=\cup\lambda_{ij}$is then a trip in $\kappa(p,q)\subset C(p,q)$
with a length 
\begin{equation}
\tau(\lambda_{j})=\sum_{i}\tau(\lambda_{ij})
\end{equation}
The length fiction $\tau(\gamma)$ of the original curve $\gamma$
can also be described by a sum of lengths of its own segments $\gamma_{ij}$
that connect the points $p_{i,j},p_{i+1,j}\in O_{i,j}$: 
\begin{equation}
\tau(\gamma)=\sum_{i}\tau(\gamma_{ij})
\end{equation}
By the first paragraph, we have 
\begin{equation}
\tau(\lambda_{ij})\geq\tau(\gamma_{ij})\forall i,j\label{eq:curvelength2}
\end{equation}
and therefore 
\begin{equation}
\tau(\lambda_{j})=\sum_{i}\tau(\lambda_{ij})\geq\sum_{i}\tau(\gamma_{ij})=\tau(\gamma)\forall j,\label{eq:curvelength}
\end{equation}
with equality in eqs (\ref{eq:curvelength2}) and (\ref{eq:curvelength})
if and only if $\gamma$ is a timelike geodesic. In this case, the
segments $\lambda_{ij}$ would all be lying on $\gamma_{ij}$ for
all $i,j$, resulting in the equality of the length functions $\tau(\lambda_{j})$
and $\tau(\gamma)$. 

Since $\{p_{i,j}\}\subset\{p_{i,j+1}\}$ are in convex neighborhoods,
repeated application of theorem \ref{twin} leads for a continuous
curve $\gamma$ to the following expression for the approximating
trips: 
\begin{equation}
\tau(\lambda_{j})\geq\tau(\lambda_{j+1})\label{eq:curvelength1}
\end{equation}
Again, in eq. (\ref{eq:curvelength1}), we have equality in if and
only if $\gamma$ is a timelike geodesic. In the latter case, the
segments $\lambda_{ij}$ and $\lambda_{j+1}$ would be just a description
of the same curve $\gamma$ and thereby have equal length. 

For a continuous curve $\gamma$ eqs. (\ref{eq:curvelength1}) and
(\ref{eq:curvelength}) imply that the length

\begin{equation}
\tau(\lambda)=\lim_{j\rightarrow\infty}\sum_{i}\tau(\lambda_{ij})
\end{equation}
of the limit curve 
\begin{equation}
\lambda=\lim_{j=\infty}\lambda_{j}
\end{equation}
is approximating the length of $\gamma$ from above, i.e. 
\begin{equation}
\tau(\lambda)\rightarrow\tau(\gamma),
\end{equation}
 where $\tau(\lambda)=\tau(\gamma)$ if and only if $\gamma$ is a
timelike geodesic. 

Assume $\gamma$ were, on some segment $\gamma_{ij}$ connecting $p_{i,j}$
and $p_{i+1,j}$, not a timelike geodesic. Then, one could get a longer
curve by replacing the segment $\gamma_{ij}$ by the corresponding
segment $\lambda_{ij}$ of our construction. Hence $\tau$ would not
reach its maximum on $\gamma$.

Now assume $\gamma$ would be a timelike geodesic. Then, we have equality
in eqs (\ref{eq:curvelength1}) and (\ref{eq:curvelength}),  which
yields 
\begin{equation}
\tau(\lambda)=\lim_{j\rightarrow\infty}\sum_{i}\tau(\lambda_{ij})=\sum_{i}\tau(\lambda_{i1})=\tau(\lambda_{1})=\tau(\gamma)
\end{equation}
for the entire curve. Hence, $\tau(\gamma)$ would have reached its
maximum, since on each segment $\lambda_{ij},$ the length $\tau(\lambda_{ij})$
is at its maximum. 

Finally, a length maximizing geodesic must be orthogonal to $\Sigma$,
since by theorem \ref{orthogonal}, a length maximizing geodesic starting
on $\Sigma$ must be orthogonal to $\Sigma$ or would otherwise be
possible to increase its length. One can choose some geodesically
convex open neighborhood $U,$ where $U\cap\Sigma\neq\emptyset$ and
$U\cap\gamma\neq\emptyset$. Then one can connect a point $u\in\gamma$
where $u\in U$ to a point $p\in\Sigma$ with a geodesic segment orthogonal
to $\Sigma$. The curve obtained by following $\gamma$ from $q$
to $u$ and the segment connecting $u$ to $\Sigma$ with a geodesic
segment orthogonal to $\Sigma$, would by theorem \ref{orthogonal}
have bigger length than $\gamma$. 
\end{proof}

\subsubsection{The Penrose and Penrose-Hawking singularity theorems of General Relativity}
\begin{prop}
Let $(M,g)$ be a global hyperbolic spacetime with Cauchy surface
where $R_{\mu\nu}\zeta^{\mu}\zeta^{\nu}\geq0$ for all timelike $\zeta$
Suppose there exists a Cauchy surface for which the trace of the extrinsic
curvature $K<0$. Then no future directed timelike curve can be extended
beyond proper time $\tau_{0}=3/|K|$\end{prop}
\begin{proof}
Suppose there exists a future directed timelike curve $\lambda$ from
$\Sigma$ that can be defined for some proper time greater than $\tau_{0}$.
Then, one can define the curve up to a point $p=\lambda(\tau_{0}+\epsilon)$
with $\epsilon>0$. According to theorem \ref{thm:maximal }, there
would exist a timelike geodesic $\gamma$ with maximal length that
connects $\Sigma$ to $p,$ with $\gamma$ being orthogonal to $\Sigma$.
Because $\tau(\lambda(p))=\tau_{0}+\epsilon$, we would have $\tau(\gamma(p))\geq\tau_{0}+\epsilon$.
Theorem \ref{conjugate} implies that $\gamma$ would develop conjugate
points at $\tau_{0}$. Finally, theorem \ref{maximizelength} implies
that then $\gamma$ then fails to maximize length. Hence, no timelike
curve can be extended beyond $\tau_{0}$. 
\end{proof}
The most problematic assumption in this theorem is that the spacetime
should be globally hyperbolic. However, Hawking and Penrose were able
to eliminate this condition in \cite{Hawking}.
\begin{defn}
To state the Penrose-Hawking singularity theorem, one needs the following
definitions: 
\begin{itemize}
\item An edge of a closed achronal set $\Sigma$ is the set of all points
$p\in\Sigma$ such that every open neighborhood $O$ of $p$ contains
a point $q\in I^{+}(p)$. 
\item A spacetime $(M,g)$ satisfies the timelike generic condition if each
timelike geodesic with tangent $\zeta^{\mu}$possesses at least on
point where $R_{\mu\nu}\zeta^{\mu}\zeta^{\nu}\neq0$ . 
\item $(M,g)$ satisfies the null generic condition if every null geodesic
with tangent $k^{\mu}$possesses at least one point where either $R_{\mu\nu}k^{\mu}k^{\nu}\neq0$
or 
\[
k_{[\mu}R_{\nu]\rho\sigma[\lambda}k_{\gamma]}k^{\rho}k^{\sigma}\neq0
\]
. 
\item A compact two dimensional smooth spacelike submanifold $\Omega$ where
the expansion $\theta$ along all ingoing and outgoing future directed
null geodesics orthogonal to $\Omega$ is negative, is called a trapped
surface. 
\end{itemize}
\end{defn}
Then, the Penrose-Hawking singularity theorem \cite{Hawking} reads:
\begin{prop}
Let $(M,g)$ be a spacetime where 1) $R_{\mu\nu}\zeta^{\mu}\zeta^{\nu}\geq0$
for all timelike and null $\zeta$ 2) The timelike and null generic
conditions are satisfied, 3) no closed timelike curves exist, 4) one
of the following a) $(M,g)$ possesses a compact achronal set without
edge, b) $(M,g)$ possesses a trapped null surface, c) there exists
a point $p\in M$ such that the expansion of future or past directed
null geodesics emanating from p becomes negative along each geodesic
in this congruence. Then $(M,g)$ must contain at least one incomplete
timelike or null geodesic. 
\end{prop}
The generic conditions can be assumed to hold in any reasonable model
of the universe. Similarly, closed timelike curves must reasonably
be excluded for a physical spacetime. 

Contracting Einstein's equation with matter, 
\begin{equation}
R_{\mu\nu}-\frac{1}{2}Rg_{\mu\nu}=8\pi T_{\mu\nu}
\end{equation}
we get 
\begin{equation}
R=-8\pi T,
\end{equation}
or
\begin{equation}
R_{\mu\nu}=8\pi(T_{\mu\nu}-\frac{1}{2}Tg_{\mu\nu})
\end{equation}
we get from $R_{\mu\nu}\zeta^{\mu}\zeta^{\nu}\geq0$
\begin{equation}
8\pi\left(T_{\mu\nu}\zeta^{\mu}\zeta^{\nu}-\frac{1}{2}T\zeta_{\nu}\zeta^{\nu}\right)\geq0
\end{equation}
which is called strong energy condition. In general, the energy momentum
tensor $T_{\mu\nu}$ is not diagonalizable. However, it seems that
in most cases an energy momentum tensor for a perfect fluid
\begin{eqnarray*}
T_{\mu\nu} & = & (\rho+p)\zeta_{\mu}\zeta_{\nu}+pg_{\mu\nu}
\end{eqnarray*}
is a reasonable assumption, where $\zeta$ is the fluid's four velocity,
$\rho$ the energy density and $p$ the pressure, the latter being
assumed to be equal in every direction. Then, $T_{\mu\nu}\zeta^{\mu}\zeta^{\nu}=\rho$
and $T_{\mu\nu}g^{\mu\nu}\zeta_{\nu}\zeta^{\nu}=-(-(\rho+p)+4p)=+\rho-3p$.
Therefore, the strong energy condition becomes:
\begin{equation}
(\rho-\frac{1}{2}\rho+\frac{1}{2}3p)\geq0
\end{equation}
or 
\begin{equation}
\rho+3p\geq0
\end{equation}
If we neglect shear and rotation, the Raychaudhuri equation is then
equivalent to 
\begin{equation}
\frac{d\theta}{d\tau}=-8\pi(\rho+3p)\leq0.
\end{equation}
So, the strong energy condition is responsible for gravity being attractive.
This is a condition that must have been violated during the inflation
era of the universe, however, it seems to hold for the universe that
we observe today

A trapped surface is formed in the interior of the solutions for black
holes. The latter can be indirectly observed by their orbiting stars
\cite{MTW}. The trajectory of those stars reveal a super massive
black hole with large Schwarzschild radius in the center of most galaxies.
By the singularity theorems above, it seems inevitable that by classical
general relativity, a singularity must form. One can show, see e.g.
\cite{MTW} p. 862 that all matter inside a Schwarzschild black hole
falls inevitably towards the singularity, and the length scales of
the matter coming close to the singularity are getting proportional
to 
\begin{equation}
V\propto\tau_{singularity}-\tau
\end{equation}
where $\tau$ is the proper time since the objects have crossed the
event horizon and $\tau_{singularity}$ is the proper time where the
matter has arrived at the singularity. Since no quantum mechanical
object can be pressed in lengths shorter than its Compton wavelength,
one has to seek for an alternative theory that describes the behavior
of matter and the gravitational field close to a singularity. Because
the curvature in the vicinity of a singularity is extremely large
and even becomes infinite at the singularity, a first attempt would
be to use quantum field theory in curved spacetimes. In this framework,
one unfortunately finds that close to the singularity, one is not
able to define ordinary particle propagators, see \cite{DeWit4}.
In order to find a resolution of this problem, one has to define a
consistent theory of quantum gravity. This theory should replace classical
general relativity on the Planck scale and yield expressions that
allow a consistent description of matter inside of black holes.

\section{Covariant quantization of general relativity }

For computing scattering amplitudes, the covariant quantization based
on the perturbative evaluation of the path integral provides an excellent
tool. In the first sub section, we will derive the Feynman rules for
gravitational amplitudes. In the second sub section we will review
the problems associated with loop calculations.

\subsection{The Feynman rules of gravitation}

This section is merely a review of the excellent introductions \cite{Velt,popov,Kiefer,Hamber,thooft}.
With the Lagrangian density $\mathcal{L}=\sqrt{-g}R$, and $L=\int d^{3}x\mathcal{L}$,
the action of the gravitational field 
\begin{equation}
S=\int dtL=\int\sqrt{-g}Rd^{4}x
\end{equation}
is invariant under infinitesimal gauge transformations of the form
\begin{equation}
g_{\mu\nu}\rightarrow\overline{g}_{\mu\nu}\equiv g_{\mu\nu}+g_{\alpha\nu}\partial_{\mu}\eta^{\alpha}+g_{\mu\alpha}\partial_{\nu}\eta^{\alpha}+\eta^{\alpha}\partial_{\alpha}g_{\mu\nu}\label{eq:gaugetrans1}
\end{equation}
where $\eta^{\mu}$ are the components of an arbitrary infinitesimal
vector field. In general relativity, one can show a local flatness
theorem. This theorem says that around a given point $p$ in the spacetime
manifold described by Einstein's equation, there is a local neighborhood
where $g_{\mu\nu}=\eta_{\mu\nu}$ with $\eta_{\mu\nu}$ as the Minkowski
metric. As we are developing a local field theory, we now choose such
a neighborhood and decompose the metric as follows: 
\begin{equation}
g_{\mu\nu}\rightarrow\overline{g}_{\mu\nu}=\eta_{\mu\nu}+h_{\mu\nu}\label{eq:backgroundmetric}
\end{equation}
The symmetric tensor $g_{\mu\nu}$ has the function of a classical
background metric and the symmetric tensor $h_{\mu\nu}$ is supposed
to be described by some quantum field theory. It is invariant under
gauge transformations 
\begin{equation}
h_{\mu\nu}\rightarrow h_{\mu\nu}^{\eta}\equiv h_{\mu\nu}+\nabla_{\mu}\eta_{\nu}+\nabla_{\nu}\eta_{\mu}\label{eq:gauge2}
\end{equation}
Although a decomposition of the metric as in eq. (\ref{eq:backgroundmetric})
can always can be made as a kind of gauge, one should note that this
composition of the metric is in fact highly problematic. The idea
that $h_{\mu\nu}$ is subject to quantum processes implies that a
finite change of $h_{\mu\nu}$ through some quantum process could,
according to Faddeev and Popov (see \cite{popov} on p. 780), for
example, change the signature of the metric $\overline{g}_{\mu\nu}$,
if $h_{\mu\nu}$is not sufficiently small. A change of the signature
of the metric may not be problematic, since in the path integral over
metrics, one would have to sum over all possible metrics, including
ones with different signature. However, Wald notes in \cite{Wald}
on p. 384 that the perturbation theory we will obtain below satisfies,
``for all orders, the causality relations with respect to the background
metric $g_{\mu\nu}$and not with respect to $\overline{g}_{\mu\nu}$''.
Wald writes that the entire summed series may still obey the correct
causality relation with respect to $\overline{g}_{\mu\nu}$, if it
were to converge. Later, we will see, however that it does not converge.
This would effectively imply that the theory makes only sense for
small $h_{\mu\nu}$.

Here, however, we will not be bothered by such problems, and postpone
the descriptions of quantum mechanical black holes. Instead, we derive
the Feynmanrules for perturbative quantum gravity. We can derive them
from the path integrals like 
\begin{equation}
Z=\int\mathcal{D}g_{\mu\nu}e^{iS}.\label{eq:Pathintegral}
\end{equation}
where $\mathcal{D}g_{\mu\nu}$ is understood as the usual summation
over all possible metrics. Using the expansion of the metric in eq.
(\ref{eq:backgroundmetric}), we can expand the action around the
background $\eta_{\mu\nu}$ as 
\begin{equation}
S=\int d^{4}x\left(\sqrt{-g}R+\underline{\mathcal{L}}+\underline{\underline{\mathcal{L}}}\right)
\end{equation}
where $g$ and $R$ are constructed from the background metric $g_{\mu\nu}=\eta_{\mu\nu}$
and $\underline{\mathcal{L}}$ is linear in $h_{\mu\nu}$ and $\underline{\underline{\mathcal{L}}}$
is quadratic in $h_{\mu\nu}$. The path integral of eq (\ref{eq:Pathintegral})
will then become an expression like 
\begin{equation}
Z=e^{i\int d^{4}x\sqrt{-g}R}\int\mathcal{D}h_{\mu\nu}e^{i\int d^{4}x\left(\underline{\mathcal{L}}+\underline{\underline{\mathcal{L}}}+\ldots\right)}
\end{equation}
where $\int d^{4}x\sqrt{-g}R=0$ if $g_{\mu\nu}=\eta_{\mu\nu}$.

The part of the action quadratic in $h_{\mu\nu}$ can be used to derive
the propagator. In the following, we will adopt the notation of `t
Hooft and Veltman for the metric tensor $g^{\alpha\gamma}h_{\gamma\beta}=h_{\beta}^{\alpha}$.
Since the metric is a symmetric tensor, we do not need to worry on
which index comes first. Following Veltman and `t Hooft, we begin
the calculation by setting 
\begin{equation}
\overline{g}_{\mu\nu}=g_{\mu\nu}+h_{\mu\nu}=g_{\mu\alpha}(\delta_{\nu}^{\alpha}+h_{\nu}^{\alpha})\label{eq:metric}
\end{equation}
where $g_{\mu\nu}$ is at first arbitrary, and indices are raised
and lowered by $g_{\mu\nu}$. `t Hooft and Veltman get 
\begin{equation}
\overline{g}^{\mu\nu}=g^{\mu\alpha}(\delta_{\alpha}^{\nu}-h_{\alpha}^{\nu}+h_{\alpha}^{\beta}h_{\beta}^{\nu})=g^{\mu\nu}-h^{\mu\nu}+h^{\mu\beta}h_{\beta}^{\nu}\label{eq:metric1}
\end{equation}
and using 
\[
\sqrt{-\overline{g}}=\sqrt{-det(\overline{g})}=e^{\frac{1}{2}tr(ln(-det(\overline{g})))}=\sqrt{-det(g)}e^{\frac{1}{2}tr(ln(\delta_{\nu}^{\alpha}+h_{\nu}^{\alpha}))}
\]
they arrive at 
\begin{equation}
\sqrt{-\overline{g}}=\sqrt{-g}(1+\frac{1}{2}h_{\alpha}^{\alpha}-\frac{1}{4}h_{\alpha\beta}h^{\alpha\beta}+\frac{1}{8}(h_{\alpha}^{\alpha})^{2}).\label{eq:determinant1}
\end{equation}
With eq.(\ref{eq:christoffel}), `t Hooft and Veltman decompose the
Christoffel connection as 
\begin{equation}
\overline{\Gamma}_{\mu\nu}^{\alpha}=+\Gamma_{\mu\nu}^{\alpha}+\underline{\Gamma}_{\mu\nu}^{\alpha}+\underline{\underline{\Gamma}}_{\mu\nu}^{\alpha}\label{eq:christoffel1}
\end{equation}
with
\begin{equation}
\underline{\Gamma}_{\mu\nu}^{\alpha}=\frac{1}{2}(\nabla_{\mu}h_{\nu}^{\alpha}+\nabla_{\nu}h_{\mu}^{\alpha}-\nabla^{\alpha}h_{\mu\nu})\label{eq:christoffel2}
\end{equation}
containing terms linear in $h_{\mu\nu}$ and
\begin{equation}
\underline{\underline{\Gamma}}_{\mu\nu}^{\alpha}=-\frac{1}{2}h^{\alpha\gamma}(\nabla_{\mu}h_{\gamma\nu}+\nabla_{\nu}h_{\mu\gamma}-\nabla_{\gamma}h_{\mu\nu})
\end{equation}
containing terms quadratic in $h_{\mu\nu}$. Following \cite{Velt},
we can, using eq. (\ref{eq:riemanntensor}), then decompose the Riemann
tensor in separate parts $R_{\nu\alpha\beta}^{\mu},\underline{R}_{\nu\alpha\beta}^{\mu},\underline{\underline{R}}_{\nu\alpha\beta}^{\mu}$,
each depending on $\Gamma_{\mu\nu}^{\alpha},\underline{\Gamma}_{\mu\nu}^{\alpha},$
or $\underline{\underline{\Gamma}}_{\mu\nu}^{\alpha}$: 
\begin{equation}
\overline{R}_{\nu\alpha\beta}^{\mu}=R_{\;\nu\alpha\beta}^{\mu}+\underline{R}_{\;\nu\alpha\beta}^{\mu}+\underline{\underline{R}}_{\;\nu\alpha\beta}^{\mu}
\end{equation}
\begin{equation}
\underline{R}_{\;\nu\alpha\beta}^{\mu}=\frac{1}{2}(\nabla_{\alpha}\nabla_{\nu}h_{\beta}^{\mu}-\nabla^{\mu}\nabla_{\alpha}h_{\nu\beta}-\nabla_{\beta}\nabla_{\nu}h_{\alpha}^{\mu}+\nabla^{\beta}\nabla_{\mu}h_{\nu\alpha})+\frac{1}{2}R_{\;\gamma\alpha\beta}^{\mu}h_{\nu}^{\gamma}+\frac{1}{2}R_{\;\nu\beta\alpha}^{\gamma}h_{\gamma}^{\mu}
\end{equation}

\begin{equation}
\underline{\underline{R}}_{\;\sigma\alpha\beta}^{\mu}=\partial_{\alpha}\underline{\underline{\Gamma}}_{\sigma\beta}^{\mu}-\partial_{\beta}\underline{\underline{\Gamma}}_{\sigma\alpha}^{\mu}+\underline{\Gamma}_{\nu\alpha}^{\text{\ensuremath{\mu}}}\underline{\Gamma}_{\sigma\beta}^{\nu}-\underline{\Gamma}_{\nu\beta}^{\mu}\underline{\Gamma}_{\sigma\alpha}^{\nu}
\end{equation}
one then can construct $\overline{R}_{\nu\alpha}=\overline{R}_{\nu\lambda\alpha}^{\lambda}$
and $\overline{R}=\overline{g}^{\mu\nu}\overline{R}_{\mu\nu}$. 

Veltman and `t Hooft find for the Lagrangian density up to second
order the expression:
\begin{eqnarray}
\mathcal{L} & = & \sqrt{-\overline{g}}\overline{R}\nonumber \\
 & = & \sqrt{-g}R+\underline{\mathcal{L}}+\underline{\underline{\mathcal{L}}}\label{eq:lagerangians}
\end{eqnarray}
where 
\begin{equation}
\underline{\mathcal{L}}=\sqrt{-g}(h_{\alpha}^{\beta}R_{\beta}^{\alpha}-\frac{1}{2}h_{\alpha}^{\alpha}R)\label{eq:linearpart}
\end{equation}
 is the part linear in $h_{\mu\nu}$ and 
\begin{eqnarray}
\underline{\underline{\mathcal{L}}} & =\sqrt{-g} & \left(-R\left(\frac{1}{8}(h_{\alpha}^{\alpha})^{2}-\frac{1}{4}h_{\beta}^{\alpha}h_{\alpha}^{\beta}\right)-h_{\beta}^{\nu}h_{\alpha}^{\text{\ensuremath{\beta}}}R_{\nu}^{\alpha}+\frac{1}{2}h_{\alpha}^{\alpha}h_{\beta}^{\nu}R_{\nu}^{\beta}-\frac{1}{4}\nabla_{\nu}h_{\alpha}^{\beta}\nabla^{\nu}h_{\beta}^{\alpha}\right.\nonumber \\
 &  & \left.+\frac{1}{4}\nabla_{\mu}h_{\alpha}^{\alpha}\nabla^{\mu}h_{\beta}^{\beta}-\frac{1}{2}\nabla_{\beta}h_{\alpha}^{\alpha}\nabla^{\mu}h_{\mu}^{\beta}+\frac{1}{2}\nabla^{\alpha}h_{\beta}^{\nu}\nabla_{\nu}h_{\alpha}^{\beta}\right)\label{eq:quadratic}
\end{eqnarray}
is part of the Lagrangian that is quadratic in $h_{\mu\nu}$ (with
total derivatives being ommitted).

We now choose the background $g_{\mu\nu}=\eta_{\mu\nu}$. The scalar
$R$ and the tensor $R_{\nu}^{\alpha}$ contain sums of derivatives
of the metric, which must vanish for $g_{\mu\nu}=\eta_{\mu\nu}$.
Therefore $R$ and $R_{\nu}^{\alpha}$ must vanish too. Furthermore
$-g=1$ and in the second line, the covariant derivatives become partial
ones. The result is: 
\begin{eqnarray}
\underline{\underline{\mathcal{L}}} & = & -\frac{1}{4}\partial_{\nu}h_{\alpha}^{\beta}\partial^{\nu}h_{\beta}^{\alpha}+\frac{1}{4}\partial_{\mu}h_{\alpha}^{\alpha}\partial^{\mu}h_{\beta}^{\beta}-\frac{1}{2}\partial_{\beta}h_{\alpha}^{\alpha}\partial^{\mu}h_{\mu}^{\beta}+\frac{1}{2}\partial^{\alpha}h_{\beta}^{\nu}\partial_{\nu}h_{\alpha}^{\beta}-\frac{1}{4}\partial_{\nu}h_{\alpha\beta}\partial^{\nu}h^{\alpha\beta}\nonumber \\
 & = & +\frac{1}{4}\partial_{\mu}h_{\alpha}^{\alpha}\partial^{\mu}h_{\beta}^{\beta}-\frac{1}{2}\partial^{\beta}h_{\alpha}^{\alpha}\partial^{\mu}h_{\beta\mu}+\frac{1}{2}\partial^{\nu}h_{\mu\nu}\partial_{\beta}h^{\mu\beta}\label{eq:quadratic-1}
\end{eqnarray}
where partial integration was used for the last term in the sum.

Historically, this was first obtained by Pauli and Fierz \cite{PauliFierz}.
Feynman noted in \cite{Feynman} that this expression does not have
a generalized inverse, and one must add a gauge breaking term, in
order to define a propagator. This was similar to the electromagnetic
field. However, Feynman's group also found that when computing loop
amplitudes, the theory with this propagator still failed to be unitary.
For one loop amplitudes, Feynman was able to find a method around
this problem. In a loop ``fictious particles'' or ``ghosts'' have
to be inserted in order to compensate overcounting. This happens in
the path integral because of the summation over fields that are physically
equivalent by a gauge transformation. In 1967, deWitt \cite{Dewitt2}
was able to generalize this procedure to all orders of perturbation
theory. Finally, Fadeev and Popov \cite{Fadeev} found a simple way
to describe this procedure, which we will outline here. This is merely
a summary of \cite{Kiefer}, p. 51. 

The path integral from which the Feynman diagrams of the S-matrix
are constructed is formally an expression 
\begin{equation}
Z=\int\mathcal{D}h_{\mu\nu}e^{iS}
\end{equation}
with $S=\int d^{4}x\sqrt{-\overline{g}}\overline{R}.$ We described
the relevant gauge transformations $\eta_{\mu}$ in eq. (\ref{eq:gaugetrans1}).
Now we choose a gauge constraint to fix a gauge $\eta_{\mu}$. This
corresponds to four conditions $G_{\alpha}(h_{\mu\nu})=0$. Then,
we define a functional 
\begin{equation}
\Delta_{G}^{-1}(h_{\mu\nu})=\int\mathcal{D}\eta_{\mu}\prod_{\alpha}\delta(G_{\alpha}(h_{\mu\nu}^{\eta}))\label{eq:fadeevpopovdet}
\end{equation}
where $h_{\mu\nu}^{\eta}$ is the gauge transformed metric and $\eta_{\mu}$
is the gauge transformation from eq. (\ref{eq:gaugetrans1}). For
compact gauge groups, the integration measure $\mathcal{D}\eta_{\mu}$
is left-invariant. Therefore, we have for the gravitational field
\begin{equation}
D(\eta_{\mu}'\eta_{\mu})=D\eta_{\mu},
\end{equation}
and it follows that 
\begin{eqnarray}
\Delta_{G}^{-1}(h_{\mu\nu}^{\eta}) & = & \int\mathcal{D}\eta_{\mu}\prod_{\alpha}\delta(G_{\alpha}(h_{\mu\nu}^{\eta'\eta}))\nonumber \\
 & = & \int\mathcal{D}\eta_{\mu}'\eta_{\mu}\prod_{\alpha}\delta(G_{\alpha}(h_{\mu\nu}^{\eta'\eta}))\nonumber \\
 & = & \Delta_{G}^{-1}(h_{\mu\nu})\label{eq:faddeeevpopov}
\end{eqnarray}
hence $\Delta_{G}^{-1}(h_{\mu\nu}^{\eta})$ is gauge invariant. This
can be inserted into the path integral as 
\begin{eqnarray}
Z & = & \int\mathcal{D}h_{\mu\nu}^{\eta}\Delta_{G}(h_{\mu\nu}^{\eta})\int\mathcal{D}\eta_{\mu}\prod_{\alpha}\delta(G_{\alpha}(h_{\mu\nu}^{\eta}))e^{iS}\nonumber \\
 & = & \int\mathcal{D}h_{\mu\nu}\Delta_{G}(h_{\mu\nu})\int\mathcal{D}\eta_{\mu}\prod_{\alpha}\delta(G_{\alpha}(h_{\mu\nu}))e^{iS}\nonumber \\
 & = & \int\mathcal{D}\eta_{\mu}\int\mathcal{D}h_{\mu\nu}\prod_{\alpha}\delta(G_{\alpha}(h_{\mu\nu}))\Delta_{G}(h_{\mu\nu})e^{iS}
\end{eqnarray}
where we have used the gauge invariance of $S$, $G_{\alpha}(h_{\mu\nu})$
and $\mathcal{D}h_{\mu\nu}^{\eta}$ . The factor $\int\mathcal{D}\eta_{\mu}$
only contributes a multiplicative constant, since the rest of the
integrand does not depend on $\eta$. Therefore, $\int\mathcal{D}\eta_{\mu}$
may be dropped. 

Because of the deltafunction in the integrand, we can expand $G_{\alpha}(h_{\mu\nu}^{\eta})$
in eq. (\ref{eq:fadeevpopovdet}) around $\eta_{\mu}=0$, since for
other gauges, the contribution of the integrand is zero. We get 
\begin{equation}
G_{\alpha}(h_{\mu\nu}^{\eta})=G_{\alpha}(h_{\mu\nu})+(A\eta)_{\alpha}
\end{equation}
where 
\begin{equation}
G_{\alpha}(h_{\mu\nu})=0
\end{equation}
from the gauge condition, and 
\begin{equation}
A_{\alpha\beta}=\frac{\delta G_{\alpha}(h_{\mu\nu}^{\eta})}{\delta\eta^{\beta}}
\end{equation}
is a matrix of the derivatives of $G_{\alpha}(h_{\mu\nu}^{\eta})$
with respect to $\eta^{\beta}$. 

For n-dimensional vectors $\vec{a}$ and some vectorfield $\vec{g}(\vec{a})$
we have the identity 
\begin{equation}
1=det\left(\frac{\partial g_{i}}{\partial a_{j}}\right)\prod_{i}^{n}\int da_{i}\delta^{n}(\vec{g}(\vec{a}))\label{eq:continuum}
\end{equation}
From eq. (\ref{eq:fadeevpopovdet}), we have 
\begin{equation}
1=\Delta_{G}(h_{\mu\nu})\int\mathcal{D}\eta_{\mu}\prod_{\alpha}\delta(G_{\alpha}(h_{\mu\nu}^{\eta}))
\end{equation}
which is the continuous analogue of eq. (\ref{eq:continuum}). Thereby,
:
\begin{equation}
\Delta_{G}(h_{\mu\nu})=det(A_{\alpha\beta}).
\end{equation}
This determinant is conventially expressed in terms of an integral
over anticommuting grassman variables $(\tilde{\eta}^{\alpha})^{*}$and
$\tilde{\eta}^{\beta}$ which are called ghost fields:
\begin{equation}
det(A)=\int\prod_{\alpha}\mathcal{D}(\tilde{\eta}^{\alpha})^{*}\mathcal{D\tilde{\eta}}^{\alpha}e^{-i\int d^{4}x(\tilde{\eta}^{\alpha})^{*}A_{\alpha\beta}\tilde{\eta}^{\beta}}\label{eq:fadeevpopovdeterminant}
\end{equation}
inserting this into the path integral, we get
\begin{eqnarray}
Z & = & \int\mathcal{D}h_{\mu\nu}\prod_{\alpha}\delta(G_{\alpha}(h_{\mu\nu}))\Delta_{G}(h_{\mu\nu})e^{iS}\nonumber \\
 & = & \int\mathcal{D}h_{\mu\nu}\int\prod_{\alpha}\mathcal{D}(\tilde{\eta}^{\alpha})^{*}\mathcal{D\tilde{\eta}}^{\alpha}\nonumber \\
 &  & \delta(G_{\alpha}(h_{\mu\nu}))e^{iS+i\int d^{4}x(\tilde{\eta}^{\alpha})^{*}A_{\alpha\beta}\tilde{\eta}^{\beta}}\label{eq:functional}
\end{eqnarray}

We can also choose $G_{\alpha}(h_{\mu\nu})=c_{\alpha}$ as gauge condition.
This leads to an integral 
\begin{eqnarray}
Z & = & \int\mathcal{D}h_{\mu\nu}\int\prod_{\alpha}D(\tilde{\eta}^{\alpha})^{*}\mathcal{D\tilde{\eta}}^{\alpha}\nonumber \\
 &  & \delta(G_{\alpha}(h_{\mu\nu})-c_{\alpha})e^{iS+i\int d^{4}x(\tilde{\eta}^{\alpha})^{*}A_{\alpha\beta}\tilde{\eta}^{\beta}}
\end{eqnarray}
Since $c_{\alpha}$ does not depend on $h_{\mu\nu}$, $\Delta_{G}(h_{\mu\nu})$
and also $Z$ are not changed by $c_{\alpha}$. Since $Z$ is independent
of $c_{\alpha},$ we may integrate over $c_{\alpha}$. Inserting the
normalization factor 
\[
\int\mathcal{D}c_{\alpha}e^{-\frac{i}{4\zeta}\int d^{4}xc_{\alpha}c^{\alpha}}
\]
we get 
\begin{eqnarray}
Z & = & \int\mathcal{D}c_{\alpha}e^{-\frac{i}{4\zeta}\int d^{4}xc_{\alpha}c^{\alpha}}\int\mathcal{D}h_{\mu\nu}\int\prod_{\alpha}D(\tilde{\eta}^{\alpha})^{*}\mathcal{D\tilde{\eta}}^{\alpha}\nonumber \\
 &  & \delta(G_{\alpha}(h_{\mu\nu})-c_{\alpha})e^{iS+i\int d^{4}x(\tilde{\eta}^{\alpha})^{*}A_{\alpha\beta}\tilde{\eta}^{\beta}}\nonumber \\
 & = & \int\mathcal{D}h_{\mu\nu}\int\prod_{\alpha}D(\tilde{\eta}^{\alpha})^{*}\mathcal{D\tilde{\eta}}^{\alpha}e^{iS-\frac{i}{4\zeta}\int d^{4}xG_{\alpha}G^{\alpha}+i\int d^{4}x(\tilde{\eta}^{\alpha})^{*}A_{\alpha\beta}\tilde{\eta}^{\beta}}.\label{eq:functional1}
\end{eqnarray}
The expression 
\[
\frac{i}{4\zeta}\int d^{4}xG_{\alpha}G^{\alpha}
\]
is called gauge fixing term. 

We can derive the propagator of the gravitational field, if we take
the part of $S=\int d^{4}x\sqrt{-\overline{g}}\overline{R}$ which
is quadratic in $h_{\mu\nu}$ together with some suitable gauge fixing.
Choosing, for example the DeDonder gauge, which is equal to 
\[
G_{\mu}=\partial_{\nu}h_{\mu\nu}-\frac{1}{2}\partial_{\mu}h_{\nu\nu}
\]
and subtracting $\frac{1}{2}G_{\mu}^{2}$from eq (\ref{eq:quadratic-1}),
we obtain, using in the 6.th line that $h_{\mu\nu}$ must be symmetric:
\begin{eqnarray}
\underline{\underline{\mathcal{L}}}-\frac{1}{2}C_{\mu}^{2} & = & -\frac{1}{4}\partial_{\nu}h_{\alpha\beta}\partial^{\nu}h^{\alpha\beta}+\frac{1}{4}\partial_{\mu}h_{\alpha}^{\alpha}\partial^{\mu}h_{\beta}^{\beta}-\frac{1}{2}\partial^{\beta}h_{\alpha}^{\alpha}\partial^{\mu}h_{\beta\mu}+\frac{1}{2}\partial^{\nu}h_{\mu\nu}\partial_{\beta}h^{\mu\beta}\nonumber \\
 &  & -\frac{1}{2}\partial_{\nu}h_{\mu\nu}\partial^{\beta}h^{\mu\beta}-\frac{1}{8}\partial_{\mu}h_{\nu\nu}\partial^{\mu}h_{\beta\beta}+\frac{1}{4}\partial_{\beta}h_{\alpha\alpha}\partial^{\mu}h^{\mu\beta}+\frac{1}{4}\partial^{\beta}h^{\alpha\alpha}\partial_{\mu}h_{\mu\beta}\nonumber \\
 & = & -\frac{1}{4}\partial_{\nu}h_{\alpha\beta}\partial^{\nu}h^{\alpha\beta}+\frac{1}{4}\partial_{\mu}h_{\alpha}^{\alpha}\partial^{\mu}h_{\beta}^{\beta}-\frac{1}{2}\partial^{\beta}h_{\alpha}^{\alpha}\partial^{\mu}h_{\beta\mu}+\frac{1}{2}\partial^{\nu}h_{\mu\nu}\partial_{\beta}h^{\mu\beta}\nonumber \\
 &  & -\frac{1}{2}\partial^{\nu}h_{\mu\nu}\partial_{\beta}h^{\mu\beta}-\frac{1}{8}\partial_{\mu}h_{\nu}^{\nu}\partial^{\mu}h_{\beta}^{\beta}+\frac{1}{2}\partial^{\beta}h_{\alpha}^{\alpha}\partial^{\mu}h_{\mu\beta}\nonumber \\
 & = & -\frac{1}{4}\partial_{\nu}h_{\alpha\beta}\partial^{\nu}h^{\alpha\beta}+\frac{1}{4}\partial_{\mu}h_{\alpha}^{\alpha}\partial^{\mu}h_{\beta}^{\beta}-\frac{1}{2}\partial^{\beta}h_{\alpha}^{\alpha}\partial^{\mu}h_{\beta\mu}+\frac{1}{2}\partial^{\nu}h_{\mu\nu}\partial_{\beta}h^{\mu\beta}\nonumber \\
 &  & -\frac{1}{2}\partial^{\nu}h_{\mu\nu}\partial_{\beta}h^{\mu\beta}-\frac{1}{8}\partial_{\mu}h_{\nu}^{\nu}\partial^{\mu}h_{\beta}^{\beta}+\frac{1}{2}\partial^{\beta}h_{\alpha}^{\alpha}\partial^{\mu}h_{\beta\mu}\nonumber \\
 & = & -\frac{1}{4}\partial_{\nu}h_{\alpha\beta}\partial^{\nu}h^{\alpha\beta}+\frac{1}{8}\partial_{\lambda}h_{\alpha}^{\alpha}\partial_{\lambda}h_{\beta}^{\beta}\nonumber \\
 & = & -\frac{1}{2}\partial_{\lambda}h_{\alpha\beta}V^{\alpha\beta\mu\nu}\partial^{\lambda}h_{\mu\nu}
\end{eqnarray}
with 
\begin{equation}
V_{\alpha\beta\mu\nu}=\frac{1}{2}\eta_{\alpha\mu}\eta_{\beta\nu}-\frac{1}{4}\eta_{\alpha\beta}\eta_{\mu\nu}
\end{equation}
Inversion of the matrix $V_{\alpha\beta\mu\nu}$ yields $\eta_{\mu\alpha}\eta_{\nu\beta}+\eta_{\mu\beta}\eta_{\nu\alpha}-\eta_{\mu\nu}\eta_{\alpha\beta}$
and leads to the graviton propagator

\begin{fmffile}{fmfdefault2abc}\raisebox{-0.5\height}{\begin{fmfgraph*}(70,40) \fmfright{f1}\fmfleft{f2} \fmf{photon,label=$k$}{f1,f2}  \fmfv{l.a=90,lab=$\mu\nu$}{f2}\fmfv{l.a=90,lab=$\alpha\beta$}{f1} \end{fmfgraph*}}$\;\;\;=D_{\mu\nu\alpha\beta}(k)=\frac{1}{k^{2}-i\epsilon}(\eta_{\mu\alpha}\eta_{\nu\beta}+\eta_{\mu\beta}\eta_{\nu\alpha}-\eta_{\mu\nu}\eta_{\alpha\beta})$ \end{fmffile}
Note that similar to the situation in electrodynamics, the form of
the graviton propagator depends crucially on the gauge. 

In their article \cite{thooft},  `t Hooft and Veltman also derive
the propagator for the Prentki gauge which we will not consider here.
The ghost Lagrangian is obtained by subjecting $C_{\mu}$ to the gauge
transformation of eq. (\ref{eq:gauge2}). We get, 
\begin{eqnarray}
G_{\mu} & = & \partial_{\nu}\left(h_{\mu\nu}+\nabla_{\mu}\eta_{\nu}+\nabla_{\nu}\eta_{\mu}\right)-\frac{1}{2}\partial_{\mu}\left(h_{\nu\nu}+\nabla_{\nu}\eta_{\nu}+\nabla_{\nu}\eta_{\nu}\right)\nonumber \\
 & = & \partial_{\nu}h_{\mu\nu}-\frac{1}{2}\partial_{\mu}h_{\nu\nu}+\partial_{\nu}\left(\nabla_{\mu}\eta_{\nu}+\nabla_{\nu}\eta_{\mu}\right)-\partial_{\mu}\nabla_{\nu}\eta_{\nu}\nonumber \\
 & = & \partial_{\nu}h_{\mu\nu}-\frac{1}{2}\partial_{\mu}h_{\nu\nu}+\partial_{\nu}\partial_{\mu}\eta_{\nu}-\partial_{\nu}\Gamma_{\mu\nu}^{\lambda}\eta_{\lambda}+\partial_{\nu}\partial_{\nu}\eta_{\mu}-\partial_{\nu}\Gamma_{\nu\mu}^{\lambda}\eta_{\lambda}-\left(\partial_{\mu}\partial_{\nu}\eta_{\nu}-\partial_{\mu}\Gamma_{\nu\nu}^{\lambda}\eta_{\lambda}\right)\nonumber \\
 & = & \partial_{\nu}h_{\mu\nu}-\frac{1}{2}\partial_{\mu}h_{\nu\nu}+\partial_{\nu}\partial_{\nu}\eta_{\mu}-\partial_{\nu}\Gamma_{\mu\nu}^{\lambda}\eta_{\lambda}-\partial_{\nu}\Gamma_{\nu\mu}^{\lambda}\eta_{\lambda}+\partial_{\mu}\Gamma_{\nu\nu}^{\lambda}\eta_{\lambda}\label{eq:ghost-1}
\end{eqnarray}
We note that in the expansion of the Christoffel symbols, eq. (\ref{eq:christoffel1}),
we have $\Gamma_{\mu\nu}^{\alpha}=0$ with $g_{\mu\nu}=\eta_{\mu\nu}$,
and the part $\underline{\Gamma}_{\mu\nu}^{\alpha}$ depends linearly
on $h_{\mu\nu}$ . Taking the derivative of $G_{\mu}$ with respect
to $\eta_{\mu}$, we observe that the part of $G_{\mu}$ which does
not include any order of $h_{\mu\nu}$ is $\partial_{\nu}^{2}$. Therefore,
the lowest order ghost Lagrangian is 
\begin{equation}
\mathcal{L}_{ghost}=-\partial_{\nu}\tilde{\eta}_{\alpha}\partial_{\nu}\eta^{\alpha}+\mathcal{O}(h_{\mu\nu})=-\partial_{\nu}\tilde{\eta}_{\alpha}\eta_{\mu\nu}\partial^{\mu}\eta^{\alpha}+\mathcal{O}(h_{\mu\nu})
\end{equation}
 In diagrammatic language, this corresponds to a ghost particle with
a propagator given by, see \cite{Hamber}:

\begin{fmffile}{fmfdefault2ab}\begin{fmfgraph*}(70,25) \fmfright{f1}\fmfleft{f2} \fmf{fermion,label=$k$}{f1,f2}  \end{fmfgraph*}$=D_{\mu\nu}(k)=\frac{\eta_{\mu\nu}}{k^{2}-i\epsilon}$ \end{fmffile} 

The graviton couples to itself. The three graviton vertex was first
given in the third article on quantum gravity by deWitt \cite{DeWitt3}.
It is a quite complicated expression. Nevertheless, deWitt notes that
it could be ``straightforwardly'' computed from the cubic part of
the action. I have however, not tried to derive it and I merely state
this result here. In his article, deWitt also describes the four graviton
vertex coming from the quartic part of the action, which I do not
present here, simply because of the length of that formula. We have 

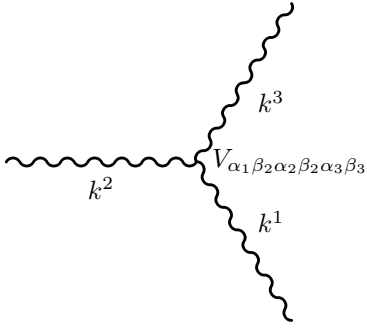
\begin{fmffile}{fmfdefault1}       \begin{fmfgraph*}(120,120)    \fmfleft{i1}        \fmfright{o1,o2}        \fmf{photon,label=$k^{2}$}{i1,w1}        \fmf{photon,label=$k^{1}$}{w1,o1}        \fmf{photon,label=$k^{3}$}{w1,o2}        \fmfv{lab=$V_{\alpha_{1}\beta_{2}\alpha_{2}\beta_{2}\alpha_{3}\beta_{3}}$,lab.dist=0.05w}{w1}   
\end{fmfgraph*}  \end{fmffile} 
\begin{eqnarray}
V(k^{1},k^{2},k^{3})_{\alpha_{1}\beta_{1}\alpha_{2}\beta_{2}\alpha_{3}\beta_{3}} & = & -\left(k_{(\alpha_{1}}^{2}k_{\beta_{1})}^{3}\left(2\eta_{\alpha_{2}(\alpha_{3}}\eta_{\beta_{3})\beta_{2}}-\eta_{\alpha_{2}\beta_{2}}\eta_{\alpha_{3}\beta_{3}}\right)\right.\nonumber \\
 &  & +k_{(\alpha_{2}}^{1}k_{\beta_{2})}^{3}\left(2\eta_{\alpha_{1}(\alpha_{3}}\eta_{\beta_{3})\beta_{1}}-\eta_{\alpha_{1}\beta_{1}}\eta_{\alpha_{3}\beta_{2}}\right)\nonumber \\
 &  & +k_{(\alpha{}_{3}}^{1}k_{\beta_{3})}^{2}\left(2\eta_{\alpha_{1}(\alpha_{2}}\eta_{\beta_{2})\beta_{1}}-\eta_{\alpha_{1}\beta_{1}}\eta_{\alpha_{2}\beta_{2}}\right)\nonumber \\
 &  & +2k_{(\alpha_{2}}^{3}\eta_{\beta_{2})}(_{\alpha_{1}}\eta_{\beta_{1}})(_{\alpha_{3}}k_{\beta_{3}}^{2})+2q_{(\alpha_{3}}^{1}\eta_{\beta_{3})(\alpha_{2}}\eta_{\beta_{2})(\alpha_{1}}k_{\beta_{1})}^{3}\nonumber \\
 &  & +2k_{(\alpha_{1}}^{2}\eta_{\beta_{1})(\alpha_{3}}\eta_{\beta_{3})(\alpha_{2}}k_{\beta_{2})}^{1}\nonumber \\
 &  & +k^{2}k^{3}\left(\eta_{\alpha_{1}(\alpha_{2}}\eta_{\beta_{2})\beta_{1}}\eta_{\alpha_{3}\beta_{3}}+\eta_{\alpha_{1}(\alpha_{3}}\eta_{\beta_{3})\beta_{1}}\eta_{\alpha_{2}\beta_{2}}-2\eta_{\alpha_{1}(\alpha_{2}}\eta_{\beta_{2})(\alpha_{3}}\eta_{\beta_{3})\beta_{1}}\right)\nonumber \\
 &  & +k^{1}k^{3}\left(\eta_{\alpha_{2}(\alpha_{1}}\eta_{\beta_{1})\beta_{2}}\eta_{\alpha_{3}\beta_{3}}+\eta_{\alpha_{2}(\alpha_{3}}\eta_{\beta_{3})\beta_{2}}\eta_{\alpha_{1}\beta_{1}}-2\eta_{\alpha_{2}(\alpha_{1}}\eta_{\beta_{2})(\alpha_{3}}\eta_{\beta_{3})\beta_{2}}\right)\nonumber \\
 &  & +k^{1}k^{2}\left(\eta_{\alpha_{3}(\alpha_{1}}\eta_{\beta_{1})\beta_{3}}\eta_{\alpha_{2}\beta_{2}}+\eta_{\alpha_{3}(\alpha_{2}}\eta_{\beta_{2})\beta_{3}}\eta_{\alpha_{1}\beta_{1}}\right.\nonumber \\
 &  & \left.\left.-2\eta_{\alpha_{3}(\alpha_{1}}\eta_{\beta_{1})(\alpha_{2}}\eta_{\beta_{2})\beta_{3}}\right)\right)\label{eq:vertex2}
\end{eqnarray}
where have choosen the notation of \cite{Hamber}.

Furthermore, the graviton couples to the ghosts with the vertex $V(k^{1},k^{2},k^{3})_{\alpha\beta\gamma\mu}$
that is given by

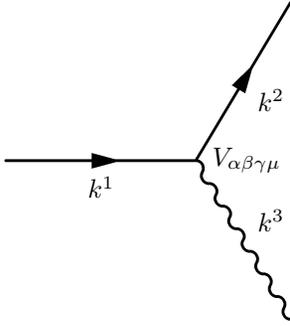
\begin{fmffile}{fmfdefault1v}      \begin{fmfgraph*}(120,120)  \fmfleft{i1}       \fmfright{o1,o2}        \fmf{fermion,label=$k^{1}$}{i1,w1}        	\fmf{photon,label=$k^{3}$}{w1,o1}        	\fmf{fermion,label=$k^{2}$}{w1,o2}          \fmfv{lab=$V_{\alpha\beta\gamma\mu}$,lab.dist=0.05w}{w1}   \end{fmfgraph*} \end{fmffile}It
should be derivable from eq. (\ref{eq:ghost-1}) with eq. (\ref{eq:christoffel2})
and is, see \cite{Hamber}, eq. (70), or \cite{popov}: 
\begin{equation}
V(k^{1},k^{2},k^{3})_{\alpha,\beta,\gamma\mu}=-\eta_{\gamma(\alpha}k_{\beta)}^{1}k_{\mu}^{2}+\eta_{\gamma\mu}k_{(\alpha}^{2}k_{\beta)}^{3}\label{eq:ghostvertex}
\end{equation}

Gravity does not only couple to itself but also to ordinary matter.
For example, we can add to the Lagrangian density $\mathcal{L}=\sqrt{-\overline{g}}\overline{R}$
a massive scalar field with a Lagrangian density 
\begin{equation}
\mathcal{L}_{m}=\sqrt{-\overline{g}}\left(-\frac{1}{2}\nabla_{\mu}\phi\nabla^{\mu}\phi-\frac{1}{2}m^{2}\phi^{2}\right)
\end{equation}
and a well known propagator:

\begin{fmffile}{fmfdefault2a}\raisebox{-0.5\height}{\begin{fmfgraph*}(70,25) \fmfright{f1}\fmfleft{f2} \fmf{fermion,label=$k$}{f1,f2}  \end{fmfgraph*}}$\;\;=\frac{1}{(2\pi)^{4}i}\frac{1}{k^{2}+m^{2}-i\epsilon}$ \end{fmffile}

Using $\nabla_{\mu}\phi=\partial_{\mu}\phi$, $\nabla^{\mu}\phi=\overline{g}^{\mu\nu}\partial_{\nu}\phi$
and eqs. (\ref{eq:metric1})-(\ref{eq:determinant1}) with $g_{\mu\nu}=\eta_{\mu\nu}$,
we get, neglecting terms of quadratic or higher order 
\begin{eqnarray}
\mathcal{L}_{m} & = & \sqrt{-\overline{g}}\left(-\frac{1}{2}\partial_{\mu}\phi\overline{g}^{\mu\nu}\partial_{\nu}\phi-\frac{1}{2}m^{2}\phi^{2}\right)\nonumber \\
 & = & (1+\frac{1}{2}h_{\alpha}^{\alpha})\left(-\frac{1}{2}\partial_{\mu}\phi\left(\eta^{\mu\nu}-h^{\mu\nu}\right)\partial_{\nu}\phi-\frac{1}{2}m^{2}\phi^{2}\right)\nonumber \\
 & = & (1+\frac{1}{2}h_{\alpha}^{\alpha})\left(-\frac{1}{2}\partial_{\mu}\phi\eta^{\mu\nu}\partial_{\nu}\phi+\frac{1}{2}\partial_{\mu}\phi h^{\mu\nu}\partial_{\nu}\phi-\frac{1}{2}m^{2}\phi^{2}\right)\nonumber \\
 & \approx & -\frac{1}{2}\partial_{\mu}\phi\eta^{\mu\nu}\partial_{\nu}\phi+\frac{1}{2}\partial_{\mu}\phi h^{\mu\nu}\partial_{\nu}\phi-\frac{1}{2}m^{2}\phi^{2}\nonumber \\
 &  & -\frac{1}{4}h_{\alpha}^{\alpha}\partial_{\mu}\phi\eta^{\mu\nu}\partial_{\nu}\phi+\frac{1}{4}h_{\alpha}^{\alpha}\partial_{\mu}\phi h^{\mu\nu}\partial_{\nu}\phi-\frac{1}{4}h_{\alpha}^{\alpha}m^{2}\phi^{2}\nonumber \\
 & \approx & -\frac{1}{2}\partial_{\mu}\phi\partial^{\mu}\phi-\frac{1}{2}m^{2}\phi^{2}-\frac{1}{4}h_{\alpha}^{\alpha}(\partial_{\mu}\phi\eta^{\mu\nu}\partial_{\nu}\phi+m^{2}\phi^{2})+\frac{1}{2}\partial_{\mu}\phi h^{\mu\nu}\partial_{\nu}\phi\nonumber \\
 & = & -\frac{1}{2}\partial_{\mu}\phi\partial^{\mu}\phi-\frac{1}{2}m^{2}\phi^{2}\nonumber \\
 &  & -\frac{1}{2}h^{\mu\nu}\left(\frac{1}{2}\eta_{\mu\nu}\left(\partial_{\mu}\phi\partial^{\mu}\phi+m^{2}\phi^{2}\right)-\partial_{\mu}\phi\partial_{\nu}\phi\right)
\end{eqnarray}
 The vertex function between the scalar field and the graviton is
then in momentum space:

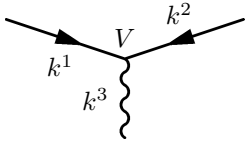
\begin{fmffile}{fmfdefault3a}  
\begin{fmfgraph*}(90,90)    
\fmfleft{i1}       
\fmfright{o2}  
\fmfbottom{o1}        
\fmf{fermion,label=$k^{1}$}{i1,w1}        
\fmf{photon,label=$k^{3}$}{w1,o1}        
\fmf{fermion,label=$k^{2}$}{o2,w1}        
\fmfv{l.a=90, lab=$V$,lab.dist=0.05w}{w1}   
\end{fmfgraph*} 
\end{fmffile} 
\begin{equation}
V(k^{1},k^{2},k^{3})_{\mu\nu}=(2\pi^{4})\left(\frac{1}{2}\eta_{\mu\nu}(k^{1}k^{2})-\frac{1}{2}\eta_{\mu\nu}m^{2}-k_{\mu}^{1}k_{\nu}^{2}\right)\label{eq:vertex1}
\end{equation}

Having the Feynman rules defined, one can use them to describe scattering
processes. In his article, deWitt reports the result of the cross
secction for gravitational scattering of scalar particles in the center
of mass frame. deWitt gets, with $v=p/E$ 
\begin{eqnarray}
\frac{d\sigma}{d\Omega} & \propto & \left(\frac{(1+3v^{2})(1-v^{2})+4v^{2}(1+v^{3})cos^{2}(\theta/2)}{v^{2}sin^{2}(\theta/2)}\right.\nonumber \\
 &  & +\frac{(1+3v^{2})(1-v^{2})+4v^{2}(1+v^{2})sin^{2}(\theta/2)}{v^{2}cos^{2}(\theta/2)}\nonumber \\
 &  & +\left.(3-v^{2})(1+v^{2})+2v^{2}sin^{2}\theta\right)^{2}
\end{eqnarray}

From the vertices of eqs. (\ref{eq:vertex1}) and (\ref{eq:vertex2})
it also follows that the theory is non-renormalizable. By counting
momentum powers one observes that the superficial degree of divergence
is 
\begin{equation}
D=-2I+2V+4L
\end{equation}
where $I$ is the number of internal lines, $V$ is the number of
vertices, and $L$ the number of loops. Because the graviton propagators
are $\propto1/k^{2}$, each graviton propagator lowers the degree
of divergence by 2. On the other hand, the vertices have terms $\propto k^{\mu}k^{\nu}$
which implies that each vertex adds 2 to the degree of divergence.
However, the number of loops is 
\begin{equation}
L=1+I-V
\end{equation}
 or 
\begin{equation}
2(-L+1)=-2I+2V
\end{equation}
 that is 
\begin{equation}
D=2-2L+4L=2+2L=2(L+1)
\end{equation}
which is independent of the number of external lines and increases
with every loop. This estimation should, however, not be taken as
a definite result. It may be that because of some symmetry principle,
or an invariant, divergences cancel and the theory is still finite.
In the next sub section, we will see that in fact this happens at
the first loop level for pure gravity. Unfortunately, the divergences
appear on the first loop level if matter fields are coupled to gravity.
This stimulated the hope that with the addition of the right matter
fields, e.g within a supersymmetric theory, one could get a finite
theory of quantum gravity for all orders. Unfortunately, the divergences
appear on the two loop level even for pure gravitation.

\subsection{Loops and divergencies in the framework of covariant quantization
of gravity }

The investigation of loops in the covariant quantization of gravity
began with early computations by Feynman \cite{Feynman} and deWitt
\cite{Dewitt2}, but the complete understanding of perturbative quantum
gravity at one loop level was only achived by the articles of `t Hoofl
and Veltman. \cite{thooft,Velt}. By computeing the counterterm for
all divergent one loop graviton amplitudes, they found that pure gravity
is actually one loop finite. In the following, we will sketch their
derivation. 

In his famous article \cite{Dewitt2} deWitt introduced the so called
background field method. The introduction of the background field
method given here follows the excellent articles \cite{Background1,Background2,Background3}.
A usual generating functional of a non gauge field theory with quantum
field $\phi$ is 
\begin{equation}
Z(J)=\int\mathcal{D}\phi e^{i\int d^{4}x\mathcal{L}(\phi)+J\phi}\label{eq:generating}
\end{equation}
where $J$ is called source and $J\phi=\int d^{4}xJ\phi$. The Green's
functions are defined by 
\[
\frac{1}{i^{n}}\frac{\delta^{n}Z(J)}{\delta^{n}J(x_{1})\ldots J(x_{n})}|_{J=0}=\langle0|T(\phi(x_{1}),\ldots\phi(x_{n}))|0\rangle
\]
For a non-gauge theory, it is found that the two point Green's functions,
i.e the propagators, are simply the inverse of the part of $\mathcal{L}$
that is quadratic in $\phi$. For a gauge theory, the gauge invariance
has to be broken by the addition of a gauge fixing term and a ghost
Lagrangian, before one can define a propagator as a generalized inverse
of the quadratic part of $\mathcal{L}$, as we saw in the sub section
above. In general, the Green's functions obtained from $Z$ contain
connected graphs like

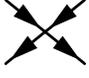
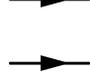
\begin{fmffile}{fmfdefault5a} \begin{fmfgraph*}(30,30)   \fmftop{i1,i2}       \fmfbottom{o1,o2}       \fmf{fermion}{i1,w1,o2}       \fmf{fermion}{i2,w1,o1}        \end{fmfgraph*} \end{fmffile}
as well as disconnected graphs like \begin{fmffile}{fmfdefault6a} \begin{fmfgraph*}(30,30)   \fmftop{i1,i2}       \fmfbottom{o1,o2}       \fmf{fermion}{i1,i2}       \fmf{fermion}{o1,o2}            \end{fmfgraph*} \end{fmffile} 

Disconnected graphs do not contribute to the s-matrix. The so called
energy-functional 
\begin{equation}
W(J)=-i\ln(Z(J))
\end{equation}
generates only connected graphs. A diagram is called one particle
irreducible (1PI), if it can be split into two disjoined pieces by
cutting a single internal line. It is simpler to compute the entire
collection of Feynman diagramms by calculating first 1PI diagrams
and then connecting these graphs together. The 1PI graphs are generated
by the so called effective action, it is the Legendre transform 
\begin{equation}
\Gamma(\hat{\phi})=W(J)-J\hat{\phi}
\end{equation}
 where 
\begin{equation}
\hat{\phi}=\frac{\delta W(J)}{\delta J}
\end{equation}
In the background field method, instead of using eq. (\ref{eq:generating})
one uses 
\begin{equation}
\tilde{Z}(J,\phi_{B})=\int\mathcal{D}\phi e^{i\int d^{4}x\mathcal{L}(\phi+\phi_{B})+J\phi}\label{eq:generating-1}
\end{equation}
where we have split the original field into $\overline{\phi}=\phi_{B}+\phi$.
Only the part $\phi$, which is called the quantum part of the field,
appears in the measure of the path integral and is coupled to the
source. The part $\phi_{B}$ is called the background field, for which
the classical equations of motion 
\[
\frac{\delta\mathcal{L}(\overline{\phi})}{\delta\overline{\phi}}|_{\overline{\phi}=\phi_{B}}=0
\]
are assumed to hold. One defines the so called Background energy functional
\[
\tilde{W}(J,\phi_{B})=-i\ln(\tilde{Z}(J,\phi_{B}))
\]
and with 
\begin{equation}
\tilde{\phi}=\frac{\delta\tilde{W}(J,\phi_{B})}{\delta J}\label{eq:differentiate}
\end{equation}
we define the background effective action
\begin{equation}
\tilde{\Gamma}(\tilde{\phi},\phi_{B})=\tilde{W}(J,\phi_{B})-J\tilde{\phi}\label{eq:gdsfgd}
\end{equation}
By the variable shift $\phi\rightarrow\phi-\phi_{B}$ in eq (\ref{eq:generating-1}),
we observe that 
\begin{equation}
\tilde{Z}(J,\phi_{B})=Z(J)e^{-iJ\phi_{B}}
\end{equation}
and 
\begin{equation}
\tilde{W}(J,\phi_{B})=W(J)-J\phi_{B}\label{eq:dsjfdsf}
\end{equation}
Using eq. (\ref{eq:differentiate}), we get 
\begin{equation}
\tilde{\phi}=\frac{\delta\tilde{W}(J,\phi_{B})}{\delta J}=\frac{\delta W(J)}{\delta J}-\phi_{B}=\hat{\phi}-\phi_{B}\label{eq:dsdfsf}
\end{equation}
With help of eqs. (\ref{eq:dsjfdsf}) and (\ref{eq:dsdfsf}), we get
from eq.(\ref{eq:gdsfgd}) 
\begin{eqnarray}
\tilde{\Gamma}(\tilde{\phi},\phi_{B}) & = & W(J)-J\phi_{B}-J\tilde{\phi}\nonumber \\
 & = & W(J)-J\hat{\phi}\nonumber \\
 & = & \Gamma(\hat{\phi})\nonumber \\
 & = & \Gamma(\tilde{\phi}+\phi_{B})\label{eq:dsfsdgf}
\end{eqnarray}
eq. (\ref{eq:dsfsdgf}) implies that the background effective action
$\tilde{\Gamma}(\tilde{\phi},\phi_{B})$ is equal to the ordinary
effective action in the presence of a background field $\phi_{B}$.
Therefore $\tilde{\Gamma}(\tilde{\phi},\phi_{B})$ generates the same
1PI Diagrams than the conventional approach. Hence, the background
field method therefore produces the same S-Matrix as usual field theory. 

As in eq. (\ref{eq:metric}), we define the metric as 
\begin{equation}
\overline{g}_{\mu\nu}=g_{\mu\nu}+h_{\mu\nu}
\end{equation}
 with the arbitrary background field $g_{\mu\nu}$ and the quantum
field $h_{\mu\nu}$. Using the expression for the path integral, eq.
(\ref{eq:functional1}), with this metric and a covariant source term
$J^{\mu\nu}$, we can write the generating functional for the gravitational
field as 
\begin{equation}
Z(J^{\mu\nu})=\int\mathcal{D}h_{\mu\nu}\int\prod_{\alpha}D(\tilde{\eta}^{\alpha})^{*}\mathcal{D\tilde{\eta}}^{\alpha}e^{iS-\frac{i}{4\zeta}\int G_{\alpha}G^{\alpha}+i\int d^{4}x(\tilde{\eta}^{\alpha})^{*}A_{\alpha\beta}\tilde{\eta}^{\beta}+J^{\mu\nu}h_{\mu\nu}}
\end{equation}
where $\overline{R}$ and $\overline{g}$ are constructed from the
metric in eq. (\ref{eq:metric}). In his famous second article \cite{DeWitt3}
on quantum gravity, deWitt has shown this action functional is independent
of the gauge fixing term, provided that the gauge fixing term is background
invariant and that the classical equations of motion hold for the
background field. It is for this reason, why the background method
can be applied to gauge theories.

We can now expand the Lagrangian as in eq. (\ref{eq:lagerangians}).
The part of the Lagrangian that is linear in $h_{\mu\nu}$ can be
written as 
\begin{equation}
\underline{\mathcal{L}}_{B}=\sqrt{-g}(-\frac{1}{2}h_{\alpha}^{\alpha}R+h_{\alpha}^{\beta}R_{\beta}^{\alpha})=h_{\mu\nu}\sqrt{-g}(R^{\mu\nu}-\frac{1}{2}g^{\mu\nu}R)\label{eq:dgtfdfvgffcwfrsy}
\end{equation}
In $\underline{\mathcal{L}}_{B}$,  $R$ and $R^{\mu\nu}$are constructed
with the background field $g_{\mu\nu}$ for which the classical equations
of motion are assumed to hold. For $g_{\mu\nu}$ these are are the
Einstein equations in vacuum 
\begin{equation}
R^{\mu\nu}-\frac{1}{2}g^{\mu\nu}R=0
\end{equation}
and therefore eq. (\ref{eq:dgtfdfvgffcwfrsy}) implies that $\underline{\mathcal{L}}_{B}$
must vanish. The Lagrangian in eq. (\ref{eq:lagerangians}) is invariant
with respect to the following gauge transformation 
\begin{equation}
h_{\mu\nu}\rightarrow h_{\mu\nu}+(g_{\alpha\nu}+h_{\alpha\nu})\nabla_{\mu}\eta^{\alpha}+(g_{\mu\alpha}+h_{\mu\alpha})\nabla_{\nu}\eta^{\alpha}+\eta^{\alpha}\nabla_{\alpha}h_{\mu\nu}\label{eq:gaugetrans}
\end{equation}
This is a gauge invariance of the quantum field $h_{\mu\nu}$. Of
cource the Lagrangian is also a gauge invariant with respect to a
gauge transformation of the background field, but for now this invariance
will not interest us further, and we use a gauge fixing term.that
only breakes the gauge invariance with respect to $h_{\mu\nu}$. `t
Hooft and Veltman have chosen 
\begin{equation}
C_{\mu}=\sqrt[4]{-g}(\nabla_{\nu}h_{\mu}^{\nu}-\frac{1}{2}\nabla_{\mu}h_{\nu}^{\nu})t^{\mu\alpha}
\end{equation}
where $t^{\mu\alpha}t^{\alpha\nu}=g^{\mu\nu}$. Subjecting this gauge
fixing term to the gauge transformation of eq. (\ref{eq:gaugetrans}),
`t Hooft and Veltman find a ghost Lagrangian : 
\begin{equation}
\mathcal{L}_{g}=\sqrt{-g}\eta_{\mu}^{*}(g_{\mu\nu}\square-R_{\mu\nu})\eta^{\nu}
\end{equation}
where terms containing $h{}_{\mu\nu}$ have been omitted. For the
full ghost Lagrangian, see Goroff and Sagnotti \cite{Goroff2} on
p 215.

As the gauge invariance with respect to the background field is not
broken, the entire generating functional in the background field method
is gauge invariant with respect to $g_{\mu\nu}$, Therefore, the counterterms
must be solely composed of expressions that are invariant with respect
to gauge transformations of the $g_{\mu\nu}$. The counter Lagrangian
$\Delta\mathcal{L}$ must be dimensionless. From the integration $\int dx^{4}$
at the one loop level, we deduce that $\Delta\mathcal{L}$ must have
dimension $dx^{-4}$ in order for $\int d^{4}x\Delta\mathcal{L}$
to be dimensionless. This is achived by invariant scalars involving
4 derivatives. Hence the counterterm must be a linear combination
of invariant quadratic functions of the curvature. It can therefore
only be composed of terms like 
\begin{equation}
\Delta\mathcal{L}=\alpha R_{\mu\nu\alpha\beta}^{2}+\beta R_{\mu\nu}^{2}+\gamma R^{2}
\end{equation}
the term $R_{\mu\nu\alpha\beta}^{2}$ needs not to be considered since
one can show that 
\begin{equation}
R_{\mu\nu\alpha\beta}^{2}-4R_{\mu\nu}^{2}+R^{2}
\end{equation}
is a total derivative in 4 dimensions, see Appendix B of \cite{thooft}
for a proof. Therefore, we end up with a counter Lagrangian 
\begin{equation}
\Delta\mathcal{L}=\beta R_{\mu\nu}^{2}+\gamma R^{2}
\end{equation}
In \cite{thooft,Velt}, `t Hooft and Veltman succeeded in computing
the coefficients $\beta,\gamma$. First, they showed that to a Lagrangian
\begin{equation}
\mathcal{L}=\sqrt{-g}\left(-\partial_{\mu}\phi^{*}g^{\mu\nu}\partial_{\nu}\phi+2\phi_{i}^{*}N^{\mu}\partial_{\mu}\phi+\phi^{*}M\phi\right)\label{eq:masterformula}
\end{equation}
where $N^{\mu}$and $M$ are c-number functions of spacetime, there
corresponds a counterterm 
\begin{eqnarray}
\mathcal{\Delta L} & = & \frac{\sqrt{-g}}{2\epsilon}tr\left(\frac{1}{12}Y^{\mu\nu}Y_{\mu\nu}+\frac{1}{2}\left(M-N^{\mu}N_{\mu}-\nabla_{\mu}N^{\mu}-\frac{1}{6}R\right)^{2}\right.\label{eq:counterterm1}\\
 &  & \left.+\frac{1}{60}(R_{\mu\nu}R^{\mu\nu}-\frac{1}{3}R^{2})\right)
\end{eqnarray}
where 
\begin{equation}
Y_{\mu\nu}=\nabla_{\mu}N_{\nu}-\nabla_{\nu}N_{\mu}+[N_{\mu},N_{\nu}]
\end{equation}
Then, `t Hooft and Veltman consider the Lagrangian of gravity coupled
to a massless scalar field 
\begin{equation}
\mathcal{L}=\sqrt{-\overline{g}}(\overline{R}-\frac{1}{2}\partial_{\mu}\overline{\phi}g^{\mu\nu}\partial_{\nu}\overline{\phi})
\end{equation}
where $\overline{\phi}=\phi_{B}+\phi$ with $\phi_{B}$ as a background
field. After expanding this Lagrangian to second order in the quantum
fields $h_{\mu\nu}$and $\phi_{B}$ and finding the appropriate gauge
fixing terms, `t Hooft and Veltmann bring their Lagrangian into a
similar form as eq. (\ref{eq:masterformula}). After some arithmetic
manipulation, they get a counterterm 
\begin{eqnarray}
\Delta\mathcal{L} & = & \frac{\sqrt{-g}}{\epsilon}\left(\frac{9}{720}R^{2}+\frac{43}{120}R_{\alpha\beta}R^{\alpha\beta}+\frac{1}{2}(\partial_{\mu}\phi_{B}g^{\mu\nu}\phi_{B})^{2}\right.\label{eq:counterterm}\\
 &  & \left.-\frac{1}{12}R(\partial_{\mu}\phi_{B}g^{\mu\nu}\partial_{\nu}\phi_{B})+2(\nabla_{\mu}\nabla^{\mu}\phi_{B})^{2}\right)\nonumber 
\end{eqnarray}
this counterterm still contains contributions from closed loops of
$\phi$ particles. The contribution of these closed loops can be obtained
from eq (\ref{eq:counterterm1}) by setting $M=N=0$ and is
\begin{equation}
\frac{1}{2}\sqrt{-g}\left(\frac{1}{144}R^{2}+\frac{1}{120}(R_{\mu\nu}R^{\mu\nu})-\frac{1}{3}R^{2}\right)
\end{equation}
Subtracting this contribution and setting $\phi_{B}=0$ gives the
counterterm for pure gravitation 
\begin{equation}
\mathcal{\Delta L}=\sqrt{-g}\left(\frac{1}{120}R^{2}+\frac{7}{20}R_{\mu\nu}R^{\mu\nu}\right)\label{eq:counterterm1-1}
\end{equation}
The equations of motion for the background field can be inferred from
setting $\underline{\mathcal{L}}$ to zero, similar as in eq. (\ref{eq:dgtfdfvgffcwfrsy}).
From these equations, one gets 
\begin{equation}
\nabla_{\mu}\nabla^{\mu}\phi_{B}=0,\label{eq:formel12}
\end{equation}
 
\begin{equation}
R_{\mu\nu}=-\frac{1}{2}\nabla_{\mu}\phi_{B}\nabla_{\nu}\phi_{B},\label{eq:form2}
\end{equation}
and 
\begin{equation}
R=-\frac{1}{2}\nabla_{\mu}\phi_{B}\nabla^{\mu}\phi_{B}.\label{eq:form1}
\end{equation}
For pure gravitation, we have $\phi_{B}=0$ and therefore $R=R_{\mu\nu}=0$
and the counterterm vanishes. On the other hand, inserting eqs (\ref{eq:formel12}),
(\ref{eq:form2}) and (\ref{eq:form1}) into eq. (\ref{eq:counterterm}),
gives 
\begin{equation}
\Delta\mathcal{L}=\frac{203\sqrt{-g}}{80\epsilon}R^{2}.
\end{equation}
In case of pute gravity, the counterterm vanishes on shell if the
equations of motion are fulfilled. It therefore can be written as
\begin{equation}
\Delta\mathcal{L}=\frac{1}{\epsilon}F(g_{\mu\nu})\frac{\mathcal{\delta L}}{\delta g_{\mu\nu}}
\end{equation}
where $F$ is a covariant combination of $g_{\mu\nu}$ such that $\Delta\mathcal{L}$
consists of invariant quadratic functions of curvature. This actually
corresponds to a to a simple field redefinition 
\begin{equation}
\mathcal{L}(g_{\mu\nu}+\frac{1}{\epsilon}F(g_{\mu\nu}))=\mathcal{L}(g_{\mu\nu})+\frac{1}{\epsilon}F(g_{\mu\nu})\frac{\mathcal{\delta L}}{\delta g_{\mu\nu}}
\end{equation}
Such a redefinition of the quantum fields can not have any observable
effect, as it does not lead to different Feynman diagrams. Therefore,
the vanishing of $\Delta\mathcal{L}$ in case of pure gravity implies
that there are no divergences at the one loop level. With the inclusion
of matter, the scalar $R^{2}$ appears in the countertern that does
not vanish. Moreover, a term like $R^{2}$ is absent in the original
Lagrangian and therefore, the theory of gravity coupled to a massless
scalar has nonrenormalizable divergences at the one loop level. 

One may ask whether the massless scalar is just the wrong matter field.
However, using similar techniques than 't Hooft and Veltman, Deser,
Tsao, and Nieuwenhuizen found non renormalizable divergences at one
loop level if the Lagrangian of gravitation was coupled to Yang Mills
\cite{Einstein YM}, or Maxwell Fields\cite{EinsteinMaxwell-1}. Finally,
Deser and Nieuwenhuizen also computed non-renormalizable divergences
in one loop amplitudes if the gravitational field is coupled to a
fermion field \cite{Einsteindirac}. The latter is more difficult
to compute since the Lagrangian of gravitation can only be coupled
to Dirac fermions if one describes the gravitational field by using
tedrads tetrads. 

All this leads to the assumption that perhaps an additional symmetry
is needed that cancels the divergences. One candidate for such a symmetry
is supersymmetry, where for each fermion there exists a supersymmetric
partner boson. In 1986, however, Goroff and Sagnotti \cite{Goroff1,Goroff2},
using the help of computeralgebra, calculated counterterms for the
second loop level. They found nonrenormalizable divergences even in
pure gravity. 

At two loop order, one has an additional integration and the counterterm
therefore must be composed out of three Riemann tensors or out of
two Riemann tensors and two derivatives. This leaves the combinations
$\nabla^{\mu}R\nabla_{\mu}R$, $R^{3}$, $\nabla_{\mu}R_{\alpha\beta}\nabla^{\mu}R^{\alpha\beta}$,
$RR_{\alpha\beta}R^{\alpha\beta}$, $R_{\alpha\gamma}R_{\beta\delta}R^{\alpha\beta\gamma\delta}$,
$R_{\alpha}^{\;\;\beta}R_{\beta}^{\;\;\gamma}R_{\gamma}^{\;\;\alpha}$,
$RR_{\alpha\beta\gamma\delta}R^{\alpha\beta\gamma\delta}$, $R^{\alpha\eta\gamma\delta}R_{\text{\ensuremath{\alpha\beta\gamma\epsilon}}}R_{\delta}^{\;\;\epsilon}$,
$R_{\;\;\;\gamma\delta}^{^{\alpha\beta}}R_{\;\;\;\epsilon\zeta}^{\gamma\delta}R_{\;\;\;\alpha\beta}^{\epsilon\zeta}$,
$R_{\alpha\beta\gamma\delta}R_{\;\;\epsilon\;\;\zeta}^{\alpha\;\;\gamma}R^{\beta\epsilon\delta\zeta}$.
The first six terms vanish on mass shell, due to eq.(\ref{eq:form2})
and (\ref{eq:form1}). This leaves the last two expressions- According
to Goroff and Sagnotti, they are linearly dependent, and so we are
left with a two loop counternetn in form of 
\begin{equation}
\Delta\mathcal{L}=\frac{\alpha}{\epsilon}\sqrt{-g}R_{\;\;\;\gamma\delta}^{^{\alpha\beta}}R_{\;\;\;\epsilon\zeta}^{\gamma\delta}R_{\;\;\;\alpha\beta}^{\epsilon\zeta}.
\end{equation}
Goroff and Sagnotti found that the one shell diagrams contributing
to this invariant are the following two loop vertex corrections:

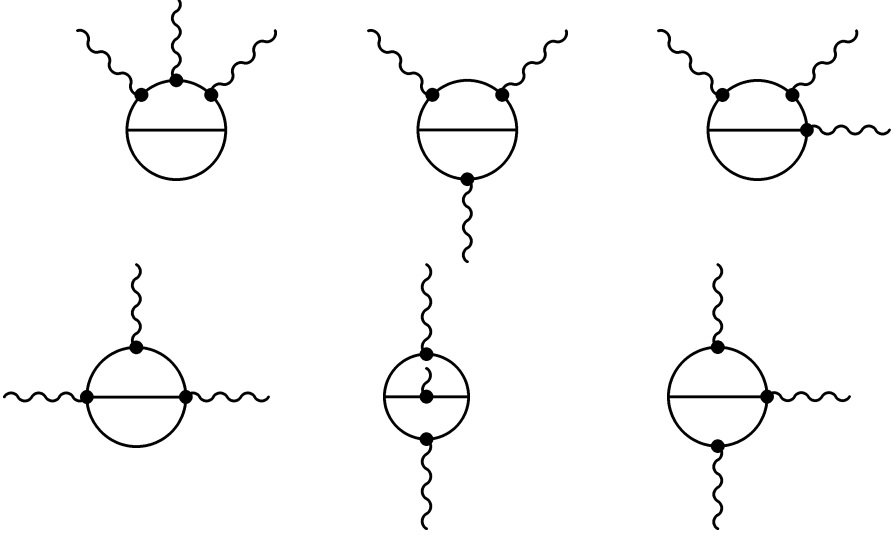
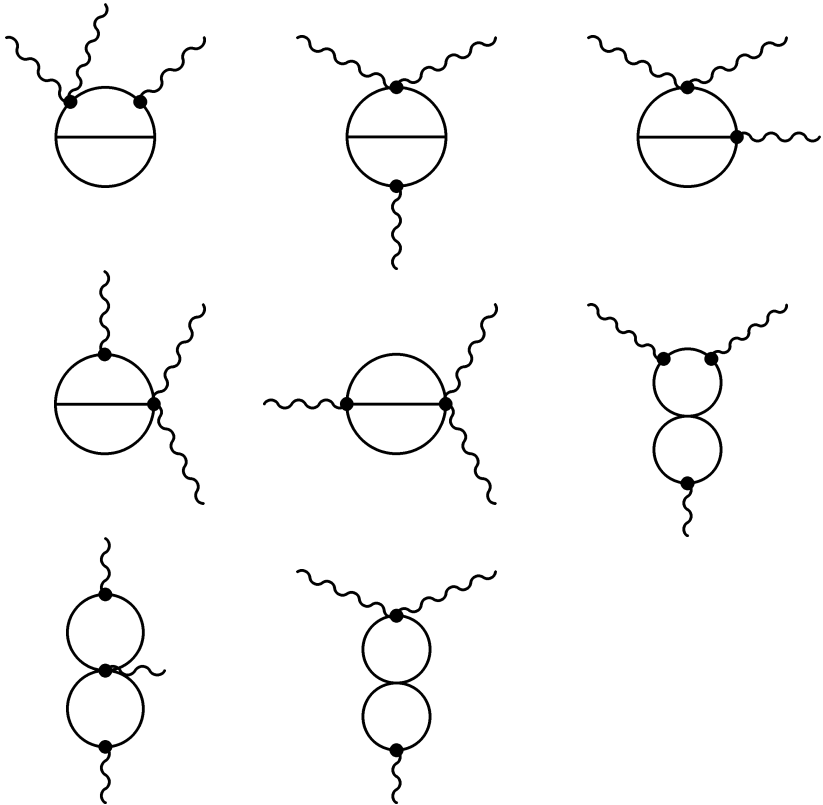
\begin{fmffile}{fmfdefykj}
\begin{fmfgraph}(100,100)\fmfsurroundn{e}{8}\begin{fmffor}{n}{1}{1}{8}\fmf{phantom}{e[n],i[n]} \end{fmffor} \fmfcyclen{phantom,tension=8/8}{i}{8} \fmf{plain,left,tension=1/3,tag=1}{i5,i1} \fmf{plain,left,tension=1/3,tag=2}{i1,i5} \fmf{plain,tension=1/3,tag=3}{i1,a1,i5} \fmfposition \fmffreeze \fmfipath{p[]} \fmfiset{p1}{vpath1(__i5,__i1)} \fmfiset{p2}{vpath2(__i1,__i5)} \fmfi{photon}{point 1length(p1)/4 of p1 -- vloc(__e4)} \fmfi{photon}{point 2length(p1)/4 of p1 -- vloc(__e3)} \fmfi{photon}{point 3length(p1)/4 of p1 -- vloc(__e2)}\fmfiv{d.sh=circle,d.f=1,d.siz=2thick}{point 1length(p1)/4 of p1} \fmfiv{d.sh=circle,d.f=1,d.siz=2thick}{point 2length(p1)/4 of p1} \fmfiv{d.sh=circle,d.f=1,d.siz=2thick}{point 3length(p1)/4 of p1}
\end{fmfgraph}  
\begin{fmfgraph}(100,100)
\fmfsurroundn{e}{8} \begin{fmffor}{n}{1}{1}{8} \fmf{phantom}{e[n],i[n]} \end{fmffor} \fmfcyclen{phantom,tension=8/8}{i}{8} \fmf{plain,left,tension=1/3,tag=1}{i5,i1} \fmf{plain,left,tension=1/3,tag=2}{i1,i5} \fmf{plain,tension=1/3,tag=3}{i1,a1,i5} \fmfposition \fmffreeze \fmfipath{p[]} \fmfiset{p1}{vpath1(__i5,__i1)} \fmfiset{p2}{vpath2(__i1,__i5)} \fmfi{photon}{point 1length(p1)/4 of p1 -- vloc(__e4)} \fmfi{photon}{point 3length(p1)/4 of p1 -- vloc(__e2)} \fmfi{photon}{point 2length(p1)/4 of p2 -- vloc(__e7)}
\fmfiv{d.sh=circle,d.f=1,d.siz=2thick}{point 1length(p1)/4 of p1} \fmfiv{d.sh=circle,d.f=1,d.siz=2thick}{point 3length(p1)/4 of p1} \fmfiv{d.sh=circle,d.f=1,d.siz=2thick}{point 2length(p1)/4 of p2}
\end{fmfgraph}
\begin{fmfgraph}(100,100)
\fmfsurroundn{e}{8} \begin{fmffor}{n}{1}{1}{8} \fmf{phantom}{e[n],i[n]} \end{fmffor} \fmfcyclen{phantom,tension=8/8}{i}{8} \fmf{plain,left,tension=1/3,tag=1}{i5,i1} \fmf{plain,left,tension=1/3,tag=2}{i1,i5} \fmf{plain,tension=1/3,tag=3}{i1,a1,i5} \fmfposition \fmffreeze \fmfipath{p[]} \fmfiset{p1}{vpath1(__i5,__i1)} \fmfiset{p2}{vpath2(__i1,__i5)} \fmfi{photon}{point 1length(p1)/4 of p1 -- vloc(__e4)} \fmfi{photon}{point 3length(p1)/4 of p1 -- vloc(__e2)} \fmfi{photon}{point 4length(p1)/4 of p1 -- vloc(__e1)}
\fmfiv{d.sh=circle,d.f=1,d.siz=2thick}{point 1length(p1)/4 of p1} \fmfiv{d.sh=circle,d.f=1,d.siz=2thick}{point 3length(p1)/4 of p1} \fmfiv{d.sh=circle,d.f=1,d.siz=2thick}{point 4length(p1)/4 of p1}
\end{fmfgraph} \\
\begin{fmfgraph}(100,100)
\fmfsurroundn{e}{8} \begin{fmffor}{n}{1}{1}{8} \fmf{phantom}{e[n],i[n]} \end{fmffor} \fmfcyclen{phantom,tension=8/8}{i}{8} \fmf{plain,left,tension=1/3,tag=1}{i5,i1} \fmf{plain,left,tension=1/3,tag=2}{i1,i5} \fmf{plain,tension=1/3,tag=3}{i1,a1,i5} \fmfposition \fmffreeze \fmfipath{p[]} \fmfiset{p1}{vpath1(__i5,__i1)} \fmfiset{p2}{vpath2(__i1,__i5)} \fmfi{photon}{point 0length(p1)/4 of p1 -- vloc(__e5)} \fmfi{photon}{point 2length(p1)/4 of p1 -- vloc(__e3)} \fmfi{photon}{point 4length(p1)/4 of p1 -- vloc(__e1)}
\fmfiv{d.sh=circle,d.f=1,d.siz=2thick}{point 0length(p1)/4 of p1} \fmfiv{d.sh=circle,d.f=1,d.siz=2thick}{point 2length(p1)/4 of p1} \fmfiv{d.sh=circle,d.f=1,d.siz=2thick}{point 4length(p1)/4 of p1}
\end{fmfgraph}
\begin{fmfgraph}(100,100)
\fmfsurroundn{e}{8} \begin{fmffor}{n}{1}{1}{8} \fmf{phantom}{e[n],i[n]} \end{fmffor} \fmfcyclen{phantom,tension=8/8}{i}{8} \fmf{plain,left,tension=1/3,tag=1}{i5,i1} \fmf{plain,left,tension=1/3,tag=2}{i1,i5} \fmf{plain,tag=3}{i1,a1,i5} \fmf{phantom,tension=0.0,tag=3}{a1,i3}
\fmfposition \fmffreeze \fmfipath{p[]} \fmfiset{p1}{vpath1(__i5,__i1)} \fmfiset{p2}{vpath2(__i1,__i5)} \fmfiset{p3}{vpath3(__a1,__i3)} \fmfi{photon}{subpath (0,length(p1)/6) of p3}
\fmfi{photon}{point 2length(p1)/4 of p1 -- vloc(__e3)} \fmfi{photon}{point 2length(p1)/4 of p2 -- vloc(__e7)}
\fmfiv{d.sh=circle,d.f=1,d.siz=2thick}{point 2length(p1)/4 of p1} \fmfiv{d.sh=circle,d.f=1,d.siz=2thick}{point 2length(p1)/4 of p2} \fmfiv{d.sh=circle,d.f=1,d.siz=2thick}{point 0length(p1)/4 of p3}
\end{fmfgraph} 
\begin{fmfgraph}(100,100)
\fmfsurroundn{e}{8} \begin{fmffor}{n}{1}{1}{8} \fmf{phantom}{e[n],i[n]} \end{fmffor} \fmfcyclen{phantom,tension=8/8}{i}{8} \fmf{plain,left,tension=1/3,tag=1}{i5,i1} \fmf{plain,left,tension=1/3,tag=2}{i1,i5} \fmf{plain,tension=1/3,tag=3}{i1,a1,i5} \fmfposition \fmffreeze \fmfipath{p[]} \fmfiset{p1}{vpath1(__i5,__i1)} \fmfiset{p2}{vpath2(__i1,__i5)} \fmfi{photon}{point 2length(p1)/4 of p1 -- vloc(__e3)} \fmfi{photon}{point 4length(p1)/4 of p1 -- vloc(__e1)} \fmfi{photon}{point 2length(p1)/4 of p2 -- vloc(__e7)}
\fmfiv{d.sh=circle,d.f=1,d.siz=2thick}{point 2length(p1)/4 of p1} \fmfiv{d.sh=circle,d.f=1,d.siz=2thick}{point 4length(p1)/4 of p1} \fmfiv{d.sh=circle,d.f=1,d.siz=2thick}{point 2length(p1)/4 of p2}
\end{fmfgraph}\\
\begin{fmfgraph}(100,100)
\fmfsurroundn{e}{8} \begin{fmffor}{n}{1}{1}{8} \fmf{phantom}{e[n],i[n]} \end{fmffor} \fmfcyclen{phantom,tension=8/8}{i}{8} \fmf{plain,left,tension=1/3,tag=1}{i5,i1} \fmf{plain,left,tension=1/3,tag=2}{i1,i5} \fmf{plain,tension=1/3,tag=3}{i1,a1,i5} \fmfposition \fmffreeze \fmfipath{p[]} \fmfiset{p1}{vpath1(__i5,__i1)} \fmfiset{p2}{vpath2(__i1,__i5)} \fmfi{photon}{point 1length(p1)/4 of p1 -- vloc(__e4)} \fmfi{photon}{point 1length(p1)/4 of p1 -- vloc(__e3)} \fmfi{photon}{point 3length(p1)/4 of p1 -- vloc(__e2)}
\fmfiv{d.sh=circle,d.f=1,d.siz=2thick}{point 1length(p1)/4 of p1} \fmfiv{d.sh=circle,d.f=1,d.siz=2thick}{point 3length(p1)/4 of p1}
\end{fmfgraph}
\begin{fmfgraph}(100,100)
\fmfsurroundn{e}{8} \begin{fmffor}{n}{1}{1}{8} \fmf{phantom}{e[n],i[n]} \end{fmffor} \fmfcyclen{phantom,tension=8/8}{i}{8} \fmf{plain,left,tension=1/3,tag=1}{i5,i1} \fmf{plain,left,tension=1/3,tag=2}{i1,i5} \fmf{plain,tension=1/3,tag=3}{i1,a1,i5} \fmfposition \fmffreeze \fmfipath{p[]} \fmfiset{p1}{vpath1(__i5,__i1)} \fmfiset{p2}{vpath2(__i1,__i5)} \fmfi{photon}{point 2length(p1)/4 of p1 -- vloc(__e2)} \fmfi{photon}{point 2length(p1)/4 of p1 -- vloc(__e4)} \fmfi{photon}{point 2length(p1)/4 of p2 -- vloc(__e7)}
\fmfiv{d.sh=circle,d.f=1,d.siz=2thick}{point 2length(p1)/4 of p1} \fmfiv{d.sh=circle,d.f=1,d.siz=2thick}{point 2length(p1)/4 of p2}
\end{fmfgraph}
\begin{fmfgraph}(100,100)
\fmfsurroundn{e}{8} \begin{fmffor}{n}{1}{1}{8} \fmf{phantom}{e[n],i[n]} \end{fmffor} \fmfcyclen{phantom,tension=8/8}{i}{8} \fmf{plain,left,tension=1/3,tag=1}{i5,i1} \fmf{plain,left,tension=1/3,tag=2}{i1,i5} \fmf{plain,tension=1/3,tag=3}{i1,a1,i5} \fmfposition \fmffreeze \fmfipath{p[]} \fmfiset{p1}{vpath1(__i5,__i1)} \fmfiset{p2}{vpath2(__i1,__i5)} \fmfi{photon}{point 2length(p1)/4 of p1 -- vloc(__e2)} \fmfi{photon}{point 2length(p1)/4 of p1 -- vloc(__e4)} \fmfi{photon}{point 4length(p1)/4 of p1 -- vloc(__e1)}
\fmfiv{d.sh=circle,d.f=1,d.siz=2thick}{point 2length(p1)/4 of p1} \fmfiv{d.sh=circle,d.f=1,d.siz=2thick}{point 4length(p1)/4 of p1}
\end{fmfgraph}\\
\begin{fmfgraph}(100,100)
\fmfsurroundn{e}{8} \begin{fmffor}{n}{1}{1}{8} \fmf{phantom}{e[n],i[n]} \end{fmffor} \fmfcyclen{phantom,tension=8/8}{i}{8} \fmf{plain,left,tension=1/3,tag=1}{i5,i1} \fmf{plain,left,tension=1/3,tag=2}{i1,i5} \fmf{plain,tension=1/3,tag=3}{i1,a1,i5} \fmfposition \fmffreeze \fmfipath{p[]} \fmfiset{p1}{vpath1(__i5,__i1)} \fmfiset{p2}{vpath2(__i1,__i5)} \fmfi{photon}{point 2length(p1)/4 of p1 -- vloc(__e3)} \fmfi{photon}{point 4length(p1)/4 of p1 -- vloc(__e2)} \fmfi{photon}{point 4length(p1)/4 of p1 -- vloc(__e8)}
\fmfiv{d.sh=circle,d.f=1,d.siz=2thick}{point 2length(p1)/4 of p1} \fmfiv{d.sh=circle,d.f=1,d.siz=2thick}{point 4length(p1)/4 of p1}
\end{fmfgraph}
\begin{fmfgraph}(100,100)
\fmfsurroundn{e}{8} \begin{fmffor}{n}{1}{1}{8} \fmf{phantom}{e[n],i[n]} \end{fmffor} \fmfcyclen{phantom,tension=8/8}{i}{8} \fmf{plain,left,tension=1/3,tag=1}{i5,i1} \fmf{plain,left,tension=1/3,tag=2}{i1,i5} \fmf{plain,tension=1/3,tag=3}{i1,a1,i5} \fmfposition \fmffreeze \fmfipath{p[]} \fmfiset{p1}{vpath1(__i5,__i1)} \fmfiset{p2}{vpath2(__i1,__i5)} \fmfi{photon}{point 0length(p1)/4 of p1 -- vloc(__e5)} \fmfi{photon}{point 4length(p1)/4 of p1 -- vloc(__e2)} \fmfi{photon}{point 4length(p1)/4 of p1 -- vloc(__e8)}
\fmfiv{d.sh=circle,d.f=1,d.siz=2thick}{point 0length(p1)/4 of p1} \fmfiv{d.sh=circle,d.f=1,d.siz=2thick}{point 4length(p1)/4 of p1}
\end{fmfgraph}
\begin{fmfgraph}(100,100) \fmfsurroundn{e}{8} \begin{fmffor}{n}{1}{1}{8} \fmf{phantom}{e[n],i[n]} \end{fmffor} \fmfcyclen{phantom,tension=8/8}{i}{8}
\fmf{plain,left,tension=1/2,tag=1}{i3,v1} \fmf{plain,left,tension=1/2,tag=2}{v1,i3} \fmf{plain,left,tension=1/2,tag=3}{v1,i7} \fmf{plain,left,tension=1/2,tag=4}{i7,v1} \fmf{photon}{i7,e7} \fmffreeze \fmfipath{p[]} \fmfiset{p1}{vpath1(__i3,__v1)} \fmfiset{p2}{vpath2(__v1,__i3)} \fmfiset{p3}{vpath4(__v1,__i7)} \fmfi{photon}{point 1length(p1)/4 of p1 -- vloc(__e2)} \fmfi{photon}{point 3length(p1)/4 of p2 -- vloc(__e4)}
\fmfiv{d.sh=circle,d.f=1,d.siz=2thick}{point 1length(p1)/4 of p1} \fmfiv{d.sh=circle,d.f=1,d.siz=2thick}{point 3length(p1)/4 of p2} \fmfiv{d.sh=circle,d.f=1,d.siz=2thick}{point 0length(p1)/4 of p3}
\end{fmfgraph}\\
\begin{fmfgraph}(100,100) \fmfsurroundn{e}{8} \begin{fmffor}{n}{1}{1}{8} \fmf{phantom}{e[n],i[n]} \end{fmffor} \fmfcyclen{phantom,tension=8/8}{i}{8}
\fmf{plain,left,tension=1/2,tag=1}{i3,v1} \fmf{plain,left,tension=1/2,tag=2}{v1,i3} \fmf{plain,left,tension=1/2,tag=3}{v1,i7} \fmf{plain,left,tension=1/2,tag=4}{i7,v1} \fmf{phantom}{i5,v1} \fmf{photon}{v1,i1} \fmf{photon}{i3,e3} \fmf{photon}{i7,e7} \fmffreeze \fmfipath{p[]} \fmfiset{p1}{vpath1(__i3,__v1)} \fmfiset{p2}{vpath2(__v1,__i3)} \fmfiset{p3}{vpath3(__v1,__i7)} \fmfiset{p4}{vpath4(__i7,__v1)} \fmfiv{d.sh=circle,d.f=1,d.siz=2thick}{point 0length(p1)/4 of p1} \fmfiv{d.sh=circle,d.f=1,d.siz=2thick}{point 0length(p3)/4 of p3} \fmfiv{d.sh=circle,d.f=1,d.siz=2thick}{point 4length(p3)/4 of p3}
\end{fmfgraph}
\begin{fmfgraph}(100,100) \fmfsurroundn{e}{8} \begin{fmffor}{n}{1}{1}{8} \fmf{phantom}{e[n],i[n]} \end{fmffor} \fmfcyclen{phantom,tension=8/8}{i}{8}
\fmf{plain,left,tension=1/2,tag=1}{i3,v1} \fmf{plain,left,tension=1/2,tag=2}{v1,i3} \fmf{plain,left,tension=1/2,tag=3}{v1,i7} \fmf{plain,left,tension=1/2,tag=4}{i7,v1} \fmf{photon}{i7,e7} \fmffreeze \fmfipath{p[]} \fmfiset{p1}{vpath1(__i3,__v1)} \fmfiset{p2}{vpath2(__v1,__i3)} \fmfiset{p3}{vpath3(__v1,__i7)} \fmfiset{p4}{vpath4(__i7,__v1)} \fmfi{photon}{point 0length(p1)/4 of p1 -- vloc(__e2)} \fmfi{photon}{point 0length(p1)/4 of p1 -- vloc(__e4)}
\fmfiv{d.sh=circle,d.f=1,d.siz=2thick}{point 0length(p1)/4 of p1} \fmfiv{d.sh=circle,d.f=1,d.siz=2thick}{point 4length(p3)/4 of p3}
\end{fmfgraph}
\end{fmffile}

In the computation, the one loop counterterms where treated as a field
redefinition, in order to prevent problems with terms mixed of background
and quantum field that arise in loop computations with the background
field method. Furthermore, a combined notation for gravitons and ghosts
was used that made it possible to implement the calculation more easily
on a computer (the solid circles in the diagrams actually represent
these combined internal ghost graviton lines, while the whiggled external
lines are the gravitons). The calculation took three days on a VAX
11/780 gave the result
\begin{equation}
\Delta\mathcal{L}=\frac{209}{2880(4\pi)^{4}}\frac{1}{\epsilon}\int d^{4}x\sqrt{-g}R_{\;\;\gamma\delta}^{^{\alpha\beta}}R_{\;\;\epsilon\zeta}^{\gamma\delta}R_{;\;\;\alpha\beta}^{\epsilon\zeta}
\end{equation}

Unfortunately, the descriptions in the article of Goroff and Sagnotti
are not very detailed. More details of a different approach were given
by deVen \cite{Vandeven}, who used a different technique to reproduce
the result of Goroff and Sagnotti, However, he also had to resort
to a computer implementation.

\subsection{Considerations on the non perturbative evaluation of the gravitational
path integral }

One can also try to evaluate the path integral without using a perturbation
expansion of the action. In the following we make some comments on
attempts to evaluate the path integral non-perturbativly by the methods
of Euclidean quantum gravity. Usually, the action of the gravitational
field is taken to be $S=\int d^{4}x\sqrt{-g}R$. However, in this
action a boundary term is omitted that will be crucial for the following.
If we vary the action, we get 
\begin{eqnarray}
\delta S & = & \int d^{4}x\left(\sqrt{-g}g^{\mu\nu}\delta R_{\mu\nu}+\sqrt{-g}R_{\mu\nu}\delta g^{\mu\nu}+R\delta\sqrt{-g}\right)\nonumber \\
 & = & \int d^{4}x\left(\sqrt{-g}g^{\mu\nu}\delta R_{\mu\nu}+\sqrt{-g}R_{\mu\nu}\delta g^{\mu\nu}-\frac{1}{2}R\sqrt{-g}g_{\mu\nu}\delta g^{\mu\nu}\right)\nonumber \\
 & = & \int d^{4}x\left(\sqrt{-g}g^{\mu\nu}\delta R_{\mu\nu}+\sqrt{-g}\left(R_{\mu\nu}-\frac{1}{2}Rg_{\mu\nu}\right)\delta g^{\mu\nu}\right)
\end{eqnarray}

From the last term on the right hand side, Einstein's equations can
be derived. The first term yields with . 
\begin{equation}
\delta R_{\mu\nu}=\nabla_{\lambda}\delta\Gamma_{\nu\mu}^{\lambda}-\nabla_{\nu}\delta\Gamma_{\lambda\mu}^{\lambda}
\end{equation}
the result
\begin{equation}
\int d^{4}x\sqrt{-g}\nabla_{\sigma}\left(g^{\mu\nu}\delta\Gamma_{\mu\nu}^{\sigma}-g^{\mu\nu}\delta\Gamma_{\lambda\mu}^{\lambda}\right)=\int d^{4}x\sqrt{-g}\nabla_{\sigma}\left(g_{\mu\nu}\nabla^{\sigma}\delta g^{\mu\nu}-\nabla_{\lambda}\delta g^{\sigma\lambda}\right)
\end{equation}
By Stokes theorem, the volume integral over the divergence can be
converted into a surface integral around the boundary $\partial M$
of the spacetime manifold. Further evaluation of the integrand, see
\cite{Wald}, shows that it is ecactly equal to $-2\sqrt{-\gamma}K$
where $K$ is the trace of the extrinsic curvature and $\gamma$ is
the three metric on the boundary. Accordingly, we have to add this
term to the action for the gravitational field. Our new action now
reads 
\begin{equation}
S=\int d^{4}x\sqrt{-g}R+2\int_{\partial M}d^{3}x\sqrt{-\gamma}K-C
\end{equation}
where $C$ is a constant that is independent of $g$. In asymptotically
flat space a natural choice for $C$ is such that $S=0$ for the Minkowski
metric $\eta_{\mu\nu}$, or 
\begin{equation}
C=-2\int_{\partial M}d^{3}x\sqrt{-\gamma}K^{0}
\end{equation}
with $K^{0}$ as the extrinsic curvature at the boundary embedded
in flat space. Thereby, the action becomes 
\begin{equation}
S=\int d^{4}x\sqrt{-g}R+2\int_{\partial M}d^{3}x\sqrt{-\gamma}K-2\int_{\partial M}d^{3}x\sqrt{-\gamma}K^{0}
\end{equation}
For a quantum theory with a scalar field $\varphi$, the path integral
\begin{equation}
Z=\int\mathcal{D}x(t)e^{iS(\varphi)}
\end{equation}
oscillates and does not converge. However, the integral can be made
to converge by a Wick rotation. One replaces $t$ by $-i\tau$, which
introduces a factor of $-i$ in the volume integral of the action.
The path integral then becomes 
\begin{equation}
Z_{eu}=\int\mathcal{D}\varphi(t)e^{-I(\varphi)}
\end{equation}
where 
\begin{equation}
I(\varphi)=-iS(\varphi)
\end{equation}
is the Euclidean action. Since $I(\varphi)\in\mathrm{\mathbb{R}}$
and $I(\varphi)\geq0$ for fields that are real on the Euclidean space
$(\tau,x^{1},x^{2},x^{3})$, the path integral converges, which makes
its evaluation possible.

Unfortunately, this is not the case with the gravitational field.
After a Wick rotation, the volume element $d^{4}x\sqrt{-g}$ becomes
$-id^{4}x\sqrt{g}$ and Euclidean action for the gravitational field
is: 
\begin{equation}
I=-\int d^{4}x\sqrt{g}R-2\int_{\partial M}d^{3}x\sqrt{\gamma}(K-K^{0})\label{eq:fulleuclideanaction}
\end{equation}
The minus sign comes from the fact that the direction of rotation
into the complex plane must be the same as that for the matter fields
that might be included into the action. 

According to Gibbons, Hawking and Perry, \cite{Zetafunction}, a conformal
transformation $\tilde{g}_{\mu\nu}=\Omega^{2}g_{\mu\nu}$,  with $\Omega$
as a positive conformal factor that is equal to one on $\partial M$
changes this action into 
\begin{equation}
I(\Omega^{2},g)=-\int d^{4}x\sqrt{g}\left(\Omega^{2}R+6g^{\mu\nu}\partial_{\mu}\Omega\partial_{\nu}\Omega\right)-2\int_{\partial M}d^{3}x\sqrt{\gamma}\Omega^{2}(K-K^{0})
\end{equation}
Because of the covariant derivatives of the conformal factor, $I(\Omega,g)$
can be arbitrarily negative in case a rapidly varying conformal factor
is chosen. 

The Euclidean path integral can be formulated with a quantity 
\begin{equation}
Y(g)=\int\mathcal{D}\Omega e^{-I(\Omega^{2},g)}
\end{equation}
as 
\begin{equation}
Z_{eu}=\int\mathcal{D}gY(g)
\end{equation}
where one might as well include the necessary ghost and gauge fixing
terms.

The path integral in this formulation contains a summation over all
possible conformal factors $\Omega$. Some of these factors lead to
a negative action, which result in a divergent $Y(g)$. Thereby, even
when evaluated without the appeal to a perturbation series in terms
of Feynman diagrams, the path integral of gravitation is highly divergent.
However, Schoen and Yau \cite{Posact} were able to prove that at
least for asymptotically Euclidean spacetimes with $R=0$, the so
called positive action conjecture holds, which states that $I\geq0$,
with $I=0$ if $g$ is flat.

\section{Hawking's articles on the topology of spacetime at Planck-scale}

It was Stephen Hawking who gave a more rigorous outline of the topology
of spacetime that follows from the quantum gravitational path integral
\cite{HawkingSpacetimefoam}. For understanding this work, we will
need some basic definitions of algebraic topology.

\subsection{Some definitions from algebraic topology}

For what follows, a detailed description of the mathematics involved
would lead too far. We therefore redirect the interested reader to
the literature. For example the expository article of Eguchi, Gilkey
and Hanson \cite{EguchiGilkeyHanson} will provide a nice entry. Also
the book of Nakahara \cite{Nakahara} will cover most of the relevant
mathematics, and many of the mathematical definitions and proofs below
can be found in similar, and sometimes much more detailed form in
Nakahara's book. Also, the books from R. C. Kirby \cite{Kirby} as
well as A.Scorpan on four manifolds \cite{Wilde World} were used.
Sometimes, the book of Bredon \cite{Bredon} was of help. We confine
us here to give a few: definitions on topology.
\begin{defn}
Topology definitions
\begin{itemize}
\item A topological space $X$ is said to be disconnected if $X=T_{1}\cup T_{2},$
where $T_{i}\subset X,T_{i}\neq\emptyset$ are open sets with $T_{1}\cap T_{2}=\emptyset$.
If $X$ is not disconnected, $X$ is called connected.
\item We say $T\subset X$ is a connected component of $X$ if $T$ is a
maximal connected subset of $X$ (the connected subsets are ordered
by inclusion).
\item A connected manifold is a connected topological space that is also
a manifold. 
\item Let $M$ be a connected manifold covered by charts $\{U_{j}\}$.We
say that $M$ is orientable if for any overlapping charts $U_{i},U_{j}$
there exists a local coordinate system $\{x^{\mu}\}$ for $U_{i}$
and $\{y^{\nu}\}$ for $U_{j}$such that $det(\partial x^{\mu}/\partial y^{\nu})>0$
\item An n-manifold with boundary is a topological space $X$ in which each
point $x\in X$ has a neighborhood $U_{x}$ homeomorphic to either
$\mathbb{R}^{n}$ or to $H^{n}:=\left\{ (x_{1},\ldots x_{n})\in\mathbb{R}^{n}|x_{n}\geq0\right\} $.
\item The set of all points $x\in X$ of an n-dimensional manifold with
boundary that only have neighborhoods homeomorphic to $H^{n}$ is
called the boundary of $X$ and denoted by $\partial X$. If $\partial X=\emptyset$,
then the manifold is called closed.
\item A continuous function $f:[0,1]\rightarrow X$ with $f(0)=x\in X$
and $f(1)=y\in X$ is called path. 
\item If there is a path joining any two points $x,y\in X$ the topological
space $X$ is called path connected.
\item A homotopy between two continuous functions $f,g:X\rightarrow Y$
is a continuous function $H:X\times[0,1]\rightarrow Y$ where, $\forall x\in X:H(x,0)=f(x)$
and $H(x,1)=g(x)$. Continuous functions $f,g:X\rightarrow Y$ are
said homotopic, or $f\thicksim g$,  if and only if there exists a
homotopy $H$ between them. This describes an equivalence class of
homotopic functions, see, e.g. Nakahara \cite{Nakahara}, p. 124 for
a proof.
\item If $f$ and $g$ are continuous maps from $X$ to $Y$ and $K$ is
a subset of $X$, then we say that $f$ and $g$ are homotopic relative
to $K$ if there exists a homotopy $H:X\times[0,1]\mapsto Y$ between
$f$ and $g$ such that $\forall k\in K,\forall t\in[0,1]:H(k,t)=f(k)=g(k)$. 
\item A topological space $X$ is simply connected if and only if it is
path-connected, and whenever two continuous maps $u:[0,1]\rightarrow X$
and $w:[0,1]\rightarrow X$ with the same starting and ending points
$u(0)=w(0)$ and $u(1)=w(1)$, are homotopic relative to $[0,1]$. 
\item Two topological spaces $X$ and $Y$ are homotopy equivalent if $\exists$
continuous maps $f:X\mapsto Y$ and $g:Y\mapsto X$ such that $g\circ f$
is homotopic to the identity map of $X$ and $f\circ g$ homotopic
to the identity map of $Y$. 
\item Let $e_{0},e_{1},\ldots e_{r}$ linearly independent vectors in $\mathbb{R}^{m}$
where. $m\geq r$. Then an $r$-simplex is defined by 
\begin{equation}
\sigma_{r}\equiv(e_{0},e_{1},\ldots e_{r})\equiv\left\{ x\in\mathbb{R}^{m}|x=\sum_{i=0}^{r}\lambda_{i}e_{i}\wedge\sum_{i=0}^{r}\lambda_{i}=1\wedge0\leq\lambda_{i}\leq1\right\} 
\end{equation}
with the $\lambda_{i}$ called barycentric coordinates. 
\item Given an r simplex $\sigma_{r}\equiv(e_{0},e_{1},\ldots e_{r})$ we
define the $r-1$ simplex 
\begin{equation}
(e_{0},e_{1},\hat{e}_{i}\ldots e_{r})\equiv(e_{0},e_{1},e_{i-1},e_{i+1}\ldots e_{r})
\end{equation}
which we call the $i$-th $r-1$ face of $\sigma_{r}$
\item Let $M$ be an m dimensional manifold and $\sigma_{r}$ an r-simplex.
We define a smooth map $f:\sigma_{r}\rightarrow M$ where $f(\sigma_{r})=\emptyset$
if $\sigma_{r}=\emptyset$. We denote the image of $\sigma_{r}$ by
$f(\sigma_{r})=s_{r}$, and call $s_{r}$ a singular $r$ simplex
in $M$. 
\item If $\{s_{r,i}\}$ is a set of singular r-simplexes in $M$, we define
a singular r-chain in $M$ by a formal sum 
\begin{equation}
c_{r}=\sum_{i}a_{i}s_{r,i}
\end{equation}
which goes over all non-empty $s_{r,i}$ and where $a_{i}\in\mathbb{R}$.
These singular r-chains form a group $C_{r}(M)$ with multiplication
\begin{equation}
c_{r}*c_{r}'=\sum_{i}(a_{i}+a'_{i})s_{r,i}
\end{equation}
unit element 
\begin{equation}
e=\sum_{i}0s_{r,i}
\end{equation}
 and inverse 
\begin{equation}
c_{r}^{-1}=-c_{r}=\sum_{i}(-a_{i})s_{r,i}
\end{equation}
.
\item Given two groups $(A,*)$ and $(B,\cdot)$, a group homomorphism is
a map $h:A\mapsto B$ where$\forall u,v\in A:$ $h(u*v)=h(u)\cdot h(v)$.
\item If $h$ is a group homomorphism, we define $im(h)=\{x|x\in h(A)\subset B)$
and $ker(h)=(x|x\in A:h(x)=e)$ where $e$ is the identity element
of $B$
\item The boundary operator on a singular r-chain is a map $\partial_{r}:C_{r}(M)\rightarrow C_{r-1}(M)$
which is defined by 
\begin{equation}
\partial_{r}c_{r}\equiv\sum_{i}a_{i}\partial_{r}s_{r,i}
\end{equation}
where 
\begin{equation}
\partial_{r}s_{r,i}=f(\partial_{r}\sigma_{r,i})
\end{equation}
The boundary of a 0 simplex is defined to be zero:
\begin{equation}
\partial_{0}(e_{0})\equiv0
\end{equation}
 and we define the boundary of a general r-simplex as 
\begin{equation}
\partial_{r}\sigma_{r}\equiv\sum_{i=0}^{r}(-1)^{i}(e_{0},e_{1},\hat{e}_{i}\ldots e_{r})
\end{equation}

\end{itemize}
\end{defn}
Now we have the following theorems: 
\begin{lem}
$\partial_{r}$ is a group homomorphism\end{lem}
\begin{proof}
We have 
\begin{equation}
\partial_{r}(c_{r}*c_{r}')=\partial_{r}\sum_{i}(a_{i}+a'_{i})s_{r,i}=\sum_{i}(a_{i}+a'_{i})\partial_{r}s_{r,i}=c_{r-1}*c_{r-1}
\end{equation}
\end{proof}
\begin{lem}
If $f:G_{1}\rightarrow G_{2}$ is a homomorphism, with $e_{i}$ as
unit element of $G_{i}$ and $x_{i}^{-1}$as inverse of $x_{i}\in G_{i}$,
then $ker(f)$ is a subgroup of $G_{1}$ and $im(f)$ is a subgroup
of $G_{2}$\end{lem}
\begin{proof}
If 
\begin{equation}
x,y\in ker(f)\Rightarrow x*y\in ker(f)
\end{equation}
since 
\begin{equation}
f(x*y)=f(x)*f(y)=e_{2}*e_{2}=e_{2}
\end{equation}
similarly, we have $e_{1}\in ker(f)$, because 
\begin{equation}
e_{2}*f(a)=f(a)=f(e_{1}*a)=f(e_{1})*f(a)\Rightarrow e_{2}=f(e_{1}).
\end{equation}
If $x\in ker(f)$ then we have $x^{-1}\in ker(f)$ since 
\begin{equation}
f(e_{1})=e_{2}=f(x^{-1}x)=f(x^{-1})*f(x)=f(x^{-1})*e_{2}.
\end{equation}
With $f(x_{1}),f(x_{2})\in im(f)$ we have 
\begin{equation}
f(x_{1})*f(x_{2})=f(x_{1}*x_{2})\in im(f)
\end{equation}
 and $e_{2}\in im(f)$ because $f(e_{1})=e_{2}$. Furthermore, $f(x)\in im(f)$,
therefore, $\left(f(x)\right)^{-1}\in im(f)$ since 
\begin{equation}
e_{2}=f(e_{1})=f(x*x^{-1})=f(x)*f(x^{-1})\Rightarrow\left(f(x)\right)^{-1}=f(x^{-1})\in im(f).
\end{equation}

\end{proof}
We can use this result in the following definition:
\begin{defn}
Definition
\begin{itemize}
\item Let $c\in C_{r}(M)$. If $\partial_{r}c=0$, then c is called r-cycle.
The set of $r$-cycles $Z_{r}(M)=ker\partial_{r}$ is called r-cycle
group. If $\exists d\in C_{r+1}(M):c=\partial_{r+1}d$ then, c is
called an $r$-boundary. The set of r-boundaries $B_{r}(M)=im\partial_{r+1}$
is called the r-boundary group. 
\end{itemize}
\end{defn}
And we can show an important theorem:
\begin{thm}
$B_{r}(M)\subset Z_{r}(M)$ and they are both subgroups of $C_{r}(M)$\end{thm}
\begin{proof}
We first prove that $\partial_{r}(\partial_{r+1}\sigma_{r+1})=\emptyset$
for an r-simplex $\sigma_{r+1}$. This can be done by simply repeated
application of the boundary operator:
\begin{eqnarray}
\partial_{r}(\partial_{r+1}\sigma_{r+1}) & = & \sum_{i=0}^{r+1}(-1)^{i}\partial_{r}(e_{0},e_{1},\hat{e}_{i}\ldots e_{r+1})\nonumber \\
 & = & \sum_{i=0}^{r+1}(-1)^{i}\left(\sum_{j=0}^{i-1}(-1)^{j}(e_{0},e_{1},\ldots,\hat{e}_{j},\ldots\hat{e}_{i}\ldots e_{r+1})\right.\nonumber \\
 &  & \left.+\sum_{j=i+1}^{r+1}(-1)^{j-1}(e_{0},e_{1},\ldots,\hat{e}_{i},\ldots\hat{e}_{j}\ldots e_{r+1})\right)\nonumber \\
 & = & \sum_{j<i}(-1)^{i+j}(e_{0},e_{1},\ldots,\hat{e}_{j},\ldots\hat{e}_{i}\ldots e_{r+1})\nonumber \\
 &  & -\sum_{j>i}(-1)^{i+j}(e_{0},e_{1},\ldots\hat{e}_{i},\ldots\hat{e}_{j}\ldots e_{r+1})\nonumber \\
 & = & \emptyset
\end{eqnarray}
Using our definitions, we get 
\begin{equation}
\partial_{r}(\partial_{r+1}s_{r+1,i})=f(\partial_{r}(\partial_{r+1}\sigma_{r+1,i}))=f(\emptyset)=\emptyset
\end{equation}
But 
\begin{equation}
\partial_{r}(\partial_{r+1}c_{r+1})=\sum_{i}a_{i}\partial_{r}(\partial_{r+1}s_{r+1,i})
\end{equation}
 where the sum goes over all non empty $s_{r,i}$. Therefore, $\partial_{r}(\partial_{r+1}c_{r+1})=0$. 

If $c\in B_{r}(M)$, then $c=\partial_{r+1}d$ for a $d\in C_{r+1}(M)$
and 
\begin{equation}
\partial_{r}c=\partial_{r}(\partial_{r+1}d)=0
\end{equation}
implies that $c\in Z_{r}(M)$ so $B_{r}(M)\subset Z_{r}(M)$. Since
the boundary operator is a homomorphism, and because the image and
kernel of a homomorphism is a subgroup of the image and inverse image,
both $Z_{r}(M)$ and $B_{r}(M)$ are subgroups of $C_{r}$.
\end{proof}
These results can now be used to define the singular cohomology group.
\begin{defn}
Singular cohomology group
\begin{itemize}
\item Two r-cycles $z,z'\in Z_{r}(M)$ are homologous if $z-z'\in B_{r}(M)$
and we say $z\sim z'$. The homologous r-cycles define an equivalence
class, $[z]$ were $[z]=[z']$.
\item We define $H_{r}(M)$ as the set of equivalence classes of r-cycles
\begin{equation}
H_{r}(M)=\{[z]|z\in Z_{r}(M)\}
\end{equation}
and call it a singular homology group. 
\item We define $dim(H_{r}(M))=b_{r}$as the r-th Betti number of $M$
\item The Euler characteristic, or Euler number of an n-dimensional manifold
$M$ is 
\begin{equation}
\chi=\sum_{i=0}^{n}(-1)^{i}b_{i}
\end{equation}
 where $b_{i}$is the i-th Betti number
\end{itemize}
\end{defn}
Heuristically, the r-th Betti number of a manifold $M$ is the number
of independent closed r-surfaces that are not boundaries of some $r+1$
surface. This means, for example $b_{0}$ is the number of connected
components of $M$, $b_{1}$ is the number of one-dimensional or \textquotedbl{}circular\textquotedbl{}
holes in $M$ and $b_{2}$ is the number of two-dimensional \textquotedbl{}voids\textquotedbl{}
or \textquotedbl{}cavities\textquotedbl{} in $M$. 

We now need some definitions on r-forms:
\begin{defn}
r- forms
\begin{itemize}
\item A differential form $\omega_{r}$ of order $r$ or an $r$-form is
a totally antisymmetric tensor of type $(0,r)$ at $p\in M$, with
$M$ as a differentiable manifold. As a $(0,r)$ tensor, $\omega_{r}$
is a map $\omega_{r}:(V_{1},\ldots,V_{r})\rightarrow\mathbb{R}$ where
$(V_{1},\ldots V_{r})\in T_{p}(M)$
\item A wedge product of r one forms is an r-form, given by the totally
antisymmetric tensor product 
\begin{equation}
dx^{\mu_{1}}\wedge dx^{\mu_{2}}\wedge,\ldots\wedge dx^{\mu_{r}}=\sum_{\sigma}sign(\sigma(r))dx^{\mu_{\sigma(1)}}\otimes dx^{\mu_{\sigma(2)}}\otimes,\ldots\otimes dx^{\mu_{\sigma(r)}}
\end{equation}
where the sum goes over all permutations $\sigma$ of $1,2,\ldots,r$
and $sign(\sigma(r))=+1$ for an even permutation and $sign(\sigma(r))=-1$
for an odd permutation.
\item We denote the vector space of all $r$ forms at $p\in M$ by $\Omega_{p}^{r}(M)$
and expand an element $\omega\in\Omega_{p}^{r}(M)$ as 
\begin{equation}
\omega=\frac{1}{r!}\omega_{\mu_{1}\mu_{2},\ldots\mu_{r}}dx^{\mu_{1}}\wedge dx^{\mu_{2}}\wedge,\ldots\wedge dx^{\mu_{r}}
\end{equation}
 where $\omega_{\mu_{1}\mu_{2},\ldots\mu_{r}}$ are totally antisymmetric.
\item For a $p$ form $\omega$ and an $r$ form $\eta$ acting, we define
the wedge product acting on vectors $V_{1},\ldots V_{p+r}$ as follows:
\begin{equation}
\omega\wedge\eta(V_{1},\ldots V_{p+r})=\frac{1}{p!r!}\sum_{\sigma}sign(\sigma)\omega(V_{\sigma(1)},\ldots,V_{\sigma(p)})\eta(V_{\sigma(p+1)},\ldots,V_{\sigma(p+r)})
\end{equation}

\item We now assign an r form on every point in $M$ and denote the space
of all r forms on $M$ as $\Omega^{r}(M)$.
\item The exterior derivative of an r-form is a map $d_{r}:\Omega^{r}(M)\rightarrow\Omega^{r+1}(M)$,
where 
\begin{equation}
d_{r}\omega=\frac{1}{r!}\frac{\partial}{\partial x^{\nu}}\omega_{\mu_{1},\ldots\mu_{r}}dx^{\nu}\wedge dx^{\mu_{1}}\wedge dx^{\mu_{2}}\wedge,\ldots\wedge dx^{\mu_{r}}
\end{equation}

\end{itemize}
\end{defn}
\begin{thm}
We have $d_{r+1}(d_{r}\omega_{r})=0$\end{thm}
\begin{proof}
According to our definitions, 
\begin{equation}
d_{r+1}(d_{r}\omega_{r})=\frac{1}{r!}\frac{\partial^{2}}{\partial x^{\nu}\partial x^{\lambda}}\omega_{\mu_{1},\ldots\mu_{r}}dx^{\lambda}\wedge dx^{\nu}\wedge dx^{\mu_{1}},\ldots\wedge dx^{\mu_{r}}
\end{equation}
which is zero because $\frac{\partial^{2}}{\partial x^{\nu}\partial x^{\lambda}}\omega_{\mu_{1},\ldots\mu_{r}}$
is symmetric with respect to $\lambda,\nu$ while $dx^{\lambda}\wedge dx^{\nu}$
is antisymmetric.
\end{proof}
Similar as in the case with the singular homology, we can use this
result able to state the following theorem:
\begin{thm}
$im(d_{r})\subset ker(d_{r+1})$\end{thm}
\begin{proof}
if $\omega_{r}\in\Omega^{r}(M)$ then $d_{r}\omega_{r}\in im(d_{r})$
and $d_{r+1}(d_{r}\omega_{r})=0$, hence $d_{r}\omega\in ker(d_{r+1})$
and therefore $im(d_{r})\subset ker(d_{r+1})$\end{proof}
\begin{defn}
DeRahm cohomology group
\begin{itemize}
\item An r-form $\omega_{r}\in im(d_{r-1})$ is called exact, an r-form
$\omega_{r}\in ker(d_{r+1})$is called closed, i. e. if an r-form
is exact there exists an r-1 form $\psi_{r-1}:\omega_{r}=d_{r-1}\psi_{r-1}$and
if an r-form is closed, we have $d_{r}\omega_{r}=0$
\item The set of closed r-forms is called the r-th cocycle group $Z^{r}(M)=ker(d_{r+1})$
and the set of exact r-forms is called r-th coboundary group $B^{r}(M)=im(d_{r-1})$.
We say the r-th DeRahm cohomology group $H^{r}(M)$ is the set of
equivalence classes $[\omega]\in H^{r}(M)$, where
\begin{equation}
\left\{ \omega'\in Z^{r}(M)|\omega'=\omega+d\psi,\psi\in\Omega^{r-1}\right\} 
\end{equation}

\item for $[c]\in H_{r}(M)$ and $[\omega]\in H^{r}(M)$, we define 
\begin{eqnarray}
\Lambda_{r}:H_{r}(M)\times H^{r}(M)\rightarrow\mathbb{R},\Lambda_{r}([c],[\omega]) & = & \int_{c}\omega
\end{eqnarray}

\item A bilinear form $f(x,y)\mapsto\mathbb{R}$,  $x\in V,y\in V^{*}$
where $V,V^{*}$ are some finite dimensional vector spaces, defines
a pair of linear maps 
\begin{equation}
f_{1}:V\rightarrow V^{*}:x\mapsto(y\mapsto f(x,y))
\end{equation}
and 
\begin{equation}
f_{2}:V^{*}\rightarrow V,y\mapsto(x\mapsto f(x,y))
\end{equation}
If $V$ is finite dimensional and $f_{1}$ or $f_{2}$ is an isomorphism,
then both $f_{1}$ and $f_{2}$ are isomorphisms and the bilinearform
$f(x,y)$ is called non-degenerate. Clearly, $f_{1}$ can be an isomorphism,
if and only if 
\begin{equation}
f(x,y)=0\forall y\Rightarrow x=0
\end{equation}
and 
\begin{equation}
f(x,y)=0\forall x\Rightarrow y=0
\end{equation}
If V is finite dimensional, then $rank(f_{1})=rank(f_{2})$ and if
additionally $dim(V)=rank(f_{1})$ then $f_{1}$ and $f_{2}$ are
isomorphisms and the bilinearform is non-degenerate.
\item A bilinearform $f(x,y)\rightarrow\mathrm{\mathbb{R}}$, $x,y\in V$
is symmetric if $f(x,y)=f(y,x)$ and skew-symmetric if $f(x,y)=-f(y,x)$.
\end{itemize}
\end{defn}
\begin{thm}
$\Lambda_{r}([c],[\omega])$ is a bilinear form and it does not depend
on the representation of $\omega$ and $c$\end{thm}
\begin{proof}
$\Lambda_{r}([c+c'],[\omega])=\int_{c}\omega+\int_{c'}\omega$ and
\begin{equation}
\Lambda_{r}([c],[\omega+\omega'])=\int_{c}(\omega+\omega')=\int_{c}\omega+\int_{c}\omega'
\end{equation}
Furthermore, we have with Stokes theorem $\int_{c}d\omega=\int_{\partial c}\omega$
that 
\begin{equation}
\Lambda_{r}([c+\partial c'],[\omega])=\int_{c}\omega+\int_{\partial c'}\omega=\int_{c}\omega+\int_{c'}d\omega
\end{equation}
From $[\omega]\in H^{r}(M)\Rightarrow\omega\in Z^{r}(M)$ and thereby
$\omega$ is closed, and $d\omega=0$, thus $\int_{c'}d\omega=0$
and 
\begin{equation}
\Lambda_{r}([c+\partial c'],[\omega])=\int_{c}\omega.
\end{equation}
 Similarly 
\begin{equation}
\Lambda_{r}([c],[\omega+d\psi])=\int_{c}\omega+\int_{c}dd\psi
\end{equation}
but we have $d^{2}=0$, thereby 
\begin{equation}
\Lambda_{r}([c],[\omega+d\psi])=\int_{c}\omega
\end{equation}

\end{proof}
In 1931, DeRahm proved the following theorem:
\begin{thm}
DeRahm: Let M be a differentiable manifold. Then $H_{r}(M)$ and $H^{r}(M)$
are finite dimensional and the map $\Lambda_{r}:H_{r}(M)\times H^{r}(M)\rightarrow\mathbb{R}$
is bilinear and non degenerate\end{thm}
\begin{proof}
Shown for example in Bredon \cite{Bredon} on p. 286ff.
\end{proof}
Because $\Lambda(\;,[\omega])$ and $\Lambda(\{c],\;)$ are isomorphisms,
$H_{r}(M)$ and $H^{r}(M)$ are dual to another, and we have for the
Betti numbers 
\begin{equation}
dim(H_{r}(M))=rank(\Lambda([c],\;))=rank(\Lambda(\;,[\omega]))=dim(H^{r}(M))=b_{r}
\end{equation}
Now we need some more definitions on one forms
\begin{defn}
Definitions on one forms
\begin{itemize}
\item With $M$ as a differentiable closed m-dimensional manifold, $r\leq m$,
$[\omega]\in H^{r}(M)$ and $[\eta]\in H^{m-r}$we define a map 
\begin{equation}
\sigma_{r,m}:H^{r}(M)\times H^{m-r}\rightarrow\mathbb{R},\sigma_{r,m}([\omega],[\eta])=\int_{M}\omega\wedge\eta
\end{equation}

\item If $m=4$, we call $\sigma_{2,2}$ the intersection form of the manifold.
For non differentiable manifolds, see the more general definitions
in \cite{Wilde World} on p. 115
\item We say that the parity of the intersection form $\sigma_{2,2}$is
even if $\sigma_{2,2}(\omega,\omega)$ is even for any $\omega$,
where $[\omega]\in H^{r}(M)$, and $\sigma_{2,2}(\omega,\omega)$
is odd otherwise.
\end{itemize}
\end{defn}
\begin{thm}
Poincare Duality: Let M be a closed manifold. There is an isomorphism
between $H^{r}(M)$ and $H^{m-r}(M)$ and for the Betti numbers, we
have $b_{r}=b_{m-r}$.\end{thm}
\begin{proof}
(The proof below is valid only for differentiable manifolds, since
our definition of the intersection form was only for these manifolds.
For non-differentiable manifolds, see \cite{Bredon} p. 352): The
form $\sigma_{r,m}([\omega],[\eta])$ is bilinear and because of the
properties of the wedge product, it is also non-degenerate, Therefore,
there is an isomorphism between $H^{r}(M)$ and $H^{m-r}(M)$ and
this implies 
\begin{equation}
b_{r}=dim(H^{r}(M))=dim(H^{m-r}(M))=b_{m-r}
\end{equation}

\end{proof}
We still need to define the Hodge star and the Laplacian of an r-form 
\begin{defn}
Hodge star and Laplacian
\begin{itemize}
\item A Hodge star $*$ is a linear map $*\Omega^{r}(M)\rightarrow\Omega^{m-r}(M)$
which transforms an r-form 
\begin{equation}
\omega=\frac{1}{r!}\omega_{\mu_{1}\mu_{2},\ldots\mu_{r}}dx^{\mu_{1}}\wedge dx^{\mu_{2}}\wedge,\ldots\wedge dx^{\mu_{r}}
\end{equation}
on a differentiable manifold with metric $g$ into 
\begin{equation}
*\omega_{r}=\frac{\sqrt{|g|}}{r!(m-r)!}\omega_{\mu_{1}\mu_{2},\ldots\mu_{r}}\epsilon_{\;\;\;\;\;\;\;\;\;\;\;\;\;\;\;\nu_{r+1},\ldots,\nu_{m}}^{\mu_{1}\mu_{2},\ldots\mu_{r}}:dx^{\mu_{r+1}}\wedge,\ldots\wedge dx^{\mu_{m}}
\end{equation}
where $\epsilon$ is the totally antisymmetric tensor.
\item Let $d:\Omega^{r-1}(M)\rightarrow\Omega^{r}(M)$ be the exterior derivative
for r-1 forms on a Lorentzian manifold $M$ with dimension $m$. We
define the adjoint $d^{\dagger}:\Omega^{r}(M)\rightarrow\Omega^{r-1}(M)$
of $d$ as 
\begin{equation}
d^{\dagger}=(-1)^{mr+m}*d*
\end{equation}
and if $M$ is Riemannian, we have 
\begin{equation}
d^{\dagger}=(-1)^{mr+m+1}*d*
\end{equation}

\item The Laplacian $\Delta:\Omega^{r}(M)\rightarrow\Omega^{r}(M)$ is then
defined by 
\begin{equation}
\Delta=dd^{\dagger}+d^{\dagger}d
\end{equation}

\item An r form $\omega_{r}$ is called harmonic, if $\Delta\omega_{r}=0$
and the space of harmonic r forms is denoted by $Harm^{r}(M)$
\item The inner product of two r-forms $\omega$ and $\eta$ is defined
by 
\begin{equation}
(\omega,\eta)=\int\omega\wedge*\eta=\frac{1}{r!}\int\omega_{\mu_{1},\ldots\mu_{r}}\eta^{\mu_{1},\ldots\mu_{r}}\sqrt{|g|}dx^{1},\ldots dx^{m}
\end{equation}
which implies that $(\omega,\eta)=(\eta,\omega)$ and $(\omega,\omega)\geq0$
\end{itemize}
\end{defn}
\begin{thm}
The Hodge star satisfies 
\begin{equation}
**\omega_{r}=(-1)^{r(m-r)}
\end{equation}
 if $M$ is an $m$ dimensional Riemannian manifold and 
\begin{equation}
**\omega_{r}=(-1)^{1+r(m-r)}\omega_{r}
\end{equation}
if $M$ is Lorentzian.\end{thm}
\begin{proof}
Shown in Nakahara \cite{Nakahara}, p. 291
\end{proof}
We have an important theorem by Hodge:
\begin{thm}
Hodge: Let $M$ be a compact orientable differentiable manifold. Then,
there is an isomorphism between $H^{r}(M)$ and $Harm^{r}(M)$ \end{thm}
\begin{proof}
See Nakahara, p. 296 \cite{Nakahara}
\end{proof}
As a result of this theorem we have 
\begin{equation}
b_{r}=dim(H^{r}(M))=dim(Harm^{r}(M))
\end{equation}
for a compact manifold, and if the manifold is 4 dimensional, Poincare
duality implies that $b_{p}=b_{4-p}$ or $b_{0}=b_{4}$ and $b_{3}=b_{1}$.
For a compact manifold, a harmonic $0$ form, i.e. a scalar function
$\omega$ where $d\omega=0$, is a constant. Thus, for compact 4 manifolds,
$dim(H^{0}(M))=1=b_{0}=b_{4}$. If $M$ is simply connected, the number
of its one dimensional holes is zero, so for simply connected 4 dimensional
manifolds $0=b_{1}=b_{3}$. A more formal proof of this goes by showing
that $H_{1}$ is trivial if M is simply connected, see for example
Nakahara \cite{Nakahara}, p. 241. 

We can now define the signature of a manifold:
\begin{defn}
Signature of a manifold
\begin{itemize}
\item Let $M$ be a closed manifold of dimension $m=2l$, the form $\sigma_{l,l}:H^{l}(M)\times H^{l}\rightarrow\mathbb{R}$
is non degenerate and therefore it has the rank $b_{l}=dim(H^{l}(M))$.
Furthermore, $\sigma_{l,l}$ is symmetric, if $l$ is even and $\sigma_{l,l}$
is skew-symmetric if l is odd. If $l$ is even, $\sigma$ has $b_{+}$
positive and $b_{-}$ negative real eigenvalues where $b_{+}+b_{-}=b_{l}$.
The signature of $M$ is defined as 
\begin{equation}
\tau(M)\equiv b_{+}-b_{-}
\end{equation}
if l is odd, then we define $\tau(M)\equiv0$. 
\item We say that the form $\sigma_{l,l}$ is definite if $b_{l}=|\tau|$
and indefinite otherwise.
\end{itemize}
\end{defn}
For a closed simply connected 4 dimensional manifold, this means that
$\sigma_{22}$is definite if $|\tau|=\chi-2$ and indefinite otherwise.
\begin{thm}
Let M be a closed orientable $4r$ dimensional differentiable Riemannian
manifold. We have
\begin{equation}
Harm^{2r}(M)=Harm_{+}^{2r}(M)\otimes Harm_{-}^{2r}(M)
\end{equation}
 where $Harm_{-(+)}^{2r}$ is the space of harmonic one forms with
negative (positive) eigenvalues of $*$. Furthermore, 
\begin{equation}
\tau(M)=dim(Harm_{+}^{2r}(M))-dim(Harm_{-}^{2r}(M))
\end{equation}
 \end{thm}
\begin{proof}
For a $2r$ form on a $4r$ dimensional Riemannian manifold, we have
$*^{2}\omega_{2r}=1\omega_{2r}$, hence the eigenvalues of $*$ are
$\pm1$. We can therefore separate the space of harmonic one forms
$Harm^{2r}(M)$ into two a disjoined subspaces $Harm_{-(+)}^{2r}(M)$
of harmonic one forms with negative (positive) eigenvalues of $*$.
So we have
\begin{equation}
Harm^{2r}(M)=Harm_{+}^{2r}(M)\otimes Harm_{-}^{2r}(M)
\end{equation}
Since M is closed and orientable, we can apply Hodge's theorem and
deRahm's theorem: 
\begin{eqnarray}
b_{r} & = & dim(H_{r}(M))=dim(H^{r}(M))=dim(Harm^{r}(M))\nonumber \\
 & = & dim(Harm_{+}^{2r}(M))+dim(Harm_{-}^{2r}(M))
\end{eqnarray}
we define $\omega_{2r}^{+}\in Harm_{+}^{2r}(M)$ and $\omega_{2r}^{-}\in Harm_{-}^{2r}(M)$
and we have using the properties of the direct product between one
forms: 
\begin{equation}
\sigma_{2r,2r}(\omega_{2r}^{+},\omega_{2r}^{+})=\int_{M}\omega_{2r}^{+}\wedge\omega_{2r}^{+}=\int_{M}\omega_{2r}^{+}\wedge*\omega_{2r}^{+}=(\omega_{2r}^{+},\omega_{2r}^{+})\geq0
\end{equation}
and 
\begin{equation}
\sigma_{2r,2r}(\omega_{2r}^{-},\omega_{2r}^{-})=\int_{M}\omega_{2r}^{-}\wedge\omega_{2r}^{-}=-\int_{M}\omega_{2r}^{-}\wedge*\omega_{2r}^{-}=-(\omega_{2r}^{-},\omega_{2r}^{-})\leq0
\end{equation}
as well as 
\begin{eqnarray}
\sigma_{2r,2r}(\omega_{2r}^{+},\omega_{2r}^{-}) & = & -\int_{M}\omega_{2r}^{+}\wedge*\omega_{2r}^{-}\nonumber \\
 & = & -\int_{M}\omega_{2r}^{-}\wedge*\omega_{2r}^{+}\nonumber \\
 & = & -\sigma_{2r,2r}(\omega_{2r}^{-},\omega_{2r}^{+})\nonumber \\
 & = & 0
\end{eqnarray}
 Therefore $\sigma_{2r,2r}$ is block diagonal and we have $b_{r}=b_{+}+b_{-}$
with $b_{\pm}=dim(Harm_{\pm}^{2r}(M))$ and 
\begin{equation}
\tau(M)=dim(Harm_{+}^{2r}(M))-dim(Harm_{-}^{2r})
\end{equation}

\end{proof}
We have the following general theorems on the topology of 4 manifolds:
\begin{thm}
Markov\cite{Markov}: The problem of homotopy equivalence of two $n$-dimensional
manifolds $M,N$ is undecidable for $n>3$, The problem of homeomorphy
of two $n$-dimensional manifolds $M,N$ is undecidable for $n>3$.
\end{thm}
However, Milnor has been able to show, using earlier work of Whitehead
\cite{Whitehead1,Whitehead2} that 
\begin{thm}
Two simply connected, closed, orientable 4-manifolds are homotopy
equivalent if and only if their intersection forms are isomorphic.
\end{thm}
A theorem of Freedman \cite{Freedman} from 1981 shows that the classification
of oriented simply connected 4-manifolds up to homeomorphy can be
reduced to the classification of intersection forms:
\begin{thm}
Freedman \cite{Freedman}: Given an even (odd) intersection form,
there exists up to homeomorphism exactly one (two) simply connected,
closed 4 dimensional manifolds represented by that intersection form.
If the form is odd, the two 4 dimensional manifolds $M_{1}$ and $M_{2}$
are distinguished by their Kirby Siebenmann invariant $\kappa(M_{i})$
which is equal to 0 or 1, and vanishes if and only if the product
manifold $M_{i}\times\mathbb{R}$ is differentiable.
\end{thm}
One then can invoke Serre's classification theorem for indefinite
intersection forms \cite{Serre}:
\begin{thm}
Serre: Let Q and Q' be two symmetric bilinear unimodular forms. If
both Q and Q' are indefinite, then they are isomorphic if and only
if they have the same rank, signature and parity. \cite{Serre}
\end{thm}
For definite intersection forms, Donaldson \cite{Donaldson} showed 
\begin{thm}
Donaldson\cite{Donaldson}: If a definite form is the intersection
form of a smooth simply connected closed 4 manifold $M$ then it is
$\pm\otimes^{b_{r}}(1)$.
\end{thm}
As a consequence of all these theorems, the topology of smooth simply
connected closed oriented 4 manifolds is entirely characterized up
to homeomorphism by $\chi$ and $\tau$ and the parity of the intersection
form.

\subsection{Hawking's method to estimate the topologies that contribute dominantly
in the gravitational path integral}

Even if we are only able to evaluate the gravitational path integral
up to one loop order, one can actually use it to derive some estimates
how the spacetime at Planck scale looks like. The gravitational action
is, unfortunately, not scale invariant. In the following, we will
review an argument first made by Wheeler in \cite{WheelerFoam} and
then used by Hawking in \cite{HawkingSpacetimefoam} which shows that
fluctuations of the metric at short length scales, even if they can
change the entire topology of the spacetime, do not have a large action
and thus are not damped in the path integral. 

Regge calculus \cite{ReggeCalculus} is a method for approximately
computing a manifold. With this method, spacetime is decomposed into
a 4 simplicial complex, where each 4 simplex is taken to be flat and
to be determined by its edge lengths. The angles between the 2 simplices
that connect different 4 simplices are such that they could not be
connected together in 4 dimensional flat space. (In the following,
we will use the action with the correct factor $\tilde{I}=\frac{1}{16\pi}I$,
 that we discarded for simplicity in section 2). It was found by Regge
\cite{ReggeCalculus} that the Euclidean gravitational action is equal
to 
\begin{equation}
-\frac{1}{16\pi}\int d^{4}x\sqrt{g}R=-\frac{1}{8\pi}\sum A_{i}\delta_{i}\label{eq:actionregge}
\end{equation}
where $A_{i}$ is the area of the i-th 2 simplex and 
\begin{equation}
\delta_{i}=2\pi-\sum_{k}\theta_{k}
\end{equation}
is the deficit angle, where $\sum_{k}\theta_{k}$ is sum of the angles
between the 3 simplices which are connected by the $i$-th 2 simplex.
A simplicial complex that is stationary under small variations of
the edge length is an approximation to a smooth solution of the Einstein
Field equation. In the path integral, one integrates over all metrics,
and a manifold described by Regge calculus can then be regarded as
a certain metric in the summation of the path integral without any
approximation.

If the edge lengths are chosen such that some 4 simplices collapse
to simplices of lower dimension, the action is still well defined
and finite. For example, let $a,b,c$ be the edges of a triangle,
which is a 2 simplex. Then we must have $a+b>c$ and if $a+b=c$ the
2 simplex collapses to a one simplex (which is an edge). We may chose
an approximation of the spacetime metric with very small simplicial
complexes. We can make some of the simplices collapse to lower dimensions,
and then blow up some of the simplices by an arbitrary small amount.
Thereby, we can change the topology of the manifold. For example the
Euler number or signature may be changed this way, but the action
will remain finite by eq. (\ref{eq:actionregge}) . Moreover, if the
simplices we made our modification with have a vanishingly small edge
length, the action will, by Regge's formula (\ref{eq:actionregge}),
only change by an infinitesimally small amount. 

Without any cutoffs being employed, the gravitational amplitude is
a path integral over all possible metrics, including ones with topology
changing quantum fluctuations that can be described approximately
with simplices of very small edge lengths. Because the fluctuations
at small length scales have a small Euclidean action, the formula
of the Euclidean path integral 
\begin{equation}
Z=\int\mathcal{D}g_{\mu\nu}e^{-I}\label{eq:pathintegral-1}
\end{equation}
then implies the contributions from these quantum fluctuations are
not highly damped. Therefore, they should lead to comparatively large
contributions to the amplitude.

In the following, we want to give a review of Hawking's ideas which
describe a method for evaluating the topologies that give the dominant
contribution in the gravitational path integral. Hawking considers
a path integral 
\begin{equation}
Z(\Lambda)=\int\mathcal{D}g_{\mu\nu}e^{-I}
\end{equation}
that goes over all closed Euclidean metrics with some 4 volume V and
an Euclideanized action 
\begin{equation}
I=-\frac{1}{16\pi}\int d^{4}x\left(\sqrt{g}R-\sqrt{g}2\Lambda\right)\label{eq:actionasdafaf}
\end{equation}
where $\Lambda$ is a Lagrange multiplicator. Hawking writes that
the finite volume should be considered merely as a normalization device,
similar to periodic boundary conditions in field theory. The factor
$\Lambda$ is similar to the cosmological constant but its difference
is that the Lagrange multiplicator can get very large while the cosmological
constant is small. 

This partition function can be represented by 
\begin{equation}
Z(\Lambda)=tr(e^{-\Lambda V/(8\pi)})
\end{equation}
where the trace goes over all ``states'' of the gravitational field.
These ``states'' are however, not to be confused with quantum mechanical
``graviton'' states, but they are simply the solutions over which
the path integral is integrated. One can regard the sum above as the
Laplace transform of a function $N(V)$ where $N(V)dV$ is the number
of states in the gravitational field between $V$ and $V+dV$ 
\begin{equation}
Z(\Lambda)=\int_{0}^{\infty}dVN(V)e^{-\frac{\Lambda V}{8\pi}}\label{eq:amplitude1}
\end{equation}
and from the inverse Laplace transform, we get 
\begin{equation}
N(V)=\frac{1}{16\pi^{2}i}\int_{-i\infty}^{i\infty}d\Lambda Z(\Lambda)e^{\frac{\Lambda V}{8\pi}}\label{eq:amplitude2}
\end{equation}
 where the contour of integration is taken to the right of any singularities
of $Z(\Lambda)$ in order to ensure that $N(V)=0$ for $V\leq0$. 

One would expect the dominant contributions to the path integral to
be close to solutions of Einstein's field equation 
\begin{equation}
R_{\mu\nu}-\frac{1}{2}Rg_{\mu\nu}+\Lambda g_{\mu\nu}=0
\end{equation}
 Contracting with $g_{\mu\nu}$ yields 
\begin{equation}
R=4\Lambda\label{eq:einsteincontrac}
\end{equation}
 or 
\begin{equation}
R_{\mu\nu}=\Lambda g_{\mu\nu}\label{eq:einsteineq}
\end{equation}
Inserting $R=4\Lambda$ into the action of Eq. (\ref{eq:actionasdafaf}),
and integrating we get for a spacetime that fulfills the Einstein
equations 
\begin{equation}
I=-\frac{\Lambda V}{8\pi}
\end{equation}
From $R=4\Lambda$ one observes that $\Lambda$ must have the same
dimensions than the curvature scalar so we can set
\begin{equation}
\Lambda=-8\pi c\frac{1}{\sqrt{V}}\label{eq:Lambda}
\end{equation}
where the $8\pi$ is for convenience and c is a constant that will
turn out to depend on the topology. From this, we have
\begin{equation}
V=\frac{8^{2}\pi^{2}c^{2}}{\Lambda^{2}}\label{eq:Volume}
\end{equation}
 setting this into the action, and integrating, we get for a spacetime
that satisfies Einstein's equation: 
\begin{equation}
I=-\frac{\Lambda V}{8\pi}=-\frac{8\pi c^{2}}{\Lambda}\label{eq:simpleaction}
\end{equation}

One can show a so called Atiyah Singer index theorem, see\cite{AtiyahSinger2}.
From a version of it follows the Chern-Gauss-Bonnet theorem \cite{Chern}.
The latter states that the Euler characteristic of a 4 dimensional
manifold $M$ is given by: 
\begin{equation}
\chi(M)=\frac{1}{32\pi^{2}}\int_{M}\epsilon_{\alpha\beta\gamma\delta}R^{\alpha\beta}\wedge R^{\gamma\delta}-\frac{1}{32\pi^{2}}\int_{\partial M}\epsilon_{\alpha\beta\gamma\delta}(2K^{\alpha\beta}\wedge R^{\gamma\delta}-\frac{4}{3}K^{\alpha\beta}\wedge K_{\epsilon}^{\gamma}\wedge K^{\epsilon\delta})
\end{equation}
In the following we assume the manifold to be closed. Then the boundary
terms vanish. The integral over $M$ can be expressed as (see \cite{AsymEucl}):
\begin{equation}
\chi(M)=\frac{1}{32\pi^{2}}\int_{M}d^{4}x\sqrt{g}\left(C_{\mu\nu\alpha\beta}C^{\mu\nu\alpha\beta}-2R_{\mu\nu}R^{\mu\nu}+\frac{2}{3}R^{2}\right)
\end{equation}
where 
\begin{equation}
C_{\mu\nu\alpha\beta}=R_{\mu\nu\alpha\beta}-(g_{\mu[\alpha}R_{\beta]\nu}-g_{\nu[\alpha}R_{\beta]\mu})+\frac{1}{3}Rg_{\mu[\alpha}g_{\beta]\nu}
\end{equation}
is the Weyl tensor. 

Inserting eqs. (\ref{eq:einsteineq}) and (\ref{eq:einsteincontrac})
into the expression for $\chi(M)$ yields 
\begin{equation}
\chi(M)=\frac{1}{32\pi^{2}}\int_{M}d^{4}x\sqrt{g}\left(C_{\mu\nu\alpha\beta}C^{\mu\nu\alpha\beta}+\frac{8}{3}\Lambda^{2}\right)\label{eq:eulernumberdasdad}
\end{equation}
Note the typo in Hawkings original paper \cite{HawkingSpacetimefoam},
where he writes:
\begin{equation}
\chi(M)=\frac{1}{32\pi^{2}}\int_{M}d^{4}x\sqrt{g}\left(C_{\mu\nu\alpha\beta}C^{\mu\nu\alpha\beta}+2\frac{2}{3}\Lambda^{2}\right)\label{eq:eulernumberdaffadf}
\end{equation}
Similarly, from another version of the Atiyah Singer index theorem
\cite{AtiyahSinger2}, it follows that the signature is equal to 
\begin{equation}
\tau(M)=\frac{1}{48\pi^{2}}\int_{M}R_{\alpha}^{\beta}\wedge R_{\beta}^{\alpha}-\frac{1}{48\pi^{2}}\int_{\partial M}K_{\beta}^{\alpha}\wedge R_{\alpha}^{\beta}-\eta(0)
\end{equation}
 where $\eta(s)$ is the so called eta function of a certain differential
operator on $\partial M$. Assuming that M is closed, one can derive
(see \cite{AsymEucl}): 
\begin{equation}
\tau(M)=\frac{1}{48\pi^{2}}\int_{M}d^{4}x\sqrt{g}C_{\mu\nu\alpha\beta}*C^{\mu\nu\alpha\beta}
\end{equation}
 with $*C^{\mu\nu\alpha\beta}$ as the Hodge dual to $C_{\mu\nu\alpha\beta}$.
One has, see \cite{AsymEucl} 
\begin{equation}
C_{\mu\nu\alpha\beta}C^{\mu\nu\alpha\beta}\geq|C_{\mu\nu\alpha\beta}*C^{\mu\nu\alpha\beta}|
\end{equation}
and thereby 
\begin{eqnarray}
2\chi-3|\tau| & = & \frac{2}{32\pi^{2}}\int_{M}d^{4}x\sqrt{g}\left(C_{\mu\nu\alpha\beta}C^{\mu\nu\alpha\beta}+\frac{8}{3}\Lambda^{2}\right)\nonumber \\
 &  & -\frac{3}{48\pi^{2}}\int_{M}d^{4}x\sqrt{g}\left|C_{\mu\nu\alpha\beta}*C^{\mu\nu\alpha\beta}\right|\nonumber \\
 & \geq & \frac{2}{32\pi^{2}}\int_{M}d^{4}x\sqrt{g}\left(C_{\mu\nu\alpha\beta}C^{\mu\nu\alpha\beta}+\frac{8}{3}\Lambda^{2}\right)\nonumber \\
 &  & -\frac{3}{48\pi^{2}}\int_{M}d^{4}x\sqrt{g}C_{\mu\nu\alpha\beta}C^{\mu\nu\alpha\beta}\nonumber \\
 & = & \frac{2}{32\pi^{2}}\int_{M}d^{4}x\sqrt{g}\frac{8}{3}\Lambda^{2}\nonumber \\
 & = & \frac{2}{32\pi^{2}}V\frac{8}{3}\Lambda^{2}\nonumber \\
 & = & \frac{32}{3}c^{2}\label{eq:inequality}
\end{eqnarray}
where in the last line, we have used eq. (\ref{eq:Lambda})

The following paragraph lists the signature $\tau$ and Euler number
$\chi$ for some solutions of Einstein's field equations. The $S^{4}$
space with a metric, see \cite{Hawkingbook} p. 35
\begin{equation}
ds^{2}=(1-\frac{1}{3}\Lambda r^{2})dt+(1-\frac{1}{3}\Lambda r^{2})^{-1}dr^{2}+r^{2}d\Omega^{2}
\end{equation}
has $\chi=2$ and $\tau=0$, the $CP^{2}$ space with its metric,
see \cite{CP^2} 
\begin{equation}
ds^{2}=\frac{\rho{}^{2}}{\rho{}^{2}+x{}^{2}}\left(\delta_{\mu\nu}-\frac{x_{\mu}x_{\nu}+n{}_{\mu\sigma}^{l}n_{\nu\lambda}^{l}x{}^{\sigma}x{}^{\lambda}}{\rho{}^{2}+x{}^{2}}\right)dx{}^{\mu}dx{}^{\nu}
\end{equation}
where $\rho$ is some length scale and 
\begin{equation}
\eta_{\mu\sigma}^{l}=\begin{pmatrix}0 & 1 & 0 & 0\\
-1 & 0 & 0 & 0\\
0 & 0 & 0 & 1\\
0 & 0 & -1 & 0
\end{pmatrix}
\end{equation}
has $\chi=3$ and $\tau=-1$, see \cite{InstantonSymm}. The space
$S^{2}\times S^{2}$ has a metric, see \cite{S^2xS^2} p. 244 
\begin{equation}
ds^{2}=(1-\Lambda r^{2})dt+\frac{dr^{2}}{1-\Lambda r^{2}}+\frac{1}{\Lambda}d\Omega^{2}
\end{equation}
and $\chi=4$,$\tau=0$. Kummer's quartic surface$K^{3}$, see \cite{K3}
for a physical description of this metric, has $\chi=24$ and $\tau=16$
and the Schwarzschild solution has $\chi=2$ and $\tau=0$.

According to Hawking \cite{Hawkingbook}, p. 59, the Euclidean action
of eq. (\ref{eq:actionasdafaf}) 
\begin{equation}
I=-\frac{1}{16\pi}\int\left(R-2\Lambda\right)\sqrt{g}dx^{4}
\end{equation}
attains a minimum value of 
\begin{equation}
I_{min}=-3\pi/\Lambda
\end{equation}
on $S^{4}$. By using eqs. (\ref{eq:Lambda}), (\ref{eq:Volume})
and (\ref{eq:simpleaction}), we get for a spacetime that satisfies
Einstein's equation 
\begin{equation}
I_{min}=-\frac{3\pi}{\Lambda}=-\frac{8\pi c^{2}}{\Lambda}=c\sqrt{V}=c\sqrt{\frac{8^{2}\pi^{2}c^{2}}{\Lambda^{2}}},
\end{equation}
therefore, the constant $c$ has a lower bound of 
\begin{equation}
c\geq-\sqrt{\frac{3}{8}}\label{eq:lowerbound}
\end{equation}
From equation (\ref{eq:eulernumberdasdad}), one sees that for large
$\chi$ either $C_{\mu\nu\alpha\beta}C^{\mu\nu\alpha\beta}$ must
be large, and / or $\Lambda^{2}$ must be large. By eq. (\ref{eq:Lambda}),
a large $\Lambda^{2}$ implies a large $c^{2}$. Since by eq. (\ref{eq:lowerbound}),
$c$ is bounded from below, a large $\Lambda^{2}$ implies a positive
$c$. 

If $C_{\mu\nu\alpha\beta}C^{\mu\nu\alpha\beta}$ is large in eq. (\ref{eq:eulernumberdaffadf}),
 we get a converging effect on geodesics. We can not build our spacetime
from a microscopic model, where all geodesics will develop conjugate
points after short distances, since this would conflict with the existence
of macroscopic geodesics, whose existence is confirmed by measurements.
In order to prevent the Weyl curvature from converging, a considerably
large negative $\Lambda$ constant has to be put in, and due to eq..
(\ref{eq:Lambda}) this would imply a large constant $c$ because
c is bounded from below by eq. (\ref{eq:lowerbound}). 

From the inequality of eq.. (\ref{eq:inequality}) 
\begin{equation}
2\chi-3|\tau|\geq\frac{32}{3}c^{2}
\end{equation}
one would then expect a considerably small $|\tau|$ for large $\chi$
. Since $|\tau|$ can be considered to be small for large $\chi$
we can assume that then 
\begin{equation}
c\approx d\sqrt{\chi}\label{eq:constant}
\end{equation}
 where 
\begin{equation}
d=\frac{\sqrt{3}}{4}
\end{equation}

To compute the $\chi$ for the manifolds that give the dominant contribution
in the path integral, Hawking employs the so called zeta function
renormalization of the gravitational path integral. The following
short introduction is a review of \cite{Mukhanov} and \cite{Hawkingbook}
p. 45. For a scalar field with an Euclidean action 
\begin{equation}
S=\frac{1}{2}\int d^{4}x\sqrt{g}(g^{\mu\nu}\nabla_{\mu}\varphi\nabla_{\nu}\varphi+V\varphi^{2})
\end{equation}
zeta function renormalization of a path integral like $Z=\int\mathcal{D\varphi}e^{-S}$
goes as follows: After integration by parts, the action becomes 
\begin{eqnarray}
S & = & \frac{1}{2}\int d^{4}x(-\varphi\nabla_{\nu}(\sqrt{g}g^{\mu\nu}\nabla_{\mu}\varphi)+\sqrt{g}V\varphi^{2})\nonumber \\
 & = & \frac{1}{2}\int d^{4}x(-\varphi F\varphi)
\end{eqnarray}
 with 
\begin{equation}
F=-\square+V
\end{equation}
The quantity $F$ is a self adjoint operator which has an eigenvalue
problem 
\begin{equation}
F\varphi_{n}=\lambda_{n}\varphi_{n}
\end{equation}
Because $F$ is self adjoint, the eigenfunctions form an orthonormal
base 
\begin{equation}
\int d^{4}x\sqrt{g}\varphi_{n}\varphi_{m}=\delta_{nm}\label{eq:scalarproo}
\end{equation}
and can be expanded as 
\begin{equation}
\varphi=\sum_{n=0}^{\infty}c_{n}\varphi_{n}
\end{equation}
 where 
\begin{equation}
c_{n}=\int d^{4}x\sqrt{g}\varphi\varphi_{n}
\end{equation}
Putting the expansion for $\varphi$ into the action yields 
\begin{equation}
S=\frac{1}{2}\int d^{4}x\sqrt{g}\sum_{m,n}c_{m}c_{n}\lambda_{m}\varphi_{m}\varphi_{n}=\frac{1}{2}\sum_{n}c_{n}^{2}\lambda_{n}
\end{equation}
Once an orthonormal base of eigenfunctions is chosen, the coefficients
$c_{n}$ characterize the space of the functions $\varphi_{n}$ over
which the path integral is performed. Given that $\mathcal{D}\varphi$
must be covariant and that $c_{n}$ are coordinate independent, one
makes the guess 
\begin{equation}
\mathcal{D\varphi=}\prod_{n}f(c_{n})dc_{n}
\end{equation}
The simplest choice for $f$ would be a constant and comparison with
the measure of the path integral for flat space suggests 
\begin{equation}
\mathcal{D}\varphi=\prod_{n=0}^{\infty}\frac{c_{n}}{\sqrt{2\pi}}
\end{equation}
Then, the Euclidean path integral $Z=\int\mathcal{D\varphi}e^{-S}$
becomes 
\begin{equation}
Z=\int\prod_{n=0}^{\infty}\frac{dc_{n}}{\sqrt{2\pi}}e^{-\frac{1}{2}\lambda_{n}c_{n}^{2}}=\left(\prod_{n=0}^{\infty}\lambda_{n}\right)^{-1/2}
\end{equation}
By using the characteristic polynomial, one can show for a finite
dimensional matrix $A$ that the product of its eigenvalues $\lambda_{n}$is
equal to its determinant $det(A)$. In our case, we have a differential
operator which corresponds to an infinitely dimensional matrix. For
this, the product of the eigenvalues is infinite, and one must find
some way to regularize it. For this reason, we define a generalized
zeta function 
\begin{equation}
\zeta(s)=\sum_{n=0}^{\infty}\left(\frac{1}{\lambda_{n}}\right)^{s}
\end{equation}
It will converge for $Re(s)>2$ and can be analytically extended to
a meromorphic function of $s$ with poles only at $s=1$ and $s=2$.
The gradient of $\zeta(s)$ is formally equal to 
\begin{equation}
\zeta'(s)=\frac{d}{ds}\sum_{n=0}^{\infty}e^{-sln(\lambda_{n})}=-\sum_{n=0}^{\infty}e^{-sln(\lambda_{n})}ln(\lambda_{n})
\end{equation}
and we get
\begin{equation}
\zeta'(0)=-\sum_{n=0}^{\infty}ln(\lambda_{n})=-ln\left(\prod_{n=0}^{\infty}\lambda_{n}\right)
\end{equation}
or 
\begin{equation}
ln(Z)=-\frac{1}{2}ln\left(\prod_{n=0}^{\infty}\lambda_{n}\right)=\frac{1}{2}\zeta'(0)
\end{equation}
note that this computations are formal, because when applying mathematical
rigor, the sums $\sum_{n}ln(\lambda_{n})$ can not be handled as if
they were finite, as we did above. By analytic continuation of the
zeta function, this method should remove the problematic divergences
of the path integral, at least when the theory considered is renormalizable. 

For computation of the zeta function, one employs the notion of the
heath kernel of an operator $F(x,y,\tau)$. The heat kernel is a solution
of the generalized heath equation 
\begin{equation}
\frac{d}{d\tau}K(x,y,\tau)+FK(x,y,\tau)=0\label{eq:heateq}
\end{equation}
 where $x,y$ are points in spacetime, $\tau$ is an additional parameter,
and $F$ is an operator acting on the last argument of $K(x,y,\tau)$.
The heath kernel can be expressed as 
\begin{equation}
K(x,y,\tau)=e^{-F}
\end{equation}
 or in terms of eigenvalues of $F$ 
\begin{equation}
K(x,y,\tau)=\sum_{n=0}^{\infty}e^{-\lambda_{n}\tau}\varphi_{n}(x)\varphi_{n}(y)
\end{equation}
 since 
\begin{equation}
\frac{d}{d\tau}\sum_{n}e^{-\lambda_{n}\tau}\varphi_{n}(x)\varphi_{n}(y)=\sum_{n=0}^{\infty}(-\lambda_{n})e^{-\lambda_{n}\tau}\varphi_{n}(x)\varphi_{n}(y)=-FK(x,y,\tau)
\end{equation}
One defines the ``trace'' of the heath kernel as 
\begin{equation}
tr(K(\tau))=\int d^{4}x\sqrt{g}F(x,x,\tau)=\int d^{4}x\sqrt{g}\sum_{n=0}^{\infty}e^{-\lambda_{n}\tau}\varphi_{n}(x)\varphi_{n}(x)=\sum_{n=0}^{\infty}e^{-\lambda_{n}\tau}
\end{equation}
where we have used in the last line that 
\begin{equation}
\int d^{4}x\sqrt{g}\varphi_{n}(x)\varphi_{n}(x)=\delta_{nn}=1
\end{equation}
The generalized zeta function is related to the trace of the heath
kernel by a Mellin transformation 
\begin{equation}
\zeta(s)=\sum_{n=0}^{\infty}\lambda_{n}^{-s}=\frac{1}{\Gamma(s)}\int_{0}^{\infty}t^{s-1}tr(K(\tau)).\label{eq:afdadafdafdad}
\end{equation}
One can determine $K(x,y,\tau)$ by solving the heath equation (\ref{eq:heateq}),
then compute $tr(K(\tau))$ and from this one gets by eq. (\ref{eq:afdadafdafdad})
the generalized zeta function. For an operator $\square+\xi R$ on
a four dimensional compact manifold, deWitt \cite{DeWittzeta} computed
an expansion 
\begin{equation}
tr(K)=\sum B_{n}\tau^{n-2}
\end{equation}
where 
\begin{equation}
B_{n}=\int d^{4}xb_{n}\sqrt{g},
\end{equation}
and 
\begin{equation}
b_{0}=(4\pi)^{-2},
\end{equation}
\begin{equation}
b_{1}=(4\pi)^{-2}(\frac{1}{6}-\xi)R,
\end{equation}
 and 
\begin{equation}
b_{2}=\frac{1}{2880\pi^{2}}(R^{\mu\nu\alpha\beta}R_{\mu\nu\alpha\beta}-R^{\mu\nu}R_{\alpha\beta}+30(1-6\xi)^{2}R^{2}+(6-30\xi)\square R)
\end{equation}
. 

In their remarkable article \cite{Zetafunction}, Gibbons, Hawking
and Perry used this technique to threat the gravitational path integral
itself. One may express the metric with a classical background as
$\overline{g}_{\mu\nu}=g_{\mu\nu}+h_{\mu\nu}$ and then expand the
Euclidean action perturbatively as in section 2.1:
\begin{equation}
I(\overline{g}_{\mu\nu})=I(g_{\mu\nu})+\underline{I}(h_{\mu\nu})+\underline{\underline{I}}(h_{\mu\nu})+\text{higher order terms}
\end{equation}
where $\underline{I}(h_{\mu\nu})$ is linear and $\underline{\underline{I}}(h_{\mu\nu})$
is quadratic in the quantum field. As in Section 2.2, we have $\underline{I}(h_{\mu\nu})=0$
by the equations of motion, and so, omitting ghost and gauge fixing
terms, the Euclidean path integral in the background field method
is given by 
\begin{equation}
Z_{eu}=e^{-I(g_{\mu\nu})}\int\mathcal{D}h_{\mu\nu}e^{-\underline{\underline{I}}(h_{\mu\nu})}
\end{equation}
Then one must consider the addition of gauge fixing and ghost terms
in order to make the one loop path integral unitary. Gibbons, Hawking
and Perry find that one can express the terms for $\underline{\underline{I}}$,
gauge fixing and ghosts by determinants of certain operators $-F,G,C$
that have positive eigenvalues. Gibbons, Hawking and Perry find for
the amplitude 
\begin{equation}
ln(Z)=-I(g_{\mu\nu})-\frac{1}{2}ln\left(det\left(\frac{1}{2}\pi^{-1}\mu^{-2}(-F+G)\right)\right)+ln\left(det\left(\frac{1}{2}\pi^{-1}\mu^{-2}C\right)\right)
\end{equation}
where $\mu$ is some normalization factor, $C$ is the operator for
the ghosts and $-F+G$ is an operator for the Euclidean action and
the gauge fixing. 

For an operator whose eigenvalues are $k^{-1}\lambda_{n}$, one has
\begin{equation}
\zeta(s)=k^{s}\zeta(s)
\end{equation}
or 
\begin{equation}
\zeta'(0)=ln(k)\zeta(0)+\zeta'(0)
\end{equation}
Therefore, the one loop gravitational amplitude becomes with the zeta
functions expressing the eigenvalues of $F,G,$ and $C:$ 
\begin{equation}
ln(Z)=-I(g_{\mu\nu})+\frac{1}{2}\zeta_{F}'(0)+\frac{1}{2}\zeta_{G}'(0)-\zeta_{C}'(0)+\frac{1}{2}ln(2\pi\mu^{2})(\zeta_{F}(0)+\zeta_{G}(0)-2\zeta_{C}(0))
\end{equation}
From the expansion of the heath kernel, Gibbons, Perry, and Hawking
then find the astonishing result that 
\begin{eqnarray}
\zeta_{F}(0)+\zeta_{G}(0)-2\zeta_{C}(0) & = & \int d^{4}x\sqrt{g}\left(\frac{53}{720\pi^{2}}C_{abcd}C^{abcd}+\frac{763}{540\pi^{2}}\Lambda^{2}\right)\nonumber \\
 & = & \frac{106}{45}\chi+\frac{1168}{15}c^{2}\nonumber \\
 & \equiv & \text{\ensuremath{\gamma}}\label{eq:exponent}
\end{eqnarray}
where eqs. (\ref{eq:eulernumberdasdad}) and (\ref{eq:constant})
have been used. 

If we introduce a scale factor into the metric so that 
\begin{equation}
\tilde{\overline{g}}_{ab}=k\overline{g}_{ab}
\end{equation}
we have to consider that the background action transforms under a
change 
\begin{equation}
\tilde{g}_{ab}=kg_{ab}
\end{equation}
into 
\begin{equation}
I(\tilde{g}_{ab})=kI(g_{ab}),
\end{equation}
 and that the eigenvalues of the operators $F,G,C$ will get multiplied
by $k^{-1}.$ So we get a new amplitude of the form 
\begin{equation}
ln(\tilde{Z})=ln(Z)+(1-k)I(g_{\mu\nu})+\frac{1}{2}\gamma ln(k)\label{eq:rescaling}
\end{equation}
Compared to a solution of Einstein's equation with 
\begin{equation}
R_{\mu\nu}=g_{\mu\nu}.
\end{equation}
eq. (\ref{eq:einsteineq}) 
\begin{equation}
R_{\mu\nu}=\Lambda g_{\mu\nu}
\end{equation}
formally looks like a rescaling with a conformal factor $k=\Lambda$
and from eq. (\ref{eq:rescaling}), we can conclude that the one loop
amplitude behaves as 
\begin{equation}
Z\propto\left(\frac{\Lambda}{\Lambda_{0}}\right)^{-\gamma}\label{eq:Amplitude}
\end{equation}
where $\Lambda_{0}$ is some normalization factor related to $\mu$.
In eq. (\ref{eq:amplitude1}) we had an amplitude 
\begin{equation}
Z(\Lambda)=\int\mathcal{D}g_{\mu\nu}e^{-I}
\end{equation}
and we saw in eq. (\ref{eq:simpleaction}) that for a spacetime which
fulfills Einstein's field equations, we have 
\begin{equation}
I=-\frac{1}{16\pi}\left(\int\left(R-2\Lambda\right)\sqrt{g}dx^{4}\right)=-\frac{\Lambda V}{8\pi}=-\frac{8\pi c^{2}}{\Lambda}
\end{equation}
Using eq. (\ref{eq:Amplitude}), and eq. (\ref{eq:simpleaction}),
 Hawking approximates the action with a quantum and background part
as 
\begin{equation}
Z(\Lambda)\propto\left(\frac{\Lambda}{\Lambda_{0}}\right)^{-\gamma}e^{\frac{8\pi c^{2}}{\Lambda}}=\left(\frac{\Lambda}{\Lambda_{0}}\right)^{-\gamma}e^{\frac{8\pi d^{2}\chi}{\Lambda}}
\end{equation}
By eq. (\ref{eq:simpleaction}), the action is minimized for large
c and from eq. (\ref{eq:inequality}), it follows that large c implies
$|\tau|$ is small. Vanishing $|\tau|$ implies $\frac{6}{32}\chi\geq c^{2}$
by eq. (\ref{eq:inequality}) and we can write in eq (\ref{eq:exponent}):
\begin{equation}
\gamma\approx a\chi
\end{equation}
where $a>106/45$. Thereby the amplitude gets the simple form 
\begin{equation}
Z(\Lambda)\propto\left(\frac{\Lambda}{\Lambda_{0}}\right)^{-a\chi}e^{\frac{8\pi d^{2}\chi}{\Lambda}}\label{eq:amplitude3443}
\end{equation}
In eq. (\ref{eq:amplitude2}), we considered a formula for the number
of states between the Volume elements $V$ and $V+dV$ 
\begin{equation}
N(V)=\frac{1}{16\pi^{2}i}\int_{-i\infty}^{i\infty}d\Lambda Z(\Lambda)e^{\frac{\Lambda V}{8\pi}}
\end{equation}
 In order to estimate the topologies that give the most dominant contribution
in the path integral, a saddle point approximation with the amplitude
of eq. (\ref{eq:amplitude3443}) is employed and yields 
\begin{equation}
N(V)\approx\left(\frac{\Lambda}{\Lambda_{0}}\right)^{-a\chi}e^{\frac{8\pi d^{2}\chi}{\Lambda}+\frac{\Lambda V}{8\pi}}
\end{equation}
Setting $\frac{dN}{d\Lambda}=0$ gives 
\begin{equation}
\frac{e^{\frac{1}{8}\frac{64\pi^{2}\chi a^{2}+\Lambda^{2}V}{\Lambda\pi}}\left(\frac{\Lambda}{\Lambda_{0}}\right)^{-a\chi}(-64\pi\chi d^{2}-8\pi\Lambda a\chi+\Lambda^{2}V)}{\Lambda^{2}}=0
\end{equation}
and solving this for $\Lambda$ results in 
\begin{equation}
\Lambda_{s}=\frac{4\left(\pi a\chi\pm\sqrt{(a\pi\chi)^{2}+4\pi^{2}V\chi d^{2}}\right)}{V}\label{eq:lambdas}
\end{equation}
where, because the contour integral in eq. (\ref{eq:amplitude2})
should pass to the right of the singularity at $\Lambda_{0}$, one
has to take the positive sign of the square root. We can get the dominant
topologies in the path integral from 
\begin{equation}
\frac{dN(V)}{d\chi}|_{\Lambda=\Lambda_{s}}=0
\end{equation}
 which yields after simplification 
\begin{equation}
e^{\frac{1}{8}\frac{64\pi^{2}\chi d^{2}+\Lambda_{s}^{2}V}{\Lambda\pi}}\left(\frac{\Lambda_{s}}{\Lambda_{0}}\right)^{-a\chi}\left(a\Lambda_{s}ln\left(\frac{\Lambda_{s}}{\Lambda_{0}}\right)-8\pi d^{2}\right)=0
\end{equation}
or 
\begin{equation}
aln\left(\frac{\Lambda_{s}}{\Lambda_{0}}\right)\Lambda_{s}-8\pi d^{2}=0
\end{equation}
If $\Lambda_{0}\geq1$, this will be satisfied by 
\begin{equation}
\Lambda_{s}\approx\Lambda_{0}
\end{equation}
and if $\Lambda_{0}<1$, we get 
\begin{equation}
\Lambda_{s}\approx\Lambda_{0}^{a/(8\pi d^{2})}
\end{equation}
Setting $\Lambda_{s}=\Lambda_{0}$ into eq. (\ref{eq:lambdas}), one
arrives, after solving for $\chi$, at 
\begin{equation}
\frac{\Lambda_{0}^{2}}{8\pi(8\pi d^{2}+\Lambda_{0}a)}V=\chi
\end{equation}
or 
\[
\chi\propto hV
\]
where $h=\frac{\Lambda_{0}^{2}}{8\pi(8\pi d^{2}+\Lambda_{0}a)}$ is
some constant depending on the cutoff $\Lambda_{0}$ that is likely
to be set somewhere near the length scales of the Planck scale. A
similar expression follows with $\Lambda_{s}\approx\Lambda_{0}^{a/(8\pi d^{2})}$.
Finally, Hawking concludes: 
\begin{quotation}
This supports the picture of spacetime foam because it says that the
dominant contribution to the number of states comes from metrics with
one gravitational instanton per unit Planck volume $h^{-1}$.
\end{quotation}

\section{Canonical quantization of general relativity}

\subsection{The Wheeler deWitt equation}

In section 3, we have seen that the gravitational path integral is
dominated by metrics which describe virtual gravitational instantons.
One therefore could believe that the divergence of the two loop amplitude
is due to a general failure of perturbation theory in general relativity.
For example, Hawking writes in \cite{Hawkingbook}: 
\begin{quotation}
``Attempts to quantize gravity ignoring the topological possibilities
and simply drawing Feynman diagrams around flat space have not been
very successful. It seems to me that the fault lies not with the pure
gravity or supergravity theories themselves but with the uncritical
application of perturbation theory to them. In classical relativity
we have found that perturbation theory has only limited range of validity
One can not describe a black hole as a perturbation around flat space.
Yet this is what writing down a string of Feynman diagrams amounts
to.''
\end{quotation}
Therefore, we will describe the standard non perturbative quantization
framework for gravity in this section. The canonical quantization
of gravity provides an excellent tool especially for the quantization
of closed spacetimes, e.g. the closed Friedmann Robertson Walker universe. 

We consider the quantization of systems where the Lagrangian $L(q,\dot{q})$
with $q$ as canonical coordinate, is singular. By this, it is meant
that that the canonical momentum 
\begin{equation}
p_{i}=\frac{\partial L(q,\dot{q})}{\partial\dot{q}}
\end{equation}
can not be solved for $\dot{q}$. Gravity is such a theory and the
 Hamiltonian description of those theories were investigated at first
by Dirac \cite{Dirac1}, who presented a method for converting a gauge
field theory with singular Lagrangian into a form with a Hamiltonian.
In 1967, deWitt used this model to derive a quantum mechanical equation
for relativistic spacetimes \cite{deWitt} . In this section, we will
shortly review some aspects of this work. 

We begin with the assumption that spacetime is globally hyperbolic.
Then, one can find a time function $t$ such that each surface $t=const$
is a Cauchy surface $\Sigma$ and there exists a corresponding time
flow vector field $t^{\mu}$satisfying $t^{\mu}\nabla_{\mu}t=1$.
The metric tensor $g_{\mu\nu}$induces a spatial metric $\gamma_{ij}$
on $\Sigma_{t}$ and one can decompose the four metric $g_{\mu\nu}$
as follows: 
\begin{equation}
g_{\mu\nu}=\left(\begin{array}{cc}
-N^{2}+\beta_{k}\beta^{k} & \beta_{j}\\
\beta_{i} & \gamma_{ij}
\end{array}\right)\label{eq:ADMMetric}
\end{equation}
This metric corresponds to a line element 
\begin{eqnarray}
ds^{2} & = & -N^{2}dt^{2}+\gamma_{ij}(dx^{i}+\beta^{i}dt)(dx^{j}+\beta^{j}dt)\label{eq:ADMLINE}\\
 & = & -(N^{2}-\beta_{i}\beta^{i})dt^{2}+2\beta_{i}dx^{i}dt+\gamma_{ij}dx^{i}dx^{j}\nonumber 
\end{eqnarray}
The function $N$ is called lapse function and it is defined by
\begin{equation}
N=\frac{1}{n^{\mu}\nabla_{\mu}t}
\end{equation}
where $n^{\mu}$ is the unit normal to $\Sigma_{t}$, is called lapse
function. The three dimensional vector $\beta_{k}$ is the component
of $t^{\mu}$ tangential to $\Sigma_{t}$ and is called shift vector.
If one imagines spacetime as being foliated by a family of hypersurfaces
with $t=const$, $Ndt$ is then the proper time lapse between the
upper and lower hypersurfaces, and the shift vector gives the correspondence
between two points in the hypersurfaces.The point $(x^{i}+dx^{i},\beta^{i}dt)$
in the lower hypersurface corresponds to $(x^{i}+dx^{i},t+dt)$ in
the upper hypersurface. The spatial indices are raised and lowered
using the 3 metric $\gamma_{ij}$ and its inverse, with $\gamma_{ik}\gamma^{kj}=\delta_{i}^{j}$,
$\beta^{i}=\gamma^{ij}\beta_{j}$ $\gamma=det(\gamma_{ij})$. The
Lagrangian of gravity has the form 
\begin{equation}
L=\frac{1}{16\pi}\int d^{3}x\sqrt{-g}R
\end{equation}
In the following sections 3 and 4, we neglect the factor $1/16\pi$for
simplicity. This Lagrangian can, up to a total derivative be expressed
as, see \cite{Wald} p. 464: 
\begin{equation}
L=\int d^{3}x\sqrt{-g}R=\int d^{3}xN\sqrt{\gamma}\left(K_{ij}K^{ij}-K^{2}+^{(3)}R\right)\label{eq:admlagerangian}
\end{equation}
where 
\begin{equation}
K_{ij}=\frac{1}{2}N^{-1}(D_{j}\beta_{i}+D_{i}\beta_{j}-\partial_{t}\gamma_{ij})\label{eq:extcruv}
\end{equation}
is the extrinsic curvature of the hypersurface $x^{0}=const$ where
$D_{i}$ denotes covariant derivation with respect to the i-th direction
based on the three metric $\gamma_{ij}$ and $^{(3)}R$ is the curvature
scalar with respect to $\gamma_{ij}$. 

One can define conjugate momenta for $N,\beta_{i},\gamma_{ij}$ 
\begin{equation}
\pi=\frac{\delta L}{\delta\partial_{t}N}=0,
\end{equation}
and 
\begin{equation}
\pi^{i}=\frac{\delta L}{\delta\partial_{t}\beta_{i}}=0,
\end{equation}
and with help of 
\begin{equation}
\frac{\delta K_{ij}}{\delta\partial_{t}\gamma_{kl}}=-\frac{\delta_{ik}\delta_{jl}}{2N},
\end{equation}
and 
\begin{equation}
K=\gamma^{ij}K_{ij}
\end{equation}
we can derive 
\begin{eqnarray}
\pi^{kl} & = & \frac{\delta L}{\delta\partial_{t}\gamma_{kl}}\nonumber \\
 & = & N\sqrt{\gamma}\left(2K_{ij}\frac{\delta K_{ij}}{\delta\partial_{t}\gamma_{kl}}-2K\frac{\delta K}{\delta\partial_{t}\gamma_{kl}}\right)\nonumber \\
 & = & N\sqrt{\gamma}\left(-\frac{2K^{kl}}{2N}+\frac{2K\gamma^{kl}}{2N}\right)\nonumber \\
 & = & -\sqrt{\gamma}\left(K^{kl}-\gamma^{kl}K\right)\nonumber \\
 & = & \sqrt{\gamma}\left(\gamma^{kl}K-K^{kl}\right),\label{eq:conjmomentum}
\end{eqnarray}
The Hamiltonian is then 
\begin{eqnarray}
H & = & \int d^{3}x(\pi\partial_{t}N+\pi^{i}\partial_{t}\beta_{i}+\pi^{ij}\partial_{t}\gamma_{ij})-L\nonumber \\
 & = & \int d^{3}x(\pi^{ij}\partial_{t}\gamma_{ij})-L\nonumber \\
 & = & \int d^{3}x\left(\sqrt{\gamma}\left(\gamma^{ij}K-K^{ij}\right)\partial_{t}\gamma_{ij}\right)-L\nonumber \\
 & = & \int d^{3}x\left(-2N\sqrt{\gamma}\left(\gamma^{ij}K-K^{ij}\right)\frac{1}{2N}\left(D_{j}\beta_{i}+D_{i}\beta_{j}-\partial_{t}\gamma_{ij}\right)\right.\nonumber \\
 &  & +\left.\sqrt{\gamma}(K\gamma^{ij}-K^{ij})(D_{j}\beta_{i}+D_{i}\beta_{j})\right)-L\nonumber \\
 & = & \int d^{3}x\left(-2N\sqrt{\gamma}\left(\gamma^{ij}K-K^{ij}\right)\frac{1}{2N}\left(D_{j}\beta_{i}+D_{i}\beta_{j}-\partial_{t}\gamma_{ij}\right)\right.\nonumber \\
 &  & \left.+\pi^{ij}(D_{j}\beta_{i}+D_{i}\beta_{j})+N\sqrt{\gamma}\left(K^{2}-K_{ij}K^{ij}-^{(3)}R\right)\right)\nonumber \\
 & = & \int d^{3}x\left(-2N\sqrt{\gamma}\left(\gamma^{ij}K-K^{ij}\right)K_{ij}+N\sqrt{\gamma}\left(K^{2}-K_{ij}K^{ij}-^{(3)}R\right)+\pi^{ij}(D_{j}\beta_{i}+D_{i}\beta_{j})\right)\nonumber \\
 & = & \int d^{3}x\left(-2N\sqrt{\gamma}\gamma^{ij}KK_{ij}+2N\sqrt{\gamma}K^{ij}K_{ij}+N\sqrt{\gamma}\left(K^{2}-K_{ij}K^{ij}-^{(3)}R\right)+2\pi^{ij}D_{j}\beta_{i}\right)\nonumber \\
 & = & \int d^{3}x\left(-2N\sqrt{\gamma}K^{2}+2N\sqrt{\gamma}K^{ij}K_{ij}+N\sqrt{\gamma}\left(K^{2}-K_{ij}K^{ij}-^{(3)}R\right)+2\pi^{ij}D_{j}\beta_{i}\right)\nonumber \\
 & = & \int d^{3}x\left(N\sqrt{\gamma}(K_{ij}K^{ij}-K^{2}-{}^{3}R)-2\beta_{i}D_{j}\left(\gamma^{-1/2}\pi^{ij}\right)+2D_{i}(\gamma^{-1/2}\beta_{j}\pi^{ij})\right)\nonumber \\
 & = & \int d^{3}x\left(N\sqrt{\gamma}(K_{ij}K^{ij}-K^{2}-{}^{3}R)-2\beta_{i}D_{j}\left(\gamma^{-1/2}\pi^{ij}\right)+2D_{i}(\gamma^{-1/2}\beta_{j}\pi^{ij})\right)\label{eq:Hamiltonianconstraint}
\end{eqnarray}
The term $2D_{i}(\gamma^{-1/2}\beta_{j}\pi^{ij})$ only contributes
a boundary term to $H$, which, for finite spacetimes can be neglected
after the integration. It therefore will be dropped and we arrive
at

\begin{eqnarray}
H & = & \int d^{3}x\left(N\sqrt{\gamma}(K_{ij}K^{ij}-K^{2}-{}^{(3)}R)-\beta_{i}2D_{j}\left(\gamma^{-1/2}\pi^{ij}\right)\right)\nonumber \\
 &  & \int d^{3}x\left(N\mathcal{H_{G}}+\beta_{i}\chi^{i}\right)\label{Hamiltonoperatorconstraint2}
\end{eqnarray}
with 
\begin{equation}
\chi^{i}=2D_{j}\left(\gamma^{-1/2}\pi^{ij}\right)
\end{equation}
and 
\begin{equation}
\mathcal{H_{G}}=\sqrt{\gamma}(K_{ij}K^{ij}-K^{2}-{}^{3}R)
\end{equation}
When writing this, one should emphasize that the Lagrangian $L$ and
the Hamiltonian $H$ were derived both by omitting boundary terms.
These terms do not contribute anything to $H$ in case of finite worlds.
For asymptotically flat worlds, one must add, see \cite{deWitt},
or \cite{Wald}, p. 469, a contribution: 
\begin{equation}
E_{\infty}=\int_{\Sigma}N\sqrt{\gamma}\gamma^{ij}(\gamma_{ik,j}-\gamma_{ij,k})
\end{equation}
to the Hamiltonian. 

Since $\pi=0$ we have $\partial_{t}\pi=0$. Therefore, the Poisson
bracket yields 
\begin{equation}
\left\{ \pi,H\right\} =\partial_{t}\pi=0=\frac{\partial H}{\partial N}=\mathcal{H}_{g}\label{Hamiltonconstraint}
\end{equation}
Similarly, we have $\pi^{i}=0$, or $\partial_{t}\pi^{i}=0$, which
implies
\begin{equation}
\left\{ \pi^{i},H\right\} =\partial_{t}\pi^{i}=0=\frac{\partial H}{\partial\beta_{i}}=\chi^{i}=0\label{Diffeomorphismconstraint}
\end{equation}
Eq. (\ref{Diffeomorphismconstraint}) is associated with spatial diffeomorphism
invariance and therefore called diffeomorphism constraint. 

Using
\begin{equation}
\pi^{kl}\gamma_{kl}=\sqrt{\gamma}\left(\delta_{i}^{i}K-K^{kl}\gamma_{kl}\right)=\sqrt{\gamma}\left(\delta_{i}^{i}K-K\right)=\sqrt{\gamma}2K
\end{equation}
yields 
\begin{equation}
K=\frac{1}{2}\gamma^{-1/2}\pi^{kl}\gamma_{kl},
\end{equation}
 and we can simplify $\mathcal{H}_{G}$ even further: 
\begin{eqnarray}
\mathcal{H_{G}} & = & \sqrt{\gamma}K_{ij}K^{ij}-\sqrt{\gamma}K^{2}-\sqrt{\gamma}{}^{(3)}R\nonumber \\
 & = & \sqrt{\gamma}\left(K^{ij}-\gamma^{ij}K\right)K_{ij}+\pi^{ij}K\gamma_{ij}-\pi^{ij}K\gamma_{ij}-\sqrt{\gamma}{}^{\;(3)}R\nonumber \\
 & = & -\pi^{ij}(K_{ij}-K\gamma_{ij})-\pi^{ij}K\gamma_{ij}-\sqrt{\gamma}{}^{\;(3)}R\nonumber \\
 & = & \gamma^{-1/2}\pi^{ij}\pi_{ij}-\pi^{ij}K\gamma_{ij}-\sqrt{\gamma}{}^{\;(3)}R\nonumber \\
 & = & \gamma^{-1/2}\pi^{ij}\pi_{ij}-\frac{1}{2}\pi^{ij}\gamma_{ij}\gamma^{-1/2}\pi^{kl}\gamma_{kl}-\sqrt{\gamma}{}^{\;(3)}R\nonumber \\
 & = & \frac{1}{2}\gamma^{-1/2}(\gamma_{ik}\gamma_{jl}+\gamma_{il}\gamma_{jk}-\gamma_{ij}\gamma_{kl})\pi^{ij}\pi^{kl}-\sqrt{\gamma}^{\:(3)}R\nonumber \\
 & = & \mathcal{G}_{ijkl}\pi^{ij}\pi^{kl}-\sqrt{\gamma}^{\:(3)}R\label{eq:Hamiltondensity}
\end{eqnarray}
where 
\begin{equation}
\mathcal{G}_{ijkl}=\frac{1}{2}\gamma^{-1/2}(\gamma_{ik}\gamma_{jl}+\gamma_{il}\gamma_{jk}-\gamma_{ij}\gamma_{kl})
\end{equation}
Since $\gamma_{ij}$ and $\pi^{ij}$ are canonical coordinates. Therefore,
we have the following Poisson bracket: 
\begin{equation}
\left\{ \gamma_{ij}(x),\pi^{kl}(x')\right\} =\delta_{(i}^{k}\delta_{j)}^{l}\delta(x,x')\label{eq:inconsistency}
\end{equation}
In the quantum theory, this becomes: 
\begin{equation}
\left[\hat{\gamma}_{ij}(x),\hat{\pi}^{kl}(x')\right]=i\delta_{(i}^{k}\delta_{j)}^{l}\delta(x,x')\label{eq:commutator-1}
\end{equation}
where $\hat{\gamma}_{ij}$ and $\hat{\pi}^{ij}$ are now operators
acting on a state functional $\Psi$ that depends on the three metric
$\gamma_{ij}$ for which we will use the symbolic notation $\Psi(\gamma)$.
The relation (\ref{eq:commutator-1}) is fulfilled if 
\begin{equation}
\widehat{\gamma}_{ij}\Psi(\gamma)=\gamma_{ij}\Psi(\gamma)
\end{equation}
 and 
\begin{equation}
\hat{\pi}^{ij}\Psi(\gamma)=\frac{1}{i}\frac{\delta}{\delta\gamma_{ij}}\Psi(\gamma)\label{eq:momentumoperator}
\end{equation}
The Hamiltonian constraint \ref{Hamiltonconstraint} then becomes:
\begin{equation}
\left(\mathcal{G}_{ijkl}\frac{\delta}{\delta\gamma_{ij}}\frac{\delta}{\delta\gamma_{kl}}+\sqrt{\gamma}^{\:(3)}R\right)\Psi(\gamma)=0\label{eq:wheelerdewit}
\end{equation}
This is the Wheeler-deWitt equation, which describes a Schroedinger
like equation for the universe. And the diffeomorphism constraint
becomes 
\begin{equation}
2D_{j}\left(\gamma^{-1/2}\frac{1}{i}\frac{\delta}{\delta\gamma_{ij}}\Psi(\gamma)\right)=0.
\end{equation}

The Wheeler deWitt equation (\ref{eq:wheelerdewit}) can be approximately
solved with a semiclassical WKB-like ansatz 
\begin{equation}
\Psi=C(\gamma)e^{iS(\gamma)}\label{eq:wkb}
\end{equation}
where it is assumed that 
\begin{equation}
\left|\frac{\delta C(\gamma)}{\delta\gamma_{ij}}\right|<<\left|C(\gamma)\frac{\delta S(\gamma)}{\delta\gamma_{ij}}\right|
\end{equation}
deWitt \cite{deWitt} gets from the Hamiltonian constraint of eq.
(\ref{eq:wheelerdewit}) an equation for the phase 
\begin{equation}
\mathcal{G}_{ijkl}\frac{\delta S(\gamma)}{\delta\gamma_{ij}}\frac{\delta S(\gamma)}{\delta\gamma_{kl}}=\sqrt{\gamma}^{\:(3)}R\label{eq:einstein1}
\end{equation}
and one for the amplitude 
\begin{equation}
\frac{\delta}{\delta\gamma_{ij}}\left(C^{2}(\gamma)\frac{\delta S(\gamma)}{\delta\gamma_{kl}}\right)=0
\end{equation}
Additionally, deWitt found from the diffeomorphism 
\begin{equation}
D_{j}\left(\frac{\delta S(\gamma)}{\delta\gamma_{ij}}\right)=0\label{eq:einstein2}
\end{equation}
 and 
\begin{equation}
D_{j}\left(\frac{\delta C(\gamma)}{\delta\gamma_{ij}}\right)=0\label{eq:einstein3}
\end{equation}
From a time integration of equation (\ref{eq:extcruv}),
\begin{equation}
-\frac{\partial\gamma_{ij}}{\partial x^{0}}=2NK_{ij}-D_{j}\beta_{i}-D_{i}\beta_{j}
\end{equation}
one can determine the four geometry of spacetime. By Setting 
\begin{equation}
\pi^{ij}=\frac{\delta S(\gamma)}{\delta\gamma_{ij}}
\end{equation}
 and noting that 
\begin{equation}
\pi^{ij}=-\sqrt{\gamma}(K^{ij}-\gamma^{ij}K)
\end{equation}
deWitt gets the equation 
\begin{equation}
\frac{\partial\gamma_{ij}}{\partial x^{0}}=2N\mathcal{G}_{ijkl}\frac{\delta S(\gamma)}{\delta\gamma_{kl}}-D_{j}\beta_{i}-D_{i}\beta_{j}
\end{equation}
Differentiating this equation by $x^{0}$, deWitt derived a set of
differential equations that together with eqs. (\ref{eq:einstein1}),
(\ref{eq:einstein2}) and (\ref{eq:einstein3}) turned out to be equivalent
to the classical Einstein equations. Hence, the ansatz of eq. (\ref{eq:wkb})
is indeed similar to a WKB approximation.

The Hamiltonian generates time translations but it vanishes in general
relativity. To get a dynamical theory, deWitt proposed to use wave
packets in form of superpositions like 
\begin{equation}
\Psi=Ce^{-iS}+Ce^{iS}
\end{equation}
 In his article, deWitt used such an ansatz to describe the properties
of the quantized Friedmann Robertson Walker universe \cite{deWitt}.

\subsection{Problems of canonical quantum gravity, occurrence of inconsistencies
and infinities at low distance physics.}

Unfortunately, until now, no exact solutions of the Wheeler deWitt
equation have been found. In his article \cite{deWitt}, deWitt also
noted that the quantization leads to a severe inconsistency even without
considering the problems with the definition of a scalar product.
Contracting all indices in eq. (\ref{eq:inconsistency}) and setting
$x'=x$ yields 
\begin{equation}
\left[\hat{\gamma}_{ij}(x),\hat{\pi}^{ij}(x)\right]=6i\delta(x,x)
\end{equation}
And therefore, we have 
\begin{equation}
\left[6i\hbar\delta(x,x),i\int\chi_{k'}\delta\zeta^{k'}d^{3}x'\right]=0.\label{eq:commutator}
\end{equation}
where denotes $\delta\zeta^{k}$ an infinitesimal displacement. On
the other hand, deWitt computed the commutators (the notation $,k$
means partial differentation with respect to $x^{k}$ ) 
\begin{equation}
\left[\hat{\gamma}_{ij},i\int\chi_{k'}\delta\zeta^{k'}d^{3}x'\right]=-\hat{\gamma}_{ij,k}\delta\zeta^{k}-\hat{\gamma}_{kj}\delta\zeta_{,i}^{k}-\hat{\gamma}_{ik}\delta\zeta_{,j}^{k}
\end{equation}
and 
\begin{equation}
\left[\hat{\pi}^{ij},i\int\chi_{k'}\delta\zeta^{k'}d^{3}x'\right]=-\left(\hat{\pi}^{ij}\delta\zeta^{k}\right)_{k}+\hat{\pi}^{kj}\delta\zeta_{,k}^{i}+\hat{\pi}^{ik}\delta\zeta_{,k}^{j}
\end{equation}
Using 
\begin{equation}
[A-B,C]=[A,C]-[B,C]
\end{equation}
and 
\begin{equation}
[AB,C]=A[B,C]+[A,C]B
\end{equation}
we can evaluate the commutator in eq. (\ref{eq:commutator}) as 
\begin{eqnarray}
\left[\left[\hat{\gamma}_{ij},\hat{\pi}^{ij}\right],i\int\chi_{k'}\delta\zeta^{k'}d^{3}x'\right] & = & \left[\left(\hat{\gamma}_{ij}\hat{\pi}^{ij}-\hat{\pi}^{ij}\hat{\gamma}_{ij}\right),i\int\chi_{k'}\delta\zeta^{k'}d^{3}x'\right]\nonumber \\
 & = & \left[\hat{\gamma}_{ij}\hat{\pi}^{ij},i\int\chi_{k'}\delta\zeta^{k'}d^{3}x'\right]-\left[\hat{\pi}^{ij}\hat{\gamma}_{ij},i\int\chi_{k'}\delta\zeta^{k'}d^{3}x'\right]\nonumber \\
 & = & \hat{\gamma}_{ij}\left[\hat{\pi}^{ij},i\int\chi_{k'}\delta\zeta^{k'}d^{3}x'\right]+\left[\hat{\gamma}_{ij},i\int\chi_{k'}\delta\zeta^{k'}d^{3}x'\right]\hat{\pi}^{ij}\nonumber \\
 &  & -\hat{\pi}^{ij}\left[\hat{\gamma}_{ij},i\int\chi_{k'}\delta\zeta^{k'}d^{3}x'\right]-\left[\hat{\pi}^{ij},i\int\chi_{k'}\delta\zeta^{k'}d^{3}x'\right]\hat{\gamma}_{ij}\nonumber \\
 & = & \hat{\gamma}_{ij}\left((-\hat{\pi}^{ij}\delta\zeta^{k})_{,k}+\hat{\pi}^{kj}\delta\zeta_{,k}^{i}+\hat{\pi}^{ik}\delta\zeta_{,k}^{j}\right)\nonumber \\
 &  & +\left(-\hat{\gamma}_{ij,k}\delta\zeta^{k}-\hat{\gamma}_{kj}\delta\zeta_{,i}^{k}-\hat{\gamma}_{ik}\delta\zeta_{,j}^{k}\right)\hat{\pi}^{ij}\nonumber \\
 &  & -\hat{\pi}^{ij}\left(-\hat{\gamma}_{ij,k}(x)\delta\zeta^{k}-\hat{\gamma}_{kj}\delta\zeta_{,i}^{k}-\hat{\gamma}_{ik}\delta\zeta_{,j}^{k}\right)\nonumber \\
 &  & -(-\hat{\pi}^{ij}\delta\zeta^{k})_{,k}+\hat{\pi}^{kj}\delta\zeta_{,k}^{i}+\hat{\pi}^{ik}\delta\zeta_{,k}^{j}\hat{\gamma}_{ij}
\end{eqnarray}
simplifying further, we get
\begin{eqnarray}
 & = & \hat{\gamma}_{ij}(-\hat{\pi}^{ij}\delta\zeta^{k})_{,k}+\hat{\gamma}_{ij}\hat{\pi}^{kj}\delta\zeta_{,k}^{i}+\hat{\gamma}_{ij}\hat{\pi}^{ik}\delta\zeta_{,k}^{j}\nonumber \\
 &  & -\hat{\gamma}_{ij,k}\delta\zeta^{k}\hat{\pi}^{ij}-\hat{\gamma}_{kj}\delta\zeta_{,i}^{k}\hat{\pi}^{ij}-\hat{\gamma}_{ik}\delta\zeta_{,j}^{k}\hat{\pi}^{ij}\nonumber \\
 &  & +\hat{\pi}^{ij}\hat{\gamma}_{ij,k}\delta\zeta^{k}+\hat{\pi}^{ij}\hat{\gamma}_{kj}\delta\zeta_{,i}^{k}+\hat{\pi}^{ij}\hat{\gamma}_{ik}\delta\zeta_{,j}^{k}\nonumber \\
 &  & +(\hat{\pi}^{ij}\delta\zeta^{k})_{,k}\hat{\gamma}_{ij}-\hat{\pi}^{kj}\delta\zeta_{,k}^{i}\hat{\gamma}_{ij}-\hat{\pi}^{ik}\delta\zeta_{,k}^{j}\hat{\gamma}_{ij}\nonumber \\
 & = & \left(\hat{\gamma}_{ij}\hat{\pi}^{kj}-\hat{\pi}^{kj}\hat{\gamma}_{ij}\right)\delta\zeta_{,k}^{i}+\left(\hat{\pi}^{ij}\hat{\gamma}_{ik}-\hat{\gamma}_{ik}\hat{\pi}^{ij}\right)\delta\zeta_{,j}^{k}\nonumber \\
 &  & +\left(\hat{\pi}^{ij}\hat{\gamma}_{kj}-\hat{\gamma}_{kj}\hat{\pi}^{ij}\right)\delta\zeta_{,i}^{k}+\left(\hat{\gamma}_{ij}\hat{\pi}^{ik}-\hat{\pi}^{ik}\hat{\gamma}_{ij}\right)\delta\zeta_{,k}^{j}\nonumber \\
 &  & +\left(\hat{\pi}^{ij}\hat{\gamma}_{ij,k}-\hat{\gamma}_{ij,k}\hat{\pi}^{ij}\right)\delta\zeta^{k}+(\hat{\pi}^{ij}\delta\zeta^{k})_{,k}\hat{\gamma}_{ij}-\hat{\gamma}_{ij}(\hat{\pi}^{ij}\delta\zeta^{k})_{,k}\nonumber \\
 & = & \delta_{ij}^{kj}\delta\zeta_{,k}^{i}-\delta_{ik}^{ij}\delta\zeta_{,j}^{k}-\delta_{kj}^{ij}\delta\zeta_{,i}^{k}+\delta_{ij}^{ik}\delta\zeta_{,k}^{j}+\left(\hat{\pi}^{ij}\hat{\gamma}_{ij,k}-\hat{\gamma}_{ij,k}\hat{\pi}^{ij}\right)\delta\zeta^{k}\nonumber \\
 &  & +(\hat{\pi}^{ij}\delta\zeta^{k})_{,k}\hat{\gamma}_{ij}-\hat{\gamma}_{ij}(\hat{\pi}^{ij}\delta\zeta^{k})_{,k}\nonumber \\
 & = & \left(\hat{\pi}^{ij}\hat{\gamma}_{ij,k}-\hat{\gamma}_{ij,k}\hat{\pi}^{ij}\right)\delta\zeta^{k}+(\hat{\pi}^{ij}\delta\zeta^{k})_{,k}\hat{\gamma}_{ij}-\hat{\gamma}_{ij}(\hat{\pi}^{ij}\delta\zeta^{k})_{,k}\nonumber \\
 & = & -6i\left(\delta(x,x)\delta\zeta^{k}\right)_{,k}
\end{eqnarray}
All in all, we have: 
\begin{eqnarray}
\left[\left[\hat{\gamma}_{ij}(x),\hat{\pi}^{ij}(x)\right],i\int\chi_{k'}\delta\zeta^{k'}d^{3}x'\right] & = & \left[6i\delta(x,x),i\int\chi_{k'}\delta\zeta^{k'}d^{3}x'\right]\label{paradox}\\
 & = & 0\nonumber \\
 & = & -6i\left(\delta(x,x)\delta\zeta^{k}\right)_{,k}\nonumber 
\end{eqnarray}
which is a contradiction. The delta function is a distribution that,
according to deWitt, ``may, without inconsistency, be taught as a
limit of a sequence of successively narrower twin peaked functions,
all of which are smooth, have unit integral and vanish at point $x'=x$
in the valley between the peaks'' 
\begin{equation}
\delta(x)=\lim_{\epsilon\rightarrow0}\frac{1}{2\pi}\left(f_{\epsilon}(x-\sqrt{\epsilon})+f_{\epsilon}(x+\sqrt{\epsilon})-\frac{2f_{\epsilon}(x)}{1+\epsilon}\right)
\end{equation}
where 
\[
f_{\epsilon}(x)=\frac{\epsilon}{x^{2}+\epsilon^{2}}
\]
In his article\cite{deWitt}, deWitt notes on p. 1121: ``In an infinite
world, passage to $\epsilon\rightarrow0$ would correspond to the
usual cutoff going to infinity in momentum space''. On p.1120, deWitt
writes that the then appearing inconsistency from eq. (\ref{paradox})
``bears on problems of interpreting divergences''. Apparently, the
canonical version of quantum gravity becomes inconsistent at high
energies. The article by Goroff and Sagnotti \cite{Goroff1,Goroff2}
that demonstrated inconsistencies of covariant quantum gravity at
high momentum was published in the year 1985. It seems that deWitt
arrived at a similar conclusion 18 years earlier.

\subsection{A short review of Loop quantum gravity and its problems}

One attempt to regularize the Hamiltonian is the so-called loop quantum
gravity. In the following, we give a short review of these methods.
We can, however, not give all details on the rather involved techniques.
We confine us here to simply stating the main results of this regularization
procedure. The text below in this section can be viewed as a summary
of \cite{Nicolai} and sect 4.3 and 6 of \cite{Kiefer}, also \cite{Rovelli}
was of some help. The interested reader is referred to the excellent
introductions \cite{Nicolai,Rovelli,Tiemann,Kiefer} .

First, one writes the Hamiltonian in so called Ashtekar variables.
Using an orthonormal basis $e_{i}^{a}(x)$ where $i,a\in1,2,3$ one
can define 
\begin{equation}
E_{i}^{a}(x)=|det(e_{a}^{i})|e_{i}^{a}(x)
\end{equation}
and a connection 
\begin{equation}
A_{a}^{i}(x)=\Gamma_{a}^{i}(x)+\gamma K_{a}^{i}(x),\label{eq:gausconstraintprev}
\end{equation}
where 
\begin{equation}
\Gamma_{a}^{i}(x)=-\frac{1}{2}\omega_{ajk}\epsilon^{ijk},
\end{equation}
is the Levi-Civita Connection with $\epsilon^{ijk}$ as anti symmetric
tensor and the spin connection 
\begin{equation}
\omega_{a\; j}^{\; i}=\Gamma_{kj}^{i}e_{a}^{k}
\end{equation}
and 
\begin{equation}
K_{a}^{i}=K_{ab}(x)e^{bi}(x)
\end{equation}
is the extrinsic curvature on the three dimensional spacetime manifold.
The parameter $\gamma$ in Eq. (\ref{eq:gausconstraintprev}) is called
Barbero Immirizi parameter. $E_{i}^{a}$ and $A_{a}^{i}$ are canonically
conjugate variables and the Hamiltonian constraint of eq. (\ref{Hamiltonconstraint})
expressed in Ashtekar variables, becomes after some calculation, see
\cite{Nicolai}: 
\begin{equation}
\mathcal{H_{G}}=\frac{E_{a}^{m}E_{b}^{n}}{\sqrt{|det(E_{i}^{a})|}}\left(\epsilon^{abc}F_{mnc}-\frac{1}{2}(1+\gamma^{2})K_{[m}^{\;\, a}K_{n]}^{\; b}\right)
\end{equation}
where 
\begin{eqnarray}
F_{mnc} & = & \partial_{m}A_{nc}-\partial_{n}A_{mc}+\epsilon_{cvw}A_{mv}A_{nw}\nonumber \\
 & = & -\frac{1}{2}\epsilon_{cuv}R_{mnuv}+\gamma(D_{m}K_{nc}-D_{n}K_{mc})+\gamma^{2}\epsilon_{cuv}K_{m}^{\; u}K_{n}^{\; v}
\end{eqnarray}
is the field strength of $A$. This could be simplified by Thiemann,
see \cite{ThiemannHamil,Nicolai} to an expression for the Hamiltonian
operator 
\begin{eqnarray}
H_{G} & = & \int d^{3}xN\epsilon^{mnp}tr\left(F_{mn}\left\{ A_{p},V\right\} -\frac{1+\gamma^{2}}{2}\{A_{m},K\}\{A_{n},K\}\{A_{p},V\}\right)\nonumber \\
 & = & H_{1}+H_{2}\label{eq:loophamiltonian2}
\end{eqnarray}
where 
\begin{equation}
H_{1}=\int d^{3}xN(x)\epsilon^{mnp}tr\left(F_{mn}\left\{ A_{p},V\right\} \right)
\end{equation}
and$N(x)$ is the lapse function, $A_{p}=\tau_{i}A_{a}^{i}$, and
$\tau_{i}=i\sigma_{i}/2$ with $\sigma_{i}$ being Pauli matrices,
\begin{equation}
V=\int_{\Sigma}\sqrt{\frac{1}{3!}\epsilon_{abc}\epsilon^{mnp}E_{m}^{a}E_{n}^{b}E_{p}^{c}}
\end{equation}
is the volume, and 
\begin{equation}
K=\int_{\Sigma}d^{3}xK_{a}^{i}E_{i}^{a}
\end{equation}
is the trace of the extrinsic curvature. 

Loop quantum gravity considers quantum states on a Hilbert space of
so-called spin networks. The Hamiltonian gets converted to an operator
on such spin network states, which allows a proper regularization.
We call a function $\lambda:[0,1]\in\mathbb{R}\rightarrow\Sigma,s\mapsto\{\lambda^{a}(s)\}$
a link or edge. Then, $h_{\lambda}[A](s)\in SU(2)$ is called a holonomy
along $\lambda$ corresponding to $A_{a}=A_{a}^{i}\tau_{i}$ if it
fulfills 
\begin{equation}
h_{\lambda}[A](0)=\mathbf{1}
\end{equation}
 and 
\begin{equation}
\frac{d}{ds}h_{\lambda}[A](s)-A_{a}(\lambda(s))\frac{d\lambda^{a}(s)}{ds}h_{\lambda}[A](s)
\end{equation}
 The formal solution of this differential equation is 
\begin{equation}
U[A,\lambda]=1+\int_{0}^{1}dsA(\lambda(s))+\int_{0}^{1}ds\int_{0}^{1}dtA(\lambda(t))A(\lambda(s))+\ldots
\end{equation}
The holonomies transform in matrix valued $SU(2)$ representations$\rho_{j_{\lambda}}$
for arbitrary spins $j_{\lambda}=\frac{1}{2},1,\frac{3}{2}$. We will,
following Nicolai \cite{Nicolai}, denote them by $(\rho_{j_{\lambda}},h_{\lambda}[A])_{\alpha\beta}$
where $\alpha,\beta$ are indices of the $SU(2)$ representation. 

A spin network is a graph $\Gamma$ of finitely many vertices $v_{i}$
that are connected finitely many edges $\lambda_{i}$, where each
edge is associated with a holonomy. With a function $\psi_{n}:[SU(2)]^{n}\rightarrow\mathbb{C},$
we can define a cylindrical function 
\begin{equation}
\Psi_{\Gamma\psi}[A]=\psi\left((\rho_{j_{\lambda_{1}}},h_{\lambda_{1}}[A])_{\alpha_{1}\beta_{1}},(\rho_{j_{\lambda_{2}}},h_{\lambda_{2}}[A])_{\alpha_{2}\beta_{2}},\ldots,(\rho_{j_{\lambda_{n}}},h_{\lambda_{n}}[A])_{\alpha_{n}\beta_{n}}\right).
\end{equation}
When going to Ashtekar variables, one can derive from eq. (\ref{eq:gausconstraintprev})
a so called Gauss constraint: 
\begin{equation}
\partial_{m}E_{a}^{m}+\epsilon_{abc}A_{m}^{\; b}E^{cm}=0
\end{equation}
This constraint is satisfied for SU(2) invariant cylindrical functions.
These are constructed from holonomies whose spins $j$ obey the Clebsch-Gordan
rules for pairs of spins of the holonomies that are connected at each
vertex. The SU(2) invariant cylindrical function is then defined by
contracting the SU(2) indices of the holonomies at each vertex with
Clebsch-Gordan coefficients. 

For example, for a graph with three edges such that one vertex connects
all of these edges, the SU(2) gauge invariant spin state is 
\begin{equation}
\Psi_{\Gamma\left\{ ,J,C\right\} }[A]=(\rho_{j_{\lambda_{1}}},h_{\lambda_{1}}[A])_{\alpha_{1}\beta_{1}}(\rho_{j_{\lambda_{2}}},h_{\lambda_{2}}[A])_{\alpha_{2}\beta_{2}}(\rho_{j_{\lambda_{n}}},h_{\lambda_{n}}[A])_{\alpha_{3}\beta_{3}}C{}_{\beta_{1}\beta_{2}\beta_{3}}^{j_{\lambda_{1}}j_{\lambda_{2}}j_{\lambda_{3}}}
\end{equation}
where $C{}_{\beta_{1}\beta_{2}\beta_{3}}^{j_{\lambda_{1}}j_{\lambda_{2}}j_{\lambda_{3}}}$
are the Clebsch-Gordan coefficients of the representations connected
to the edges $\lambda_{1},\ldots,\lambda_{2}$.For graphs where a
vertex connects more than three links, different choices for the state
functions are possible, see \cite{Nicolai}. 

Given a graph $\Gamma$ and two gauge invariant cylindrical functions
on this graph, one can then define a norm 
\begin{equation}
\langle\Psi_{\Gamma\left\{ J,C\right\} }[A]|\Psi_{\Gamma,\left\{ J',C'\right\} '}[A']\rangle=\int\Pi_{\lambda_{n}\in\Gamma}dh_{\lambda_{n}}\Psi_{\Gamma,\left\{ J,C\right\} }[A]\Psi_{\Gamma\left\{ J',C'\right\} }[A'].
\end{equation}
where $dh_{\lambda_{n}}$ is the SU(2) Haar measure. The scalar product
for states with different graphs is defined to be zero. From this
norm, one can construct a Hilbert space $\mathcal{H}_{G}$ of gauge
invariant states, if one considers only those states which have finite
norm. 

The diffeomorphism constraint of canonical quantum gravity is then
implemented with a subspace of $\mathcal{H}_{G}$ . Following Nicolai,
we define for a state $|\Psi_{\Gamma,\left\{ J,C\right\} }[A]\rangle\in\mathcal{H}_{G}$
the sum 
\begin{equation}
\eta(\Psi_{\Gamma,\left\{ J,C\right\} }[A])\equiv\sum_{\phi\in Diff\backslash\Gamma}\Psi_{\Gamma,\left\{ J,C\right\} }[A\circ\phi]
\end{equation}
where $Diff\backslash\Gamma$ denotes the set of diffeomorphisms on
the three dimensional spacetime manifold that do not leave $\Gamma$
invariant. The sum 
\begin{equation}
\eta(\Psi_{\Gamma',\left\{ J',C'\right\} }[A])|\Psi_{\Gamma,\{J,C\}}[A]\rangle=\sum_{\phi\in Diff\backslash\Gamma'}\langle\Psi_{\Gamma',\left\{ J',C'\right\} }[A\circ\phi]|\Psi_{\Gamma,\{J,C\}}[A]\rangle
\end{equation}
consists of only finite terms, since the scalar product will be zero
if the graph$\Gamma'\circ\phi$ is different from the graph $\Gamma$.
If the graphs do not differ, then $\phi$ can only change the orientation
or order of the vertices, but since there are finite vertices in each
graph, and the scalar product is finite for $|\Psi_{\Gamma,\left\{ J,C\right\} }[A]\rangle\in\mathcal{H}_{G}$,
the sum must be finite. Since $\phi$ is in the set of diffeomorphisms
there are only contributions if $\Gamma$ and $\Gamma'$ are diffeomorphic.
By 
\begin{equation}
\langle\eta(\Psi_{\Gamma',\left\{ J',C'\right\} }[A])|\eta(\Psi_{\Gamma,\{J,C\}}[A])\rangle\equiv\langle\eta(\Psi_{\Gamma,\left\{ J,C\right\} }[A])|\Psi_{\Gamma,\{J,C\}}[A]\rangle
\end{equation}
we then get a finite norm for the diffeomorphism invariant states
from which we can construct $\mathcal{H}_{diff}\subset\mathcal{H}_{G}$.

With the area element 
\[
dF_{a}=\epsilon_{mnp}E_{a}^{m}dx^{n}dx^{p}
\]
the flux-vector 
\[
F_{S}^{a}(E)=\int_{S}dF^{a}
\]
 for a surface $S$ in the three dimensional manifold is the conjugate
variable to a holonomy $h_{\lambda}[A]$ and one can compute the Poisson
bracket as 
\begin{equation}
\left\{ (\rho_{j_{\lambda}}h_{\lambda}[A])_{\alpha\beta},F_{S}^{a}(E)\right\} =\iota(\lambda.S)\gamma\left(\rho_{j_{\lambda_{1}}}h_{\lambda_{1}}[A]\right)\tau^{a}\left(\rho_{j_{\lambda_{2}}}h_{\lambda_{2}}[A]\right)_{\alpha\beta}
\end{equation}
where $\lambda_{1}$ and $\lambda_{2}$ are the parts of the curve
that lie on different sides of the surface $S$ and the intersection
number $\iota(\lambda,S)$ is defined by: 
\begin{equation}
\iota(\lambda,S)=\int_{\lambda}dx^{m}\int_{S}dy^{n}dy^{p}\epsilon_{mnp}\delta^{3}(x,y)\in\left\{ \pm1,0\right\} 
\end{equation}
In the quantum theory, the Poisson brackets become commutators 
\begin{equation}
\left[\widehat{(\rho_{j_{\lambda}}h_{\lambda}[A])}_{\alpha\beta},\widehat{F}_{S}^{a}(E)\right]=il_{p}^{2}\iota(\lambda.S)\gamma\left(\rho_{j_{\lambda_{1}}}h_{\lambda_{1}}[A]\right)_{\alpha\beta}\tau^{a}\left(\rho_{j_{\lambda_{2}}}h_{\lambda_{2}}[A]\right)_{\alpha\beta}\label{poissonbr}
\end{equation}
where $l_{p}=1,62*10^{-33}cm$ is the Planck length. 

If an holonomy operator $\widehat{(\rho_{j_{\lambda}}h_{\lambda}[A])}$
acts on a spin network state, the holonomy just ads some edge to the
graph that is 
\begin{equation}
\widehat{(\rho_{j_{\lambda}}h_{\lambda}[A])}_{\alpha\beta}|\Psi_{\Gamma,\{J,C\}}[A]\rangle=(\rho_{j_{\lambda}}h_{\lambda}[A])_{\alpha\beta}|\Psi_{\Gamma,\{J,C\}}[A]\rangle
\end{equation}
The action of a flux on a spin network state is zero if the surface
$S$ does not intersect the Graph of $|\Psi_{\Gamma,\{J,C\}}[A]\rangle$
. If $S$ intersects the graph on an edge, a spin matrix $\tau_{\alpha}$
is inserted and the spin network changes as: 
\begin{eqnarray}
\widehat{F}_{S}^{a}(E)\left(,\dots,\left(\rho_{j_{\lambda}}h_{\lambda}[A]\right)_{\alpha\beta}\dots,\right)\;\;\;\;\;\;\;\;\nonumber \\
=8\pi il_{p}^{2}\iota(\lambda,S)\gamma\left(,\dots,\left(\rho_{j_{\lambda_{1}}}h_{\lambda_{1}}[A]\right)_{\alpha\beta}\left(\tau_{i}\right)_{\alpha\beta}\left(\rho_{j_{\lambda_{2}}}h_{\lambda_{2}}[A]\right)_{\alpha\beta},\dots,\right)\label{eq:intersect}
\end{eqnarray}
where $\lambda=\lambda_{1}\cup\lambda_{2}$ and $\left(\tau_{i}\right)_{\alpha\beta}$
is in the $SU(2)$ generator representation corresponding to $j$.
If $S$ intersects $\lambda$ at a vertex, the result depends on the
position and orientation of the surface with respect to the edge,
as well as on the choice for the Clebsch Gordan coefficients, see
\cite{Nicolai}. Unfortunately, this makes the computations in LQG
highly dependent on the vertex triangulation that is chosen for $\Sigma.$ 

Using the formula $A(S)=\int_{S}\sqrt{dF^{a}dF^{a}}$ for the area
of a surface $S$, we can define an area operator by replacing the
integral by a sum made of $N$ infinitesimal surfaces $S_{l}$ with
$S=\cup_{l}^{N}S_{l}$. 
\begin{equation}
\widehat{A}(S)=\lim_{N\rightarrow\infty}\sum_{l=1}^{N}\sqrt{\left|\widehat{F}_{S_{l}}^{a}(E)\widehat{F}_{S_{l}}^{a}(E)\right|}
\end{equation}
Using a coordinate system $\Gamma$, where each of the small areas
are intersected by only one edge, we find by eq. (\ref{eq:intersect})
and 
\begin{equation}
\left(\tau_{i}\right)_{\alpha\beta}\left(\tau_{i}\right)_{\alpha\beta}=-j_{\lambda}(j_{\lambda}+1)\mathbf{1}
\end{equation}
that 
\begin{equation}
\widehat{A}(S)\Psi_{\Gamma,\{J,C\},}=8\pi l_{p}^{2}\gamma\sum_{p=1}^{N(\Gamma)}\sqrt{j_{p}(j_{p}+1)}\Psi_{\Gamma,\{J,C\},}
\end{equation}
where $N(\Gamma)$ is the number of links in $\Gamma$. In eq. (\ref{poissonbr}),
no functions are present that could become infinite. This together
with fact that the area has a discrete spectrum in the spin network
states implies that in loop quantum gravity, one can not run into
problems with infinities as with the Wheeler deWitt equation in section
2.2. However, the result for the operators depend highly on the coordinate
system of the graph that $\Psi_{\Gamma,\{J,C\},}$ is associated with,
e.g it depends on the number of links $N(\Gamma)$ in the graph, and
the position of the nodes as well as the spin representations. 

Like the area operator, one can define the volume operator of a three
dimensional volume $\Omega$ as a finite Riemann sum of $N$ small
volume elements $\Omega_{k}$ with $\Omega=\cup_{l}^{N}\Omega_{l}\Omega$:
\begin{equation}
V(\Omega)=\lim_{N\rightarrow\infty}\sum_{l=1}^{N}\sqrt{\left|\frac{1}{3}\epsilon_{abc}\widehat{F}_{S_{l}^{1}}^{a}(E)\widehat{F}_{S_{l}^{2}}^{a}(E)\widehat{F}_{S_{l}^{3}}^{a}(E)\right|}
\end{equation}
where $S_{l}^{a}:S_{l}^{1}\cup S_{l}^{2}\cup S_{l}^{3}=\Omega_{l}$
are three non-coincident surfaces. When computing this operator, one
must, however, choose the coordinate system appropriately, since otherwise
the operator would either diverge, or vanish, see \cite{Nicolai}
p. 29. 

Now we fix a point $x$ on the three manifold and a tangent vector
$u$ at $x$ and consider a path $\lambda_{x,u}$ of length $\epsilon$
starting at $x$ tangent to $u$. Then the holonomy for this configuration
can be expanded as: 
\begin{equation}
h_{\lambda_{x,u}}[A]=1+\epsilon u^{a}A_{a}(x)+\mathcal{O}(\epsilon^{2})
\end{equation}
 see \cite{Rovelli,Tiemann,ThiemannHamil}. Similarly, we fix two
tangent vectors $u,v$ at $x$ and consider the triangular loop $x,u,v$
denoted by $\lambda_{x,u,v}$ with one vertex at $x$, and two edges
each of length $\epsilon$ and tangent to $u,v$, then 
\begin{equation}
h_{\lambda_{x,u,v}}[F]=1+\frac{1}{2}\epsilon^{2}u^{a}v^{b}F_{ab}(x)+\mathcal{O}(\epsilon^{3})
\end{equation}
We can now define tetrahedra of edge length $\epsilon$ with tangents
$u_{1},u_{2},u_{3}$ at each point $x$ whose triple product $u_{1}(u_{2}\times u_{3})=1$.
Thiemann noted that in \cite{ThiemannHamil} that the $H_{1}$ part
of the Hamiltonian in eq. (\ref{eq:loophamiltonian2}) can then be
written as 
\begin{equation}
H_{1}=\int d^{3}xN\mathcal{H}_{G}=\lim_{\epsilon\rightarrow0}\frac{1}{\epsilon^{3}}\int d^{3}xN(x)\epsilon^{ijk}tr\left(h_{\lambda_{x,u_{i},u_{j}}}h_{\lambda_{x,u_{k}}}\left\{ h_{\lambda_{x,u_{k}}}^{-1},V\right\} \right)
\end{equation}
the integral can be replaced by a Riemannian sum of small three dimensional
regions $\Omega_{m}$ of volume $\epsilon^{3}$.with $x_{m}$being
an arbitrary point in $\Omega_{m}$ 
\begin{equation}
H_{1}=\lim_{\epsilon\rightarrow0}\frac{1}{\epsilon^{3}}\sum_{m}\epsilon^{3}N(x_{m})\epsilon^{ijk}tr\left(h_{\lambda_{x_{m},u_{i},u_{j}}}h_{\lambda_{x_{m},u_{k}}}\left\{ h_{\lambda_{x_{m},u_{k}}^{-1}},V(\Omega_{m})\right\} \right)
\end{equation}
 Finally, $\epsilon$ drops out, and we can replace the Poisson bracket
in the above sum by commutators. 

We now have to choose the points $x_{m}$, the tangents $u_{1},u_{2},u_{3}$
and the paths $\lambda_{x,u,v}$ and $\lambda_{x,u}$ such that the
operator $H_{1}$ is well defined, gauge invariant and nontrivial.
For $H_{1}$, Thiemann found that such a choice can be made, see \cite{Rovelli},
p 279 or Thiemann's article \cite{ThiemannHamil}. So we end up with
a well defined first part of the Hamilton operator that can act on
a spin network state. Unfortunately, the second part, $H_{2}$, of
the Hamiltonian is omitted in the loop quantum gravity literature,
due to computational complexity. If one chooses the Barbero Immirizi
parameter to be $\gamma=\pm i$ the second part of the Hamiltonian
drops out. However, Nicolai notes on p. 6 of \cite{Nicolai2} that
such a choice would lead to the phase space of general relativity
to be complexified. Then a reality constraint would have to be imposed
in order to recover the original phase space. Nicolai writes that
quantizing this reality constraint would lead to additional difficulties.

Even if we assume that it is possible to give a meaningful interpretation
of the second part of the Hamiltonian or to solve the reality constraint,
loop quantum gravity requires a specific choice of the spin network
graph, i.e. a special coordinate system of points has to be chosen.
Some people might think that this procedure is at odds with the principles
of general relativity. Furthermore, the results of this theory seem
to depend on the spin representation that was chosen for the holonomies.
Since these representations are arbitrary, there is a quantization
ambiguity incorporated in loop quantum gravity. 

Finally, the Wheeler deWitt equation (\ref{eq:wheelerdewit}) admits
an approximate solution with superpositions of WKB states. We noted
in section 4.1 that by using these approximate solutions, deWitt was
able to show a connection of the Wheeler deWitt theory to Einstein's
equations of classical general relativity. Unfortunately, it is not
clear at all, how to derive the WKB states of the Wheeler deWitt theory
from some semi-classical limit of the spin network states that one
has in loop quantum gravity. Worse, no one has even succeeded to derive
Einstein's classical equations from the loop quantum gravity. Therefore,
the author of this note is highly skeptical of the loop quantization
approach. However, it is the opinion of the author that the certain
techniques of LQG to describe a spacetime as a discrete graph may
turn out to be of some use, when one is dealing with nonperturbative
phenomena in quantum gravity.

\section{The path integral of gravitation as solution of the Wheeler-deWitt
equation}

When the articles of deWitt appeared, the relation between the perturbative
and the canonical versions of quantum gravity were unknown. This changed
with the work of Hartle and Hawking \cite{Hartle}. For a quantum
theory with a scalar field, the path integral 
\begin{equation}
Z=\psi(x,t)=N\int\mathcal{D}\phi(t)e^{iS(\phi(t))}
\end{equation}
where N is some normalization factor, fulfills a Schroedinger equation
\begin{equation}
i\frac{\partial}{\partial t}\psi(x,t)=H\psi
\end{equation}
see, e.g \cite{Peskin}. In analogy to this, the path integral for
gravity should satisfy the Wheeler deWitt equation. Hartle and Hawking
\cite{Hartle} showed that this is the case for 
\begin{equation}
Z=\int\mathcal{D}g_{\mu\nu}e^{iS}\label{eq:pathintegral}
\end{equation}
if one does not include ghost and gauge fixing terms. Later, Barvinsky
showed in \cite{Barvinsky1} that the inclusion of ghosts and gauge
fixing leads to factor ordering ambiguities as they were present in
section 4.2.

The path integral of eq. (\ref{eq:pathintegral}) can be used to define
expectation values of the form 
\begin{eqnarray}
\langle F(g)\rangle & = & \int\mathcal{D}g_{\mu\nu}F(g)e^{iS}\nonumber \\
 & = & Z^{-1}FZ
\end{eqnarray}
The wave function of the Wheeler deWitt equation depends only on the
three geometry $\gamma_{ij}$.The metric in eq. (\ref{eq:ADMMetric})
has a line element \ref{eq:ADMLINE}: 
\begin{equation}
ds^{2}=-N^{2}dt^{2}+\gamma_{ij}(dx^{i}+\beta^{i}dt)(dx^{j}+\beta^{j}dt)
\end{equation}
where $N$ is the lapse function that gives the proper time lapse
between the upper and lower hypersurfaces. When we vary $Z$ with
respect to $N$ we therefore put it forward or backward in time, but
as the wave function of general relativity should be independent on
time, we have 
\begin{eqnarray}
0 & = & \langle\frac{\delta}{\delta N}\rangle\nonumber \\
 & = & Z^{-1}\frac{\delta}{\delta N}Z\nonumber \\
 & = & \int\mathcal{D}g_{\mu\nu}\frac{\delta}{\delta N}e^{iS}\nonumber \\
 & = & Z^{-1}\left(\frac{i\delta S}{\delta N}\right)Z\label{eq:Variation1}
\end{eqnarray}
From eq.(\ref{eq:Hamiltonianconstraint}), we have 
\begin{equation}
H=\int d^{3}x(\pi\partial_{t}N+\pi^{i}\partial_{t}\beta_{i}+\pi^{ij}\partial_{t}\gamma_{ij})-L)
\end{equation}
 or 
\begin{equation}
S=\int dtL=\int d^{4}x(\pi\partial_{t}N+\pi^{i}\partial_{t}\beta_{i}+\pi^{ij}\partial_{t}\gamma_{ij})-\int dtH\label{eq:Actionadm2}
\end{equation}
 with $\pi=\pi^{i}=0$, and $\gamma_{ij}$ being independent of $N$,
and 
\begin{equation}
H=\int d^{3}x\left(N\mathcal{H_{G}}+\beta_{i}\chi^{i}\right)\label{eq:Hamilton2}
\end{equation}
As a result, we are led to the Hamiltonian constraint 
\begin{equation}
\frac{i\delta S}{\delta N}Z=-i\mathcal{H_{G}}Z=0
\end{equation}
where 
\begin{equation}
\mathcal{H_{G}}=\mathcal{G}_{ijkl}\pi^{ij}\pi^{kl}-\sqrt{\gamma}^{\:(3)}R\label{eq:Hamilton3}
\end{equation}
Similarly, computation of $\langle\frac{\delta}{\delta\beta^{i}}\rangle=0$
immediately leads to the diffeomorphism constraints of eq. (\ref{Diffeomorphismconstraint}). 

If we compute 
\begin{equation}
Z^{-1}\left(-i\frac{\delta}{\delta\gamma_{ij}}\right)Z=Z^{-1}\int\mathcal{D}g_{\mu\nu}\frac{\delta S}{\delta\gamma_{ij}}e{}^{iS}
\end{equation}
 we get, using 
\begin{equation}
\frac{\delta S}{\delta\gamma_{ij}}=\sqrt{\gamma}\left(\gamma^{ij}K-K^{ji}\right)=\pi^{ij}
\end{equation}
the expression 
\begin{equation}
Z^{-1}\left(-i\frac{\delta}{\delta\gamma_{ij}}\right)Z=Z^{-1}\int\mathcal{D}g_{\mu\nu}\pi^{ij}e{}^{iS}
\end{equation}
which yields 
\begin{equation}
\hat{\pi}^{ij}=-i\frac{\delta}{\delta\gamma_{ij}}\label{eq:fdkdfsdf}
\end{equation}
and is equal to the momentum operator from eq. (\ref{eq:momentumoperator}).
Putting eq. (\ref{eq:fdkdfsdf}) into eq. (\ref{eq:Hamilton3}), we
arrive at the Wheeler deWitt equation. Unfortunately, this is true
only approximately. Barvinsky considered the path integral with gauge-fixing
and ghost contributions \cite{Barvinsky1}. He found that these terms
lead to additional delta functions $\delta(0)$ in the constraint
equations. Similar terms arose in our discussion of the factor ordering
problem above. Apparently, the Wheeler deWitt equation can only derived
from the path integral if we set $\delta(0)=0$ and thereby ignore
the factor ordering problem.

\section{Two ways of the description of black holes in quantum gravity}

\subsection{Calculation of the black hole entropy from the path integral}

This section reviews the computation of the black hole entropy from
Euclidean quantum gravity by Gibbons and Hawking \cite{GibbonsHawking}.
One may express the metric with a classical background as $\overline{g}_{\mu\nu}=g_{\mu\nu}+h_{\mu\nu}$
and then expand the Euclidean action perturbatively as in section
2.1:
\begin{equation}
I(\overline{g}_{\mu\nu})=I(g_{\mu\nu})+\underline{I}(h_{\mu\nu})+\underline{\underline{I}}(h_{\mu\nu})+\text{higher order terms}
\end{equation}
where $\underline{I}(h_{\mu\nu})$ is linear and $\underline{\underline{I}}(h_{\mu\nu})$
is quadratic in the quantum field. As in Section 2.2, we have $\underline{I}(h_{\mu\nu})=0$
by the equations of motion, and so, omitting ghost and gauge fixing
terms, the Euclidean path integral in the background field method
is given by 
\begin{equation}
Z_{eu}=e^{-I(g_{\mu\nu})}\int\mathcal{D}h_{\mu\nu}e^{-\underline{\underline{I}}(h_{\mu\nu})}
\end{equation}
The path integral up to the action quadratic in the metric perturbation
can be evaluated with the same techniques of dimensional regularization
from section 2.2. In their book on Euclidean quantum gravity \cite{Hawkingbook},
Gibbons and Hawking use a slightly different technique, namely zeta
function regularization, see section 3.

The background field method can handle the possibility that the background
$g_{\mu\nu}$ is a classical black hole. We now assume that $g_{\mu\nu}$
is given by the Schwarzschild metric of eq. (\ref{eq:schwarzschild}).
Going over to Kruskal coordinates, the metric becomes 
\begin{equation}
ds^{2}=\frac{32M^{3}}{r}e^{\frac{-r}{2M}}(-dz+dy)+r^{2}d\Omega^{2}
\end{equation}
with 
\begin{equation}
-z^{2}+y^{2}=\left(\frac{r}{2M}-1\right)e^{\frac{r}{2M}}\label{eq:koordinates}
\end{equation}
and 
\begin{equation}
\frac{(y+z)}{(y-z)}=e^{\frac{t}{2M}}\label{eq:gglokjhpgt}
\end{equation}
The singularity lies at $-z^{2}+y^{2}=-1$. Setting $\zeta=iz$, the
metric becomes positive definite: 
\begin{equation}
ds^{2}=\frac{32M^{3}}{r}e^{\frac{-r}{2M}}(d\zeta+dy)+r^{2}d\Omega^{2}
\end{equation}
with 
\begin{equation}
\zeta^{2}+y^{2}=\left(\frac{r}{2M}-1\right)e^{\frac{r}{2M}}
\end{equation}
the coordinate $r$ will be real and greater than or equal to $2M$
as long as $y$ and $\zeta$ are real. eq. (\ref{eq:gglokjhpgt})
shows that setting $t=-i\tau$ implies $\tau$ has a period of $8\pi M$. 

The Euclidean path integral has a connection to the canonical partition
function. We have for field configurations of a scalar field $\varphi_{1}$at
$t_{1}$and $\varphi_{2}$ at $t_{2}$with a Hamiltonian $H$
\begin{eqnarray}
\langle\varphi_{2},t_{2}|\varphi_{1},t_{1}\rangle & = & \int\mathcal{D}\varphi e^{iS}=\langle\varphi_{2},t_{2}|e^{-iH(t_{2}-t_{1})}|\varphi_{1},t_{1}\rangle
\end{eqnarray}
where the path integral is over all field configurations that take
the value $\varphi_{1}$at $t_{1}$ and $\varphi_{2}$ at $t_{2}$,
and on the right hand side, the Schroedinger picture for the amplitude
was invoked. Setting $t_{2}-t_{1}=-i\beta$ and $\varphi_{1}=\varphi_{2}$,
a summation over all $\varphi_{1}$ yields the canonical partition
sum 
\begin{equation}
tr(e^{-\beta H})=\int\mathcal{D}\varphi e^{-I}=Z_{eu}
\end{equation}
with the path integral now taken over all fields with period in $\beta$
in imaginary time. Since the Euclidean section of the Schwarzschildmetric
is periodic in 
\begin{equation}
\beta=8\pi M
\end{equation}
one should be able to compute the canonical partition sum from the
Euclidean path integral. The Euclidean section of the Schwarzschildmetric
has $R=0$. Therefore, the non-zero part comes from the boundary part
of the action in eq. (\ref{eq:fulleuclideanaction}): 
\begin{equation}
I=-\int d^{4}x\sqrt{g}R-2\int_{\partial M}d^{3}x\sqrt{\gamma}(K-K^{0})\label{eq:fulleuclideanaction-1}
\end{equation}
which is an integral over the intrinsic curvature. Evaluation of this
integral yields, see \cite{GibbonsHawking}: 
\begin{equation}
I=4\pi M^{2}=\frac{\beta^{2}}{16\pi}
\end{equation}
(Note that to obtain the correct result, one has to use the action
with the correct factor $\tilde{I}=\frac{1}{16\pi}I$,  which we discarded
for simplicity in section 2).

From statistical mechanics, we have 
\begin{equation}
\langle E\rangle=-\frac{\partial}{\partial\beta}ln(Z_{eu})
\end{equation}
 inserting the contribution of the background 
\begin{equation}
Z_{eu}\approx e^{-\frac{\beta^{2}}{16\pi}}
\end{equation}
we get 
\[
\langle E\rangle=M=\frac{\beta}{8\pi}
\]
 The entropy of the canonical ensemble is defined by 
\begin{equation}
S=\beta\langle E\rangle+ln(Z_{eu})=\frac{\beta^{2}}{8\pi}-\frac{\beta^{2}}{16\pi}=\frac{\beta^{2}}{16\pi}
\end{equation}
which yields 
\begin{equation}
S=4\pi M^{2}=\frac{1}{4}A
\end{equation}
 where $A$ is the area of the event horizon.

Gibbons and Hawking mention in their article \cite{GibbonsHawking}
that one can use this technique also for different spacetimes and
write: ``Because $R$ and $K$ are holomorphic functions on the complexified
spacetime except at singularities, the action integral is really a
contour integral and will have the same value on any section of the
complexified spacetime which is homologous to the Euclidean section
even though the metric on this section may be complex. This allows
to extend the procedure to other spacetimes which do not necessarily
have a real Euclidean section.'' Gibbons and Hawking then discuss
the Kerr solution and mention that one can use this technique also
in the presence of matter fields, where a black hole is surrounded
by a perfect fluid rotating at some angular velocity. Finally, they
also consider a star of rotating matter without an event horizon and
find that without an event horizon, the gravitational field apparently
does not contribute to entropy.

\subsection{Calculation of the black hole entropy from the Wheeler deWitt equation}

One can derive the entropy of a black hole not only from the gravitational
path integral, but also from the canonical formalism involving the
Hamiltonian and the Wheeler deWitt equation. The following section
is merely a summary of the excellent explanations \cite{Kiefer-1,Louko}.
The line element of a spacetime with spherical symmetry is, see the
calculation in \cite{Louko}: 
\begin{equation}
ds^{2}=-N^{2}(r,t)dt^{2}+\Lambda^{2}(r,t)(dr+\beta^{r}dt)^{2}+R^{2}(r,t)d\Omega^{2}
\end{equation}
where $(r,t)$ is a parametrization of the spacetime, $N$ is the
lapse and $\beta^{r}$ the shift function. The Hamiltonian constraint
of such a spacetime can be derived as, 

\begin{equation}
\mathcal{H}_{G}=\frac{\Lambda P_{\Lambda}}{2R^{2}}-\frac{P_{\Lambda}P_{R}}{R}+V_{G}=0\label{eq:Hamiltonianblackhole}
\end{equation}
see Kuchar:\cite{Kuchar}, and Louko and Whiting \cite{Louko}, where
\begin{equation}
V_{G}=\frac{RR''}{\Lambda}-\frac{RR'\Lambda'}{\Lambda^{2}}+\frac{R'^{2}}{2\Lambda}-\frac{\Lambda}{2}
\end{equation}
The diffeomorphism constraint becomes
\begin{equation}
\mathcal{H}_{r}=P_{R}R'-\Lambda P'_{\Lambda}=0
\end{equation}
 with 
\begin{equation}
P_{\Lambda}=-N^{-1}R(\dot{R}-R'N^{r})\label{eq:canmomandesfa}
\end{equation}
and 
\begin{equation}
P_{R}=-N^{-1}(\Lambda(\dot{R}-R'N^{r})+R(\dot{\Lambda}-(\Lambda N^{r})')\label{eq:canmomannewfswf}
\end{equation}
as canonical momenta that one obtains from varying the action with
respect to $\dot{R}$ and $\dot{\Lambda}$. Using the canonical momenta
from eqs. (\ref{eq:canmomandesfa}-\ref{eq:canmomannewfswf}) and
the Hamiltonian constraint of eq. (\ref{eq:Hamiltonianblackhole}),
the Wheeler deWitt equation is 
\begin{equation}
\left(\frac{-\Lambda\delta^{2}}{2R^{2}\delta\Lambda^{2}}+\frac{1}{R}\frac{\delta^{2}}{\delta\Lambda\delta R}+V_{G}\right)\psi(\Lambda,R)=0
\end{equation}
 A semi-classical solution of this equation is given by the WKB ansatz
\begin{equation}
\psi(\Lambda,R)=C(\Lambda,R)e^{iS_{0}(\Lambda,R)}
\end{equation}
 with 
\begin{equation}
\left|\frac{\delta C(\Lambda,R)}{\delta\Lambda}\right|<<\left|C(\Lambda,R)\frac{\delta S_{0}(\Lambda,R)}{\delta\Lambda}\right|
\end{equation}
and 
\begin{equation}
\left|\frac{\delta C(\Lambda,R)}{\delta R}\right|<<\left|C(\Lambda,R)\frac{\delta S_{0}(\Lambda,R)}{\delta R}\right|
\end{equation}
which implies for the Hamiltonian constraint 
\begin{equation}
\frac{-\Lambda}{2R^{2}}\left(\frac{\delta S_{0}}{\delta\Lambda}\right)^{2}+\frac{1}{R}\frac{\delta^{2}S_{0}}{\delta\Lambda\delta R}+V_{G}=0.
\end{equation}
We have for the exterior of the Schwarzschildmetric an action of the
form 
\begin{equation}
S=\int dtL=\int dt\int_{0}^{\infty}dr\left(P_{\Lambda}\dot{\Lambda}+P_{R}\dot{R}-N\mathcal{H}_{G}-\beta^{r}\mathcal{H}_{r}\right)\label{eq:admlagerangian-1}
\end{equation}
See \cite{Kiefer-1,Louko}. The action in eq. (\ref{eq:admlagerangian-1})
does not take boundary terms into account. In order to describe a
spacetime with an event horizon and a flat curvature at infinity,
the action must therefore be supplied with appropriate boundary terms.
These then lead to new degrees of freedom that give rise to additional
constraints. The additional constraints then lead to modifications
in the Hamiltonian and the functional $S_{0}$. 

According to Louko and Whiting \cite{Louko}, the following boundary
conditions are appropriate for an non-degenerate event horiuon, which
we locate at the parameter value $r=0$: 
\begin{equation}
R(t,r)=R_{0}(t)+R_{2}(t)r^{2}+\mathcal{O}(r^{4})\text{, where }R_{0}\text{ is defined by }\left(1-\frac{2M}{R}\right)|_{r=0},
\end{equation}
\begin{equation}
N(t,r)=N_{1}(t)r+\mathcal{O}(r^{3}),
\end{equation}
\begin{equation}
\beta^{r}(t,r)=\beta_{1}^{r}(t)r+\mathcal{O}(r^{3}),
\end{equation}
\begin{equation}
\Lambda(t,r)=\Lambda_{0}(t)+\mathcal{O}(r^{2})
\end{equation}
Kiefer and Brotz write in \cite{Kiefer-1} that the variation of the
action $S$ from eq. (\ref{eq:admlagerangian-1}) leads with the above
boundary conditions to the following boundary term at $r=0$: 
\begin{equation}
\delta S|_{r=0}=\left.-\frac{\partial}{\partial r}\left(N\frac{\partial\mathcal{H}_{G}}{\partial R''}\right)\delta R\right|_{r=0}=-\frac{N_{1}R_{0}}{\Lambda_{0}}\delta R_{0}
\end{equation}
This term must be subtracted from the action if $N_{1}\neq0$. Similarly,
it was found by Regge and Teitelboim in 1974 \cite{Reggeteitelboim}
that in case of asymptotically flat spacetimes, another boundary term
must be subtracted. This was also mentioned earlier by deWitt in \cite{deWitt}.
With all these boundary terms included, the action now becomes::
\begin{eqnarray}
S_{total}=\int dtL & = & \int dt\int_{0}^{\infty}dr\left(P_{\Lambda}\dot{\Lambda}+P_{R}\dot{R}-N\mathcal{H}_{G}-\beta^{r}\mathcal{H}_{r}\right)\nonumber \\
 &  & +\frac{1}{2}\int dt\frac{N_{1}R_{0}^{2}}{\Lambda}-\int dtN_{+}M\label{eq:changedaction}
\end{eqnarray}
with $M$ as the so-called ADM mass and $N_{+}$ as the lapse function
at infinity. A variation of this action would lead to an unwanted
term 
\begin{equation}
\int dt\frac{R_{0}^{2}}{2}\delta\left(\frac{N_{1}}{\Lambda}\right)
\end{equation}
This term is zero if we assume that 
\begin{equation}
\left(\frac{N_{1}}{\Lambda}\right)\equiv N_{0}(t)
\end{equation}
is fixed at $r=0$, i.e.it can not be varied. The need of fixing $N_{1}$
and $N_{+}$can be removed if we define suitable parametrizations
$N_{0}(t)\equiv\dot{\tau}(t)$ and $N_{+}\equiv\dot{\tau}_{+}(t)$
and consider $\tau$ and $\tau_{+}$ as additional variables. We then
get an action:
\begin{eqnarray}
S_{total} & = & \int dt\int_{0}^{\infty}dr\left(P_{\Lambda}\dot{\Lambda}+P_{R}\dot{R}-N\mathcal{H}_{G}-\beta^{r}\mathcal{H}_{r}\right)\nonumber \\
 &  & +\int dt\frac{R_{0}^{2}}{2}\dot{\tau}-\int dtM\dot{\tau}_{+}\label{eq:changedaction2}
\end{eqnarray}
In the canonical framework, the functions $\tau$ and $\tau_{+}$
represent additional degrees of freedom that must be supplied with
corresponding canonical momenta $\pi_{0}$ and $\pi_{+}$. These momenta
can only brought consistently into the action as additional variables,
if we impose additional constraints 
\begin{equation}
C_{0}=\pi_{0}-\frac{R_{0}^{2}}{2}=0\label{eq:constrblackhole1}
\end{equation}
 and 
\begin{equation}
C_{+}=\pi_{+}+M=0\label{eq:constrblackhole2}
\end{equation}
With $N_{0}$ and $N_{+}$ now acting as Lagrange multipliers, the
action becomes 
\begin{eqnarray}
S_{total} & = & \int dtL=\int dt\int_{0}^{\infty}dr\left(P_{\Lambda}\dot{\Lambda}+P_{R}\dot{R}-N\mathcal{H}_{G}-\beta^{r}\mathcal{H}_{r}\right)\nonumber \\
 &  & +\int dt\pi_{0}\dot{\tau}_{0}+\pi_{+}\dot{\tau}_{+}-N_{0}C_{0}-N_{+}C_{+}
\end{eqnarray}
In the quantum theory, the additional momenta become $\pi=-i\frac{\delta}{\delta\tau_{0}}$
and $\pi_{+}=-i\frac{\delta}{\delta\tau_{+}}$ and the WKB ansatz
\[
\psi=C(\Lambda,R)e^{iS_{0}(\Lambda,R.\tau_{0},\tau_{+})}
\]
with the additional degrees of freedom $\tau_{+}$ and $\tau_{0}$
then leads to the quantum mechanical constraints 
\begin{equation}
\frac{\partial_{0}S_{0}}{\partial\tau_{0}}-\frac{R_{0}^{2}}{2}=0\label{eq:quantumconstr1}
\end{equation}
\begin{equation}
\frac{\partial_{0}S_{0}}{\partial\tau_{+}}+M=0\label{eq:quantumconstr2}
\end{equation}
This changes the solution $S_{0}$ into 
\begin{equation}
S_{0}(\Lambda,R,\tau_{+},\tau_{0})\rightarrow S_{0}+\frac{R_{0}^{2}}{2}\tau_{0}-M\tau_{+}
\end{equation}

Setting $r-2M=\zeta(r)$ in the Schwarzschildmetric, we get, up to
order $\frac{1}{2M}$, see \cite{Dhabolkare} on p. 4: 
\begin{equation}
ds^{2}=-\frac{\zeta}{2M}dt^{2}+\frac{2M}{\zeta}(d\zeta)^{2}+(2M)^{2}d\Omega^{2}
\end{equation}
Defining 
\begin{equation}
\rho^{2}=8M\zeta
\end{equation}
yields 
\begin{equation}
d\zeta^{2}\frac{2M}{\zeta}=d\rho^{2}
\end{equation}
and the line element becomes 
\begin{equation}
ds^{2}=-\frac{\rho^{2}}{16M^{2}}dt^{2}+d\rho^{2}+(2M)^{2}d\Omega^{2}
\end{equation}

This corresponds to our previous line element with $\beta^{r}=0$,
$\Lambda=1$, $N(t,\rho)=N_{1}(t)\rho$, and $N_{1}=\frac{1}{4M}$.
Setting $t=-it_{e}$, the euclideanized version of this line element
is 

\begin{equation}
ds^{2}=N_{1}^{2}\rho^{2}dt_{e}^{2}+d\rho^{2}+(2M)^{2}d\Omega^{2}
\end{equation}
The Euclideanized Schwarzschildmetric has a time coordinate that is
periodic with $8\pi M$ . Thereby, using $N_{1}=\frac{1}{4M}$, the
variable $\tau_{0}$ becomes 
\begin{equation}
\tau_{0}=\int_{0}^{8\pi M}dtN_{1}=8\pi MN_{1}=2\pi
\end{equation}
and we get 
\begin{equation}
S_{0}+R_{0}^{2}\pi-M\tau_{+}=S_{0}+\frac{1}{4}A-M\tau_{+}
\end{equation}
Similarly, the parameter$\beta=\tau_{+}=\int dtN_{+}$ is interpreted
by Louko and Whiting in \cite{Louko} as the inverse of the normalized
temperature at infinity. The euclideanized WKB solution of the Wheeler
deWitt equation is then 
\begin{equation}
\psi(\Lambda,R,)=\psi_{0}(\Lambda,R)e^{-\beta M+\frac{A}{4}}
\end{equation}
We have seen in section 5 that the gravitational path integral is
a solution of the Wheeler deWitt equation. Similarly, $\psi_{0}(\Lambda,R)=e^{-S_{0}(\Lambda,R)}$
can now be considered as the quantum contribution to the path integral
and $e^{-\beta M+\frac{A}{4}}$ is considered as the background contribution
to the amplitude. By comparison with . 
\begin{equation}
\tilde{S}-\beta\langle E\rangle=ln(Z_{eu})
\end{equation}
where $\tilde{S}$ is the entropy and in the Euclidean path integral
$Z_{eu}$only background contributions are considered, we can indentify
$\tilde{S}=\frac{1}{4}A$ as black hole entropy.

\section{A comment on recent papers by Dvali and Gomez.}

\subsection{A comment on articles by Dvali and Gomez regarding a proposed ``self-completeness''
of quantum gravity}

In the following, we will critically assess various statements and
claims that Dvali and Gomez make in \cite{Dvali2}. The main subject
of this article seems to circle around the question whether one could
design an experiment to resolve so called trans-Planckian states.
Dvali and Gomez first consider a gravitational amplitude on p. 5 
\begin{equation}
T^{\mu\nu}\Delta_{\mu\nu\alpha\beta}\tau^{\alpha\beta}=\frac{T_{\mu\nu}\tau^{\mu\nu}-\frac{1}{2}T_{\mu}^{\mu}\tau_{\nu}^{\nu}}{p^{2}}
\end{equation}
with two energy momentum sources $T_{\mu\nu}$ and $\tau_{\mu\nu}$.
They then consider a modification of this amplitude of the form 
\begin{equation}
T^{\mu\nu}\Delta_{\mu\nu\alpha\beta}\tau^{\alpha\beta}=\frac{T_{\mu\nu}\tau^{\mu\nu}-\frac{1}{2}T_{\mu}^{\mu}\tau_{\nu}^{\nu}}{p^{2}}+\frac{aT_{\mu\nu}\tau^{\mu\nu}-b\frac{1}{3}T_{\mu}^{\mu}\tau_{\nu}^{\nu}}{p^{2}+(1/L)^{2}}\label{eq:gravipropdvali}
\end{equation}
where the parameters a and b are fixed according to the spin of the
new particle and $(1/L)$ symbolizes the energy of the new particle.
The latter is assumed to have a mass of $1/L$ which is higher than
the Planck mass. Dvali and Gomez then write that such a modification
of a propagator would be impossible to detect since a scattering experiment
at that energies would create a black hole of a Schwarzschild radius
larger than the impact parameter of the scattering experiment. For
this, Dvali and Gomez use a version of the so called generalized uncertainty
principle in quantum gravity. In the following paragraph, we will
mention constants like $G,c$ and the Planck's constant $h$ explicitly.
Multiplying the equations for the Schwarzschild radius 
\begin{equation}
r_{s}=\frac{2Gm}{c^{2}}
\end{equation}
and the reduced Compton wavelength 
\begin{equation}
\overline{\lambda}_{c}=\frac{\lambda_{c}}{2\pi}=\frac{\hbar}{mc}
\end{equation}
gives the so called Planck length $l_{p}$ 
\begin{equation}
r_{s}\overline{\lambda}_{c}=\frac{2G\hbar}{c^{3}}=2l_{p}^{2}\geq l_{p}^{2}
\end{equation}
Since the Planck mass $m_{p}$ is defined by $r_{s}=\lambda_{c}$
any attempt to generate a black hole of mass greater than $m_{p}=\sqrt{\frac{hc}{2G}}$
will give a black hole larger than its Compton wavelength. This black
hole should be regarded as a classical black hole for large $r$ or
$m>m_{p}$ . Then, Dvali and Gomez give an argument that should approximately
be something like: ''putting a quantum state within a box of width
$L$ with infinitely high walls would require an energy of $E\propto\frac{1}{L^{2}}$''
Actually, they write that $E=\frac{1}{L}$ on p. 4, but this seems
to be a typo. Converting this kinetic energy into a mass with $E=mc^{2}$
yields $m=E/c^{2}\propto\frac{1}{L^{2}c^{2}}$ or $\overline{\lambda}_{c}\propto L^{2}\hbar c$
and setting this into the formula for the definition of the Planck
length yields 
\begin{equation}
r_{s}\propto\frac{2l_{p}^{2}}{L^{2}\hbar c}
\end{equation}
Based on these approximations, whose validity is by no means guaranteed
even approximately at ``trans-Planckian energies'', Dvali and Gomez
conclude that any attempt to resolve distances smaller than the Planck
length $l_{p}$ should create an even larger black hole, which could
then be regarded as classical. Thereby, the modifications in the graviton
propagator in eq. (\ref{eq:gravipropdvali}) would not show up in
an experiment, since any attempt to resolve the added trans-Planckian
degrees of freedom would lead to the creation of a black hole with
a Schwarzschild radius larger than the Planck length. This very same
argument is then repeated in several different physical situations.
Finally, Dvali and Gomez note on p. 5 in eq. (8) that the path integral
of classical black holes is exponentially suppressed, as we saw in
section 6. From all this, Dvali and Gomez concludes on p. 3 that $"Deep-UV-gravity=Deep-IR-gravity"$
and they state on p. 2
\begin{quotation}
``The existing common knowledge about Einstein gravity is that it
becomes inapplicable at deep UV and that it must be completed by a
more powerful theory that will restore consistency at sub-Planckian
distances.We wish to question the above statement and suggest that
pure Einstein gravity is self complete in deep-UV. In other words,
 we argue that for restoring consistency no new propagating degrees
of freedom are necessary at energies $>>M_{p}$''
\end{quotation}
In the article of Dvali and Gomez, there are several aspects which
in the view of the author of this text are somewhat problematic. For
example. on p. 6, Dvali and Gomez write
\begin{quotation}
``However, for $L^{-1}>>M_{planck}$, the same vertex describes evaporation
of a classical BH of mass $1/L$ into a single particle-anti-particle
pair''. 
\end{quotation}
But in their entire article, they do not seem to specify the mathematical
details of a quantum field theory that would describe an ``evaporation
of a black hole'' as a ``vertex''. On p. 32, they write:
\begin{quotation}
In other words, a String theory with order one String coupling is
built in in Einstein gravity. To put it differently, by writing down
Einstein's action, we are committing ourselves to a String theory
\end{quotation}
As we saw in section 2, the amplitude of Einstein gravity is divergent
at two loops. This is different in String theory, whose two loop finiteness
has been shown by D'Hoker and Phong \cite{DHooker}. Furthermore,
String theory reduces in the classical limit to a version of Supergravity
and not Einstein gravity \cite{gsw}. In the opinion of the author
of this text, this renders the claims of Dvali and Gomez regarding
String theory highly problematic. 

However, most obvious problem inherent in the assumptions of Dvali
and Gomez is that the formula of the Schwarzschild radius, their eqs.
(5) and (7), which they invoke on p. 4, 5, 12, 15, 16, 23, 24, 25,
and 29 of their article in various different physical situations,
are something that comes entirely from purely classical physics. Without
a proper description of a quantum black hole, it is not clear what
the analogue of the Schwarzschild radius will look like, when a quantum
state of high energy is put into a box of ``trans-Planckian width''
as Dvali and Gomez call it. 

In fact, the equations for the black hole argument of Dvali and Gomez
do not seem to follow at all from the quantum field theoretical amplitudes
themselves. As we saw in section 3, up to one loop order of the gravitational
amplitude, only some estimates can be made about the topologies that
give the dominant contributions to the path integral at Planck scale.
The analysis of section 3 implies that the euler characteristic $\chi$
of the metrics which are dominant contributions fulfills 
\[
\chi\propto hV
\]
where $V$ is the volume and $h$ is some constant. According to Hawking,
this result means that we will find one gravitational instanton per
unit volume $h^{-1}$ at Planck scale. Thereby, the obvious disconnectedness
of the black hole argument from Dvali and Gomez with the quantum field
theoretical amplitude of quantum gravity can, at least for amplitudes
up to one loop order, be somewhat removed with Hawking's work on spacetime
foam from 1978. However, the amplitude that follows from quantum gravity
seems to always imply a summation over several metrics, and not only
over ones that describe classical black hole metrics. Hence there
remain important differences between the proposal from Dvali and Gomez
and the gravitational path integral even at one loop order. 

As we saw in section 2.2, an evaluation of the gravitational path
integral in two loop order leads to divergencies. One could suspect
that these divergences arise because of non-local topological features
of the spacetime that we have to expect at planck scale. The one loop
analysis in section 3 reveals a spacetime filled with some sort of
gravitational instantons. Since a gravitational instanton is a non-local
entity, further analysis of the path integral would likely need non-perturbative
methods. The traditional way of computing a perturbation series with
the gravitational path integral in terms of Feynman diagrams may then
no longer be a valid procedure at all. To quote Hawking\cite{Hawkingbook}: 
\begin{quotation}
``Attempts to quantize gravity ignoring the topological possibilities
and simply drawing Feynman diagrams around flat space have not been
very successful. It seems to me that the fault lies not with the pure
gravity or supergravity theories themselves but with the uncritical
application of perturbation theory to them. In classical relativity
we have found that perturbation theory has only limited range of validity.
One can not describe a black hole as a perturbation around flat space.
Yet this is what writing down a string of Feynman diagrams amounts
to.''
\end{quotation}
That one likely has to do with non-perturbative physics at the Planck
scale and below is acknowledged by Dvali and Gomez when they write
on p. 16:
\begin{quotation}
However, for $m>>m_{p}$ the new degree of freedom is no longer a
perturbative state.
\end{quotation}
Unfortunately, a non perturbative evaluation of the path integral
turned out to yield a divergent amplitude again. In section 2.3, we
arrived at the conclusion that the path integral of Euclidean quantum
gravity can be described by decomposing the metric into $\tilde{g}_{\mu\nu}=\Omega^{2}g_{\mu\nu}$
, with $\Omega$ as a conformal factor. We can now write the path
integral as 
\begin{equation}
Z_{eu}=\int\mathcal{D}gY(g)\label{eq:wffswfs}
\end{equation}
where 
\begin{equation}
Y(g)=\int\mathcal{D}\Omega e^{-I(\Omega^{2},g)}\label{eq:swfsdfs}
\end{equation}
and with an action 
\begin{equation}
I(\Omega^{2},g)=-\int d^{4}x\sqrt{g}\left(\Omega^{2}R+6g^{\mu\nu}\partial_{\mu}\Omega\partial_{\nu}\Omega\right)-2\int_{\partial M}d^{3}x\sqrt{\gamma}\Omega^{2}(K-K^{0})
\end{equation}
Because of the derivatives of $\Omega$, the action $I(\Omega^{2},g)$
can be arbitrarily negative if a rapidly varying conformal factor
is chosen. However, eq. (\ref{eq:swfsdfs}) says that one has to do
an integration over all possible conformal factors, including ones
that lead to a divergent path integral. 

The non-perturbative evaluation of the gravitational path integral
does not contain any mechanism that would exclude the summation over
metrics $\tilde{g}_{\mu\nu}=\Omega^{2}g_{\mu\nu}$ with a rapidly
variing conformal factor. For the tought experiment that is envisaged
by Dvali and Gomez where a collider probes distances at the Planck
scale, the quantum mechanical formulas above would therefore imply
a divergent amplitude. 

Another way of doing non-perturbative quantum gravity would be the
Wheeler deWitt equation. Unfortunately, this theory is only consistent
for low energy states within a WKB approximation. As we observed in
section 4.2, at high energies, severe factor ordering problems appear
and the theory becomes inconsistent, yielding equations like 
\[
0=-6i\left(\delta(x,x)\delta\zeta^{k}\right)_{,k}
\]
where $\zeta$ is an infinitesimal displacement and $,k$ denotes
a partial derivative with respect to $x^{k}$.

The author of this text has the opinion that one could perhaps speak
of some kind of ``self completeness'', if there were, for example,
a mechanism in the Wheeler deWitt equation, which would show that
high energy states would automatically lead to the creation of classical
solutions for black holes. Within the scope of a semi classical WKB
approximation, the Wheeler deWitt theory is able to handle large non-local
quantized metrics, like the quantized Friedmann Robertson Walker universe,
see \cite{deWitt}, or a quantized Schwarzschild metric, see section
6.2. Furthermore, from the WKB ansatz, deWitt was able to show a connection
between the classical Einstein equations and the Wheeler deWitt theory,
see section 4.1 and \cite{deWitt}. Therefore, it would be conceivable
if this theory would predict the emergence of classical black holes
at high energies. Unfortunately, such a mechanism seems not to be
there in the Wheeler deWitt theory, but one gets severe factor ordering
inconsistencies at high energies instead, as we have demonstrated
in section 4.2. 

The artificial invocation of the classical Schwarzschild radius at
trans-Planckian energies in the article \cite{Dvali2} by Dvali and
Gomez does not follow from the equations of non-perturbative quantum
gravity. The author of this text therefore has the opinion that the
proposals of Dvali and Gomez on the alleged ``self-completeness''
of gravity are not justified. Instead the author of this note advocates
the view that quantum general relativity must be replaced by some
other theory, like loop quantum gravity, or String theory that both
support the definition of a complete Hilbert space of states without
giving rise to severe inconsistencies at high energies.

But there is an additional problem in the proposal of Dvali and Gomez.
Any theory that proposes an amplitude which is dominated by virtual
gravitational instantons must take into account that the trajectories
of ordinary particles flying in this spacetime might be altered by
these instanton metrics. Hawking, Page and Pope figured that this
might even lead to predictions that can be experimentally tested \cite{Bubbles1,Bubbles2}.
They devised an approximation scheme to evaluate the path integral
of gravity and matter fields non-perturbatively.

A typical amplitude for a scalar particle propagating from a initial
field mode $u(x')$ to a final mode $v(x')$ would be 
\begin{equation}
-\int\overline{u}(x')\overleftrightarrow{\nabla}_{\mu}G(x',y')\overleftrightarrow{\nabla}_{\nu}v(x')d\Sigma^{\mu}(x')d\Sigma^{\nu}(y')
\end{equation}
where $\Sigma^{\mu}(x')$ and $\Sigma^{\nu}(y')$ are the Cauchy data
for the initial and final states and $G(x',y')$ is the Green's function
of the metric. The s-matrix is computed from initial and final states
emerging from infinity, where they are supposed to be non-interacting,
and must be governed by flat space equations of motion. Therefore,
such amplitudes make only sense in asymptotically Euclidean metrics
and we must try to convert metrics that are not asymptotically Euclidean
into an asymptotically Euclidean form with an appropriate conformal
factor. 

Taking the idea of the non-perturbative evaluation of the path integral
seriously, one can not confine oneself to simply consider an amplitude
composed of black hole metrics, but one has to sum the path integral
over all possible metrics. Unfortunately, the Green's functions for
most of these metrics are not known and we can not compute the expression
of a scattering amplitude of a particle within such an arbitrary metric. 

Hawking, Page and Pope note that by using topological sums of a certain
number of copies of $CP^{2}$ and $\overline{CP^{2}}$ (the bar means
opposite orientation), one can construct a simply connected closed
manifold of arbitrary signature $\tau$ and Euler characteristic $\chi$,
with an odd and definite intersection form. Similarly, by using certain
numbers of copies of $S^{2}\times S^{2}$ and $K^{3}$ if $\tau>0$
or $\overline{K^{3}}$ if $\tau<0$, one can construct a simply connected
closed manifold with even and indefinite intersection form, arbitrary
signature and Euler number, see \cite{Kirby} p. 26. By Freedman's
theorem, the topology of the simply connected spacetimes from these
construction would then, up to homeomorphy, be equivalent to an arbitrary
simply connected spacetime with the same Euler number and signature.
Note, however that this equivalence just holds for the topology, and
not the metric. 

Hawking, Page and Pope then propose the view that one should restrict
the path integral to simply connected spacetimes .With the argument
above, the topology of this simply connected spacetime can be build
out of building blocks like $CP^{2}$, $\overline{CP^{2}},$$S^{2}\times S^{2}$,
$K^{3}$,$\overline{K^{3}}$. Hawking et al then proceed to calculate
the scattering amplitudes of particles that are moving in these building
blocks. They weight these amplitudes with $e^{-I}$ where $I$ is
the Euclidean action of the metric, and include some weighting factor
that depends on the conformal transformation which was employed to
make the manifolds asymptotically Euclidean. Then, Hawking, Page and
Pope average over all parameter values (like scale parameters, or
certain orientations) that these metrics have. 

For example, we noted in section 3.2 that $CP^{2}$ has the following
metric 
\begin{equation}
ds'=\frac{\rho'^{2}}{\rho'^{2}+x'^{2}}\left(\delta_{\mu\nu}-\frac{x_{\mu}'x_{\nu}'+n{}_{\mu\sigma}^{l}n_{\nu\lambda}^{l}x'^{\sigma}x'^{\lambda}}{\rho'^{2}+x'^{2}}\right)dx'^{\mu}dx'^{\nu}
\end{equation}
with a scale parameter $\rho$ and 
\begin{equation}
\eta_{\mu\sigma}^{l}=\begin{pmatrix}0 & 1 & 0 & 0\\
-1 & 0 & 0 & 0\\
0 & 0 & 0 & 1\\
0 & 0 & -1 & 0
\end{pmatrix}
\end{equation}
 (see \cite{Bubbles1,Bubbles2}). From the metric of $CP^{2}$, an
asymptotically Euclidean metric can be obtained by an appropriate
conformal transformation $g=\Omega^{2}g'$ which sets the origin of
$CP^{2}$to infinity. Thereby we get a space of the same topology
as $CP^{2}$(up to homeomorphy), but with the appropriate infinity
structures. The green functions of $CP^{2}$ have the form, see \cite{Bubbles1,Bubbles2}:
\begin{equation}
G(x',y')=\frac{1}{4\pi^{2}\rho'(1-L)}
\end{equation}
 where 
\begin{equation}
L=\frac{(\rho'+x'y'-in_{\mu\nu}^{l}x'^{\mu}y'^{\nu})(\rho'^{2}+x'y'+in_{\mu\nu}^{l}x'^{\mu}y'^{\nu})}{(\rho'^{2}+x'^{2})(\rho'^{2}+y'^{2})}
\end{equation}
Note that this Green's function has additional singularities compared
to the Green's functions in flat space. For spaces like $S^{2}\times S^{2}$
the Green's functions are not known, so Hawking, Page and Pope made
an approximation to $S^{2}\times S^{2}$ by a metric that describes
a conformally flat manifold with conical singularities.

After weighting the individual amplitudes and integrating over the
metric parameters, Hawking et al conclude that the amplitudes are
of order 
\begin{equation}
A\propto\left(\frac{k_{1}k_{2}}{m_{p}}\right)^{s}
\end{equation}
where $k_{1}$ and $k_{2}$ are the momenta of the in and out states,
$m_{p}$ is the Planck mass and $s$ is the spin. For a scalar particle,
like the Higgs field, we have $s=0$ and therefore the amplitudes
would be of order one. 

Hawking writes that this would suggest that the Higgs particle is
of composite nature. Unfortunately, in 2012, the Higgs particle has
been found at the Large Hadron Collider in Genf, and further analysis
of the data provided evidence for the Higgs field to be indeed a scalar
particle \cite{Higgs2}. This puts the approximations of Hawking et
al severely into question. Additionally, Warner\cite{Warner} has
analyzed scattering amplitudes of Spin 1 fields with Hawking's model,
and he also found large amplitudes that are in disagreement with observation. 

From this one can conclude that either the approximation of the gravitational
path integral by restricting it to $CP^{2}$, $\overline{CP^{2}},$$S^{2}\times S^{2}$,
$K^{3}$,$\overline{K^{3}}$ metrics is wrong, or that the idea of
a gravitational scattering amplitude dominated at Planck scale by
virtual gravitational instantons is physically incorrect. 

Unfortunately, any model that proposes a scattering amplitude being
dominated by non-perturbative virtual gravitational instantons will
have to confront the problem that the Green's functions for these
instantons will, in general, look very different from flat space.
Thereby some changes in the particle trajectories should be expected
in all these models. 

Hawking et al tried to approximate the Euclidean path integral by
restricting the summation to $CP^{2},\overline{CP^{2}},S^{2}\times S^{2},K^{3},\overline{K^{3}}$
metrics because one does not know the Green's function of all the
metrics that the path integral has to be summed over. The fact that
their approximation turned out to be wrong also suggests that entirely
new methods are needed to describe the spacetime at Planck scale.

\subsection{A comment on articles by Dvali and Gomez on black holes}

In Section 3, we have shown that the gravitational one loop amplitude
should be dominated by virtual gravitational instantons at Planck
scale. So perhaps a way to get further understanding of gravity would
be to create reasonable quantum mechanical models for spacetimes with
high Euler numbers and signature. Recently, Dvali and Gomez have proposed
something like this. In his article ``black holes as critical point
of quantum phase transition.'' \cite{DvaliBlackholes} they write
on p. 2
\begin{quote}
``Black holes represent Bose-Einstein-Condensates of gravitons at
the critical point of a quantum phase transition''. 
\end{quote}
Dvali and Gomez proceed by writing on p. 10
\begin{quotation}
``We now wish to establish the connection between the black hole
quantum portrait and critical phenomena in ordinary BEC. We shall
consider a simple prototype that captures the key phenomenon. Let
$\psi(x)$be a field operator describing the order parameter of a
Bose gas. The simplest Hamiltonian that takes into account the self
interaction of the order parameter can be written in the form 
\[
H=-\hbar L_{0}\int d^{3}x\Psi(x)\nabla^{2}\Psi(x)-g\int d^{3}x\Psi(x)^{\dagger}\Psi(x)^{\dagger}\Psi(x)\Psi(x)
\]

\end{quotation}
Dvali and Gomez then write 
\begin{quotation}
``We shall put the system in a box of Size R and periodic boundary
conditions 
\[
\Psi(0)=\Psi(2\pi R)
\]
Performing a plane wave expansion 
\[
\Psi=\sum_{k}\frac{a_{k}}{\sqrt{V}}e^{i\frac{\vec{k}\vec{x}}{R}}
\]
we can rewrite the Hamiltonian as 
\[
\mathcal{H}=\sum_{k}k^{2}a_{k}^{\dagger}a_{k}-\frac{1}{4}\alpha\sum_{k}a_{k+p}^{\dagger}a_{k'-p}^{\dagger}a_{k}a_{k'}"
\]

\end{quotation}
In contrast to these statements by Dvali and Gomez, we have derived
the classical Hamiltonian of general relativity in section 4 as 
\begin{eqnarray}
H & = & \int d^{3}x\left(N\sqrt{\gamma}(K_{ij}K^{ij}-K^{2}-{}^{3}R)-\beta_{i}2D_{j}\left(\gamma^{-1/2}\pi^{ij}\right)\right)\nonumber \\
 & = & \int d^{3}x\left(N\left(\mathcal{G}_{ijkl}\pi^{ij}\pi^{kl}-\sqrt{\gamma}^{\:(3)}R\right)+\beta_{i}\chi^{i}\right)\label{Hamiltonoperatorconstraint2-1}
\end{eqnarray}
In the quantum theory, this gave rise to a Hamiltonian constraint
that turned out to be 
\begin{equation}
\left(\mathcal{G}_{ijkl}\frac{\delta}{\delta\gamma_{ij}}\frac{\delta}{\delta\gamma_{kl}}+\sqrt{\gamma}^{\:(3)}R\right)\Psi(\gamma_{ij})=0\label{eq:wheelerdewitt2}
\end{equation}
and a diffeomorphism constraint 
\[
2iD_{j}\left(\frac{\delta}{\delta\gamma_{ij}}\Psi(\gamma_{ij})\right)=0
\]

For comparison with the Hamiltonian proposed by Dvali and Gomez, we
now try to put the Hamiltonian constraint of quantum gravity into
a form with creation and annihilation operators. To the knowledge
of the author of this note, there seems to be just one article \cite{McGuigan}
that has investigated this question. Usually, one defines, in analogy
to the Klein-Gordon equation, the following norm for the wave-functionals
of the Wheeler deWitt equation: 
\begin{equation}
\langle\Psi_{1}|\Psi_{2}\rangle=\int\prod_{x}d\Sigma^{ij}(x)\Psi_{1}^{*}\left(\mathcal{G}_{ijkl}\frac{\overrightarrow{\delta}}{i\delta\gamma_{kl}}-\frac{\overleftarrow{\delta}}{i\delta\gamma_{kl}}\mathcal{G}_{ijkl}\right)\Psi_{2}\label{eq:wdwnorm}
\end{equation}
where $d\Sigma^{ab}$ is the a surface element of the $6\times\infty^{3}$
dimensional space spanned up by the pseudo-metric $\mathcal{G}_{ijkl}$. 

McGuigan writes \cite{McGuigan} that in order to convert the Hamiltonian
into a form with creation and annihilation operators, one would have
to find a complete set of solutions $\Psi_{n}$ which are orthonormal
with respect to the above norm of eq. (\ref{eq:wdwnorm}). Then one
could make an ansatz 
\begin{equation}
\Psi(\gamma_{ij})=\sum_{k}\left(a_{k}^{ij}\Psi_{n}(\gamma_{ij})+a_{k}^{ij\dagger}\Psi_{n}^{*}(\gamma_{ij})\right)
\end{equation}
where $a_{k}^{ij\dagger}$ and $a_{k}^{ij}$are the creation and annihilation
operators for gravitons with $k$ momentum. However, McGuigan states 
\begin{quotation}
``The presence of the metric $\mathcal{G}_{ijkl}$as well as the
term $\sqrt{\gamma}^{\;(3)}R$ which are not quadratic in the $\gamma_{ij}$
or its derivatives will prevent us from finding such solutions here''
.
\end{quotation}
Please note that the situation here is entirely different with that
from zeta function renormalization from section 3. There, we had an
operator where an orthonormal base could be found. This was the case
because we made a perturbation expansion that we previously have cut
at second order, and then we were able to find an orthonormal base
for these second order terms. In contrast to a perturbation series
at one loop, the Hamiltonian incorporates the full non-perturbative
information of the theory. Unfortunately, for non perturbative quantum
gravity which incorporates the full nonlinearity of the theory, an
orthonormal base of quantum states can not be found.

One should also note that similar problems occur for a Hamiltonian
that describes spherical black holes. In the Wheeler deWitt equation
for spherical black holes from section 6.2 
\begin{equation}
\left(\frac{-\Lambda\delta^{2}}{2R^{2}\delta\Lambda^{2}}+\frac{1}{R}\frac{\delta^{2}}{\delta\Lambda\delta R}+\frac{RR''}{\Lambda}-\frac{RR'\Lambda'}{\Lambda^{2}}+\frac{R'^{2}}{2\Lambda}-\frac{\Lambda}{2}\right)\psi(\Lambda,R)=0
\end{equation}
there occur terms like $1/R^{2}$ or $1/\Lambda^{2}$which would prevent
finding a complete orthonormal set of solutions.

McGuigan then goes on noting that the situation would not be that
way in linearized gravity, where he investigates a topology $S^{1}\times S^{1}\times S^{1}$.
Expanding the metric as 
\begin{equation}
\gamma_{ij}(x)=\int d^{3}k\frac{1}{\left(2\pi\right)^{3}}\gamma_{ij}(k)e^{kx}
\end{equation}
and defining ``zero modes'' $\gamma_{ij0}=\gamma_{ij}(k=0)$, McGuigan
finds a Hamiltonian 
\begin{equation}
H=H_{0G}+H_{osc}
\end{equation}
where 
\begin{equation}
H_{0G}=\frac{1}{l^{2}}c\left(-\pi_{a}^{2}\frac{l^{4}}{2a}+\frac{1}{2a^{3}}24\gamma_{ik0}\gamma_{jl0}\pi_{0}^{ij}\pi_{0}^{kl}-a\frac{k}{2}\left(1-V(\gamma_{ij0})\right)\right)
\end{equation}
is a ``zero mode'' part. In $H_{0G}$, we have $c=g_{00}$, the
constant $a^{3}$ denotes the volume of the system, $l$ is the planck
length and $\pi_{0}^{ij}$ are the canonical momenta associated to
the zero modes $\gamma_{ij0}$. The function $V$ is, according to
McGuigan, a placeholder for ``the complicated dependence'' on the
zero modes that comes from the $^{3}R$ part of the Hamiltonian. 

Furthermore we have 
\begin{equation}
H_{osc}=ca^{3}\int\frac{d^{3}k}{(2\pi)^{3}}|k|\left(a_{k}^{ij\dagger}a_{k}^{ij}+\frac{1}{2}\right)
\end{equation}
which is the part of the Hamiltonian that describes gravitons with
annihilation and creation operators. 

When writing this, one should note that one can quantize linearized
gravity much more easily by methods that were first developed by Gupta
in 1928 \cite{Gupta}. Gupta started from a metric 
\begin{equation}
g_{\mu\nu}=\eta_{\mu\nu}+h_{\mu\nu}
\end{equation}
 where $|h_{\mu\nu}|<<1$ and the linearized Einstein equations for
weak fields: 
\begin{equation}
8\pi T_{\mu\nu}=\frac{1}{2}(\partial_{\rho}\partial_{\nu}h_{\mu}^{\sigma}+\partial_{\sigma}\partial_{\mu}h_{\nu}^{\sigma}-\partial_{\mu}\partial_{\nu}h-\square h_{\mu\nu}-\eta_{\mu\nu}\partial_{\rho}\partial_{\lambda}h^{\rho\lambda}+\eta_{\mu\nu}\square h
\end{equation}
For the linearized metric $h_{\mu\nu}$ Gupta then inserted an expansion
in terms of creation and annihilation operators for the gravitons.
The creation and annihilation operators fulfilled the usual commutator
relations for bosons. Historically, it was this investigation by Gupta
that showed gravitons to be spin 2 particles.

The author of this note wants to emphasize that McGuigan concludes
a description with creation and annihilation operators can not be
given for the Hamiltonian of the full non-linearized theory of gravitation.
This is in direct contrast to the Hamiltonian proposed by Dvali and
Gomez. But unfortunately, there are an additional problematic aspects: 

We noted in section 4.2 that deWitt has found severe factor ordering
inconsistencies that always should occur for for high energy states
with the Wheeler deWitt theory. This problem essentially forbids the
construction of a full quantum theory for the Wheeler deWitt equation,
which necessarily would have to incorporate states of high energy.
In section 4.1, we therefore were only able to solve the Wheeler deWitt
equation approximately by a WKB ansatz 
\[
\Psi=Ce^{iS_{0}}
\]
. However, state of the form $a^{\text{+}}|\psi\rangle$ that Dvali
and Gomez want to construct in their articles is in general not a
semi classical WKB state. Instead, expressions like $a^{\text{+}}|\psi\rangle$
suggest that one would not work in a semiclassical limit, but then
the Wheeler deWitt theory becomes inconsistent.

In Loop quantum gravity that we have reviewed in section 4.3, one
writes the Hamiltonian in Ashtekar variables, and then one defines
a state space with holonomies that are connected to points on a suitable
triangularization of the spacetime manifold. One does this in Loop
quantum gravity just in order to get a quantum theory of gravity that
is free of the factor ordering inconsistencies that the usual Wheeler
deWitt theory is plagued with. If one could consistently describe
a black hole in form of a condensate consisting of ordinary gravitons,
one would simply not need theories like loop quantum gravity or String
theory, where one introduces an entire set of additional assumptions
just in order to be able to define a consistent quantum theory.

The Hamiltonian proposed by Dvali 
\begin{equation}
H=-\hbar L_{0}\int d^{3}x\Psi(x)\nabla^{2}\Psi(x)-g\int d^{3}x\Psi(x)^{\dagger}\Psi(x)^{\dagger}\Psi(x)\Psi(x)\label{eq:dvalihamilt}
\end{equation}
seems to be devoid of all the problems given above that one faces
when one considers the usual Hamiltonian of non perturbative quantum
gravity, eq. (\ref{eq:wheelerdewitt2}), which also seems to look
completely different than Dvali's proposal. Furthermore, from the
usual Hamiltonian of eq. (\ref{Hamiltonoperatorconstraint2}), a connection
to the classical Einstein equations could be derived with the WKB
ansatz, as we saw in section 4.1. In Contrast, a rigorous derivation
of Einstein's equation, or a derivation of the Schwarzschild metric
is absent in the corresponding article of Dvali and Gomez. All this
suggests that the model of Dvali does not correspond to the usual
quantized Hamiltonian of general relativity and it does not seem to
describe linearized gravity either. 

Finally, I want to mention an additional problem of the proposal by
Dvali and Gomez. They claim that they are able to derive the black
hole entropy from their Bose Einstein condensate model. Dvali and
Gomez write on p. 14 of \cite{DvaliBlackholes} after a step by step
calculation where their Hamiltonian of eq. (\ref{eq:dvalihamilt})
was used as a starting point:
\begin{quotation}
Thus, we have reproduced the black hole evaporation law from the depletion
of the cold Bose-Einstein condensate at criticality.
\end{quotation}
However, during the standard calculation of the black hole entropy
from the path integral in section 7.1, it became apparent that the
black hole entropy always emerges from the proper inclusion of boundary
terms at the event horizon. The entropy is computed from the background
term
\[
e^{-I(g_{\mu\nu})}
\]
 in the Euclidean path integral 
\begin{equation}
Z_{eu}=e^{-I(g_{\mu\nu})}\int\mathcal{D}h_{\mu\nu}e^{-\underline{\underline{I}}(h_{\mu\nu})}
\end{equation}
with the Euclidean action 
\begin{equation}
I=-\int d^{4}x\sqrt{g}R-2\int_{\partial M}d^{3}x\sqrt{\gamma}(K-K^{0}),\label{eq:fulleuclideanaction-1-1}
\end{equation}
Since the Euclidean section of the Schwarzschild solution has $R=0$,
only the boundary term 
\begin{equation}
2\int_{\partial M}d^{3}x\sqrt{\gamma}(K-K^{0})
\end{equation}
 at the event horizon gives a non-vanishing result and contributes
to the gravitational entropy.

The standard Hamiltonian of quantum gravity was derived in section
4 by omitting such boundary terms. Therefore, the black hole entropy
could only be computed with the canonical formalism in section 6.2
after we included the necessary boundary terms into the action. Using
a description with canonical variables, the action with the boundary
terms included was given by eq. (\ref{eq:changedaction2}): 
\begin{eqnarray}
S_{total} & = & \int dt\int_{0}^{\infty}dr\left(P_{\Lambda}\dot{\Lambda}+P_{R}\dot{R}-N\mathcal{H}_{G}-\beta^{r}\mathcal{H}_{r}\right)\nonumber \\
 &  & +\int dt\frac{R_{0}^{2}}{2}\dot{\tau}-\int dtM\dot{\tau}_{+}
\end{eqnarray}
The boundary terms are given by $\int dt\frac{R_{0}^{2}}{2}\dot{\tau}-\int dtM\dot{\tau}_{+}$
and they contain the new degrees of freedom $\tau$ and $\tau_{+}$.
The latter have to be quantized with additional constraints and corresponding
canonical momenta $\pi_{0}$ and $\pi_{+}$. The classical version
of the constraints turned out to be, see eq. (\ref{eq:constrblackhole1}-\ref{eq:constrblackhole2}):
\begin{equation}
C_{0}=\pi_{0}-\frac{R_{0}^{2}}{2}=0
\end{equation}
 and 
\begin{equation}
C_{+}=\pi_{+}+M=0
\end{equation}
Upon quantization, the canonical momenta were converted into operators
$\pi=-i\frac{\delta}{\delta\tau_{0}}$ and $\pi_{+}=-i\frac{\delta}{\delta\tau_{+}}$
that act on a state functional 
\[
\psi=C(\Lambda,R)e^{iS_{0}(\Lambda,R.\tau_{0},\tau_{+})}
\]
of the Wheeler deWitt equation for the black hole. The quantum mechanical
constraints for the boundary terms become, see eq. (\ref{eq:quantumconstr1}-\ref{eq:quantumconstr2}):
\begin{equation}
\frac{\partial_{0}S_{0}}{\partial\tau_{0}}-\frac{R_{0}^{2}}{2}=0
\end{equation}
\begin{equation}
\frac{\partial_{0}S_{0}}{\partial\tau_{+}}+M=0
\end{equation}
The inclusion of the boundary terms changes the wave functional as
\begin{equation}
\psi=C(\Lambda,R)e^{S_{0}(\Lambda,R,\tau_{+},\tau_{0})}\rightarrow\psi=C(\Lambda,R)e^{S_{0}+\frac{R_{0}^{2}}{2}\tau_{0}-M\tau_{+}}
\end{equation}
In section 5 we have shown that the amplitude from the path integral
corresponds to a solution of the Wheeler deWitt equation. Hence, the
black hole entropy could finally be derived from the changed wave
functional in section 6.2 after the values of $\tau_{0}$ and $\tau_{+}$
were computed using the properties of the Euclideanized section of
the Schwarzschild solution. 

The fact that Hawking radiation emerges from the proper quantization
of boundary terms at the event horizon can also be seen by the calculation
of Gibbons and Hawking \cite{GibbonsHawking} who showed that a spherical
star without an event horizon has no gravitational entropy at all.

In the Hamiltonian of eq. (\ref{eq:dvalihamilt}) from Dvali and Gomez,
boundary terms that give raise to additional degrees of freedom at
the event horizon, which then change the wavefunctional of the black
hole, do not seem to be present. For this season, I question whether
one can, as Dvali and Gomez claim in \cite{DvaliBlackholes}, correctly
derive the black hole entropy from a Hamiltonian like the one they
propose. It seems to me that their Hamiltonian would, at best, resemble
something like the quantum Hamiltonian in eq (\ref{eq:wheelerdewitt2}),
which alone does not suffice to derive the black hole entropy, since
the boundary terms that give rise to a change of the wave functional
are missing. It is for all these reasons, why I find the proposals
by Dvali and Gomez on black holes to be highly problematic.

\end{document}